\documentclass[a4paper, 12pt]{article}

\usepackage{amsthm, amsmath, amssymb}
\usepackage{xcolor}
\usepackage{algorithm}
\usepackage{algpseudocode}
\usepackage{algcompatible}
\usepackage{graphicx}
\usepackage{subcaption}
\usepackage{float}
\usepackage{slashed}
\usepackage{booktabs} 
\usepackage[shortlabels]{enumitem}
\usepackage[utf8]{inputenc}

\usepackage[authoryear]{natbib}

\usepackage{siunitx} 
\usepackage{multirow}

\newcommand{\vecx}[1]{\MM{#1}}

\newtheorem{theorem}{Theorem}
\newtheorem{assumption}{A}
\newtheorem{definition}[theorem]{Definition}
\newtheorem{lemma}[theorem]{Lemma}
\newtheorem{prop}{Proposition}
 
\usepackage[left=2cm,right=2cm,top=2cm,bottom=2cm]{geometry}
\newtheorem{remark}{Remark}

\newtheorem*{problem}{Problem}

\newcommand{\ShowFigures}[2]{\def\tempvara{#1}\ifnum\tempvara=1 #2\fi}

\providecommand{\conditionname}{Condition}

\def\MODIFY#1{{ #1}}
\def\MODIFYrev#1{{{ #1}}}

\def\MM#1{\boldsymbol{#1}}
\newcommand{\pp}[2]{\frac{\partial #1}{\partial #2}} 

\newcommand{\dd}[2]{\frac{\diff#1}{\diff#2}}

\def\MM#1{\boldsymbol{#1}}

\DeclareMathOperator{\curl}{curl}
\DeclareMathOperator{\vecu}{\bf u}
\DeclareMathOperator{\vecv}{\bf v}
\DeclareMathOperator{\reals}{\mathbb{R}}
\DeclareMathOperator{\gradperp}{\nabla^{\perp}}

\DeclareMathOperator{\vortproj}{\mathcal{V}}

\DeclareMathOperator{\card}{card}

\DeclareMathOperator{\identity}{Id}

\DeclareMathOperator{\ess}{ess}
\def\rmse{\mathop{\mathrm{rmse}}}
\def\sprd{\mathop{\mathrm{sprd}}}

\def\pp{{\mathrm{p}}}
\def\ppi{{\mathrm{\pi}}}
\def\qq{{\mathrm{q}}}

\def\dd{{\mathsf d}}
\def\sigmay{{\mathcal{Y}}}

\def\expect{{\mathbb{E}}}

\def\statespace{{\mathbb{S}}}
\def\statespacepde{{\mathbb{S}_{\text{PDE}}}}
\def\statespacespde{{\mathbb{S}_{\text{SPDE}}}}
\def\statespaceobs{{\mathbb{S}_{\text{obs}}}}

\def\prob{\mathbb{P}}

\def\borel{\mathcal{B}}

\usepackage{authblk}

\begin{document}

\title{A Particle Filter for Stochastic Advection by Lie Transport (SALT):
A case study for the damped and forced incompressible 2D Euler equation\footnote{This work was partially supported by the EPSRC Standard Grant EP/N023781/1.}}

\author{Colin Cotter}
\author{Dan Crisan}
\author{Darryl D. Holm}
\author{Wei Pan\thanks{Corresponding author, email wei.pan@imperial.ac.uk}}
\author{Igor Shevchenko}
\affil[]{Department of Mathematics, Imperial College London, London, SW7 2AZ, UK}


\date{\today}

\maketitle

\begin{abstract}
  In this work, we \MODIFY{combine a stochastic model reduction with a particle filter augmented with tempering and jittering, and apply the combined algorithm} to a damped and forced incompressible 2D Euler dynamics defined on a simply connected bounded domain. 
  
  \MODIFY{We show that using the combined algorithm, we are able to assimilate data from a reference system state (the ``truth") modelled by a highly resolved numerical solution of the flow that has roughly $3.1\times10^6$ degrees of freedom, into a stochastic system having two orders of magnitude less degrees of freedom, which is able to approximate the true state reasonably accurately for $5$ large scale eddy turnover times, using modest computational hardware. }
 
 The model reduction is performed through the introduction of a stochastic advection by Lie transport (SALT) model as the signal on a coarser resolution. The SALT approach was introduced as a general theory using a geometric mechanics framework from Holm, Proc. Roy. Soc. A (2015). This work follows on the numerical implementation for SALT presented by Cotter et al, SIAM Multiscale Model. Sim. (2019) for the flow in consideration. The model reduction is substantial: The reduced SALT model has $4.9\times 10^4$ degrees of freedom.  

\MODIFY{Results from reliability tests on the assimilated system are also presented.}
\end{abstract}

\tableofcontents

\section{Introduction}

Data assimilation is the process by which observations (data) are integrated with mathematical models so that inference or prediction of the evolving state of the system can be made. 
For geoscience applications such as numerical weather prediction, it is an active area of research. There the typical global-scale state space dimension is of order $O(10^9)$, and observation data of dimension $O(10^7)$ are assimilated every $6$ -- $12$ hours. 
Current established methods used in operation centres include  4DVar, (various extended versions of) ensemble Kalman filter (EnKF) and variational assimilation methods. However, for fully nonlinear systems and complex observation operators these approaches are unsatisfactory. Our work presented in this paper is part of the wider effort to tackle high dimensional nonlinear geoscience problems using particle filters, as can be seen from the survey paper \citep{VanLeeuwen2019} and the references therein. 

The idea of modelling uncertainty using stochasticity in geophysical fluid applications is well established, see \citet{buizza1999stochastic, majda1999models, majda2001mathematical}. In this paper we work with the stochastic advection by Lie transport (SALT) approach, first formulated in \citet{Holm20140963}. It can be thought of as a framework for deriving stochastic partial differential equation (SPDE) models for geophysical fluid dynamics (GFD). The stochasticity is introduced into the advection part of the dynamics via a constrained variational principle called the Hamilton-Pontryagin principle. What results is a stochastic Euler-Poincar\'e equation, in which the local acceleration part of the transport operator is in the geometric form represented by the Lie derivative of velocity one-form in the direction of a stochastic vector field (in the form of a Stratonovich semimartingale). This approach for adding stochasticity into GFD models is different from the current state-of-the-art in numerical weather prediction (NWP), where stochastic models of uncertainty is introduced into the forcing; for example \emph{stochastic perturbation by physical tendencies} (SPPT) methodology, see \citet{palmer2018ecmwf}. 

 By adding stochasticity into the advection operator, one can model uncertain transport behaviour. In particular, the SALT stochastic term can be thought of as a model on the resolvable scale for the subgrid unresolvable fluid scales for transport. The main advantage of the SALT stochastic term is that it preserves the Kelvin's Circulation Theorem (KCT). However energy is not conserved by SALT because as one can show, an extra term called {line stretching} results from the application of the Reynolds transport theorem to the time differential of the energy; and the extra term contributes positively to the rate of change of the energy. An alternative, energy conserving stochastic approach called Location Uncertainty (LU) has also been developed \citet{Me2014}, but LU models do not preserve circulation.

A fundamental ingredient in SPDEs with SALT noise is the spectrum of the \emph{velocity-velocity correlation tensor}, \MODIFY{which is prescribed in the form of scaled eigenvectors that are denoted by $\vecx{\xi_i},\ i\in \mathbb{N}$. These appear in the following Stratonovich stochastic differential equation for the} Eulerian velocity field
\begin{equation}\label{eq:eulerian_svf}
\dd \vecx{x} = \vecv(\vecx{x},t)dt + \sum_{i} \vecx{\xi}_i (\vecx{x})\circ \dd W_t^i.
\end{equation} 
\MODIFY{Here $\vecx{x}$ denotes Eulerian position, $\vecv$ denotes large-scale mean velocity field, and $W^i$ are independent 1D Brownian motions. The $\circ$ symbol denotes Stratonovich stochastic integration.}
\citet{CoGaHo2017} showed that taking the diffusive limit of a flow with two timescales leads to a stochastic differential equation in the form of \eqref{eq:eulerian_svf}, where $\vecx{\xi}_i$ should be rigorously understood as empirical orthogonal functions corresponding to the different modes of the fast flow.

For applications, the vector fields $\vecx{\xi}_i$ need to be supplied \emph{a priori}. A \emph{data driven} calibration methodology for obtaining $\vecx{\xi}_i$ is described in \citet{cotter2018modelling,cotter2018numerically}, in which the authors numerically investigated two example fluid systems: a damped and forced 2D Euler model with no-penetration boundary condition, and a two-layer 2D quasi-geostrophic model prescribed on a channel. In those works, the SPDE model is interpreted as a parameterisation for the antecedent partial differential equation (PDE) model. Using statistical uncertainty quantification tests, it is shown that by conditioning on a suitable initial prior, an ensemble of SPDE solutions is able to effectively capture the large scale behaviour of the deterministic system on a coarser resolution. It is important to stress the fact that the deterministic system has $O(10^6)$ degrees of freedom, whilst the coarse stochastic system has $O(10^4)$ degrees of freedom. Capturing large scale dynamics on coarse scales enables the reduction of the high resolution PDE system to the coarse scale as an SPDE. This motivates further investigation of the performance of SALT SPDEs using ensemble data assimilation algorithms, where the forecast model is the SPDE prescribed on coarse scales. This is the theme of the present work, where we utilise the calibrated $\vecx{\xi}_i$ described in \citet{cotter2018numerically} in a data assimilation set-up for the damped and forced 2D Euler dynamics.

For us, sequential data assimilation is mathematically formulated as a nonlinear filtering problem, which can be tackled using a particle filter, see \citet{bain2009fundamentals,reich2015probabilistic}. 
A particle filter proceeds by alternating between forecast and analysis cycles. In each analysis cycle, observations of the current (and possibly past) state of a system are combined with the results from a prediction model (the forecast) to produce an analysis. The analysis step is typically performed either in the form of a ``best estimate" or in terms of approximating conditional distributions. The model is then advanced in time and its result becomes the forecast in the next analysis cycle.
However when applied to problems in high dimensions, without additional techniques a basic particle filter algorithm would almost certainly fail. This is due to the fact that in high dimensions the data are too informative.

\MODIFYrev{Beside the model reduction described above, in this paper, we describe two additional techniques:} tempering and jittering, which are incorporated into the basic bootstrap particle filter. The combined algorithm is applied to the damped and forced 2D Euler dynamics. These techniques are all necessary for the successful assimilation of data obtained from the true state of the system, which is modelled using a highly resolved numerical solution of $3.1\times 10^6$ degrees of freedom. 

\MODIFYrev{The theoretical justification of the tempering and jittering procedures is covered in \cite{beskosCrisanjasra}. Here, we offer some intuition to provide a better understanding of their advantages. The success of particle filters can be understood by the fact that they process data in a sequential manner. In other words, at each data assimilation time, only the new observations are assimilated, without the need to (re-)assimilate all past observations. This is advantageous as it is much harder (if not impossible) to assimilate the entire set of data available from initial time up to the current time. For high dimensional problems, a large amount of data (observations) become available at \emph{each} data assimilation step. This makes each individual data assimilation step (almost) as hard as processing data in a non-sequential manner. The tempering procedure alleviates this problem by incorporating the data gradually. It does so through several steps (in a sense mimicking the sequentiality of the particle filter as a whole). The number of steps is chosen adaptively so that at each step, only the "right amount" data is incorporated. Not too much as this would make the procedure inefficient, and not too little as this would introduce too many tempering steps. The "amount of data" is measured by the effective sample size statistic (ESS). The criteria used is to choose the so-called "temperature" parameter so that the ESS does not fall below a chosen threshold (we have used 80\% in our experiments). Through subsequent tempering steps, the entire set of data is eventually incorporated: At each tempering step the temperature is gradually increased from an initial value of 0 (this means no data is incorporated) to the final value of 1 (all data is incorporated). That is why, from a heuristic perspective, the tempering procedure succeeds when the standard particle filter would normally fail. }

	\MODIFYrev{Let us describe the intuition of the jittering step. Each tempering step is followed by a resampling procedure that eliminates particles with small weights, and multiplies particles with large weights.
	As a result, an "impoverishment" of the particle sample can occur. To counteract this, we use a jittering procedure that spreads the particles around. The procedure is chosen so that, asymptotically, no additional error or bias is added (in other words, should the particles be distributed according to the target posterior, they would remain distributed according to it). Outside the asymptotic regime, there is an inherent local error which controlled by the "jittering parameter". On one hand, the jittering parameter has to be chosen so that the local error remains small. On the other hand, it should be sufficiently large to introduce a reasonable spread in the sample in order to improve the quality of the sample. We explain the jittering procedure in section 3.3.}

The rest of the paper is structured as follows.
In section \ref{sec:models} we describe the damped and forced deterministic system and its SALT version. The deterministic system resolved on a fine resolution spatial grid is viewed as the simulated truth. The reduction to the SALT version is done via the variational approach formulated in \citet{Holm20140963}. The numerical calibration of the subgrid parameters $\vecx{\xi}_i$ and the numerical methodologies for solving the two systems are described in \citet{cotter2018numerically}.

In section \ref{sec:filtering} we formulate sequential data assimilation as a nonlinear filtering problem, in which the SALT equations are used as the signal. We describe in detail each  algorithm: bootstrap particle filter, tempering and jittering, which are all required to tackle the high dimensional nonlinear filtering problem.

In section \ref{sec:numericalresults} we present and discuss the numerical experiments and results. Two main sets of experiments are considered. In the first set, which we call the
\emph{perfect model scenario}, the true underlying state is a realisation of the signal. In the second set, which we call the \emph{imperfect model scenario}, data from the fine resolution true state is assimilated.
All experiments were run on a modest workstation which has two Intel Xeon processors totalling $32$ logical cores and $64$Gb memory. Additionally, an effective method for generating initial ensembles for SALT models is discussed.

Finally, section \ref{sec:conclusion} concludes the present work and discusses the outlook for future research.

The following is a summary of the main numerical experiments \MODIFY{and findings} contained in this paper:
\begin{itemize}
    \item \MODIFY{Using $100$ particles, we ran the particle filter over a period of $5$ eddy turnover times (ett, see \eqref{eq:ett} for definition) separately for 
    two different initial ensembles and
    observation dimensions $d_y = 289$ and  {$d_y=1089$}. {We chose an assimilation interval of $0.008$ eddy turnover times} for the imperfect model scenario, and $0.04$ eddy turnover times for the perfect model scenario. Here $1$ eddy turnover time is 1000 fine resolution time steps.
    The {root} mean square error ({rmse}), ensemble spread ($\sprd$), effective sample size (ess) and computational cost (measured in terms of number of equation evaluations) are shown, in figures \ref{figure: rmse and spread initial ensemble set 1}, \ref{figure: rmse and spread initial ensemble set 2} and \ref{figure:spde_statistics} for the perfect model scenario, and in figures \ref{figure: pde rmse and spread initial ensemble set 1}, \ref{figure: pde rmse and spread initial ensemble set 2} and \ref{figure:pde_statistics} for the imperfect model scenario.}
    
    \item \MODIFY{Forecast reliability is tested for our assimilated ensemble systems, by comparing forecast rmse with forecast ensemble spread, as well as 
    	rank histograms. The results are shown in figures \ref{figure: rmse vs spread initial ensemble set 1}, \ref{figure:spde_rankhistorgram_ux} and \ref{figure:spde_rankhistorgram_ux_msh8} for the perfect model scenario, and in figures \ref{figure: pde rmse vs spread initial ensemble set 1}, \ref{figure:pde_rankhistorgram_ux} and \ref{figure:pde_rankhistorgram_ux_msh8} for the imperfect model scenario.}
    
    \item \MODIFY{Eulerian trajectories at four individual grid points are shown for the truth, the truth plus realised observation errors, the posterior ensemble mean, the prior ensemble mean and individual ensemble members. These are shown in figures \ref{figure:spde trajectory and spread initial ensemble 1} and \ref{figure:spde trajectory spread initial ensemble set 2} for the perfect model scenario, and in figures \ref{figure:pde trajectory and spread initial ensemble 1} and \ref{figure:pde trajectory spread initial ensemble set 2} for the imperfect model scenario.}
    
    \item \MODIFY{In the perfect model scenario, using the tempering based particle filter, the numerical results show that, the data we assimilated, albeit low dimensional relative to the SPDE degree of freedom, gave sufficient information, to allow a reasonably accurate approximation of the true state. Further the posterior ensemble rmse is stable given the size of the observation noise and the data dimension. Reliability tests showed no features of bias, under-dispersion or over-dispersion. }
    
    \item \MODIFY{In the imperfect model scenario, the SALT base model reduction was combined with the the tempering based particle filter. The numerical results show that, despite the discrepancy between model and data, the data we assimilated, albeit very low dimensional relative to the PDE degree of freedom, still provides sufficient information to control the posterior ensemble. We judge that the posterior ensemble mean offers a reasonably accurately approximation of the signal, with rmse errors in the same order as those in the perfect model scenario. Further, the combined algorithm is sufficiently robust, despite slight features of bias and skew shown by the reliability tests on the assimilated ensemble system.}

\end{itemize}

\section{Deterministic and stochastic advection by Lie transport GFD models}\label{sec:models}

In this section, we describe the PDE and the SPDE models with Lie transport type stochastic terms. For the theory on SALT SPDEs see for example \citet{Holm20140963, crisan2017solution}. We follow \citet{cotter2018numerically} (also \citet{cotter2018modelling}) and use a data-driven approach to numerically model the $\boldsymbol{\xi}_i$'s. Thus information regarding the stochastic dynamics is complete except for initial and boundary conditions. Viewed as a parameterisation of the subgrid scales, numerically the SPDE shall be prescribed on a coarse resolution grid and the PDE prescribed on a fine resolution grid.

The spread of the SPDE dynamics from using $\boldsymbol{\xi}_i$ parameters calibrated with the data-driven approach described in \citet{cotter2018numerically} adequately captures the large scale features of the PDE dynamics. Those results indicate the feasibility of the calibrated SPDE as model reduction, thus providing the foundation for the present work where we utilise the SPDE as the \emph{signal} process
in a nonlinear filtering formulation. Nonlinear filtering will be the topic of discussion in section \ref{sec:filtering}.

In the following, the domain $D=[0,1]^2$ is assumed for both deterministic and stochastic models.


\subsection{Deterministic model}

We consider the vorticity version of an incompressible Euler flow with forcing and damping.
Let $\vecu:D\times\left[0,\infty\right)\rightarrow\mathbb{R}^{2},$ $\vecu\left(x,y,t\right)=\left(u_{1}\left(x,y,t\right),u_{2}\left(x,y,t\right)\right)$ denote the velocity field. Let $\omega=\hat{z}\cdot \curl{\bf u}$ denote the vorticity of $\vecu$, $\hat{z}$ denotes the $z$-axis. Note that for incompressible flows in two dimensions, $\omega$ is a scalar field. For a scalar field $g:D\rightarrow\mathbb{R},$ we write $\nabla^{\perp}g=\left(-\partial_{y}g,\partial_{x}g\right)=\hat{z}\times\nabla g.$ Let $\psi:{D}\times\left[0,\infty\right)\rightarrow\mathbb{R}$ denote the \emph{stream function}. The stream function is related to the fluid velocity and vorticity by ${\bf u}=\nabla^{\perp}\psi$ and $\omega=\Delta\psi$ respectively, where $\Delta = \partial_x^2+\partial_y^2$ is the Laplacian operator in $\mathbb{R}^2$. The existence of the stream function is guaranteed by the incompressibility assumption.

We now write down the deterministic model, 
\begin{eqnarray}
\partial_{t}\omega+\mathcal{L}_{\vecu}\omega & = & Q-r\omega\label{eq:2DEulerVorticity}\\
{\bf u} & = & \nabla^{\perp}\psi\label{eq:u_is_grad_perp_stream}\\
\Delta\psi & = & \omega\label{eq:laplace_stream_is_vorticity}.
\end{eqnarray}
We choose the forcing $Q$ to be given by 
\begin{equation}
Q\left(x,y\right)=\MODIFY{a}\sin\left(\MODIFY{b}\pi x\right),\qquad (x,y)\in D\label{eq:forcing}
\end{equation}
where \MODIFY{$a,$ $b$} and $r$ are constants having the following roles: \MODIFY{$a \geq 0$} controls the strength of the forcing; \MODIFY{$b$} is an integer interpreted as the number of gyres in the external forcing; and $r>0$ can be seen as the damping rate.  $\mathcal{L}_{\vecu} \omega$ denotes the Lie derivative of $\omega$ with respect to the vector field $\vecu$. When applied to scalar fields, $\mathcal{L}_{\vecu}$ is simply the directional derivative with respect to $\vecu$, see \citet{chernchenlam}
\[
\mathcal{L}_{\vecu}  = \vecu\cdot\nabla .
\]
We consider a no-penetration spatial boundary condition 
\begin{equation}
\left.\psi\right|_{\partial\mathcal{D}}=0\label{eq:boundary_condition_streamfunction}
\end{equation} to close the system.
This system is a special case of a nonlinear, one-layer quasi-geostrophic (QG) model that is driven by winds above.

\subsection{Stochastic model}

Consider the space $\Omega = C_0([0, \infty), \reals^m)$ of continuous function whose value at $0$ is zero. It is equipped with the  Wiener measure $\prob$ and its natural filtration $\{ \mathcal{F}_* \}$.  Let $\{W_t: t \in [0, \infty) \}$ be the canonical Brownian motion on $\reals^m$, that is for $\gamma\in \Omega$, $W_t(\gamma)=\gamma(t)$ is the evaluation map. We write $W_{t}^{i}$ to denote the $i$'th component of $W_t$. The SALT version of the Euler fluid equation (\ref{eq:2DEulerVorticity}) as derived in \citet{Holm20140963}. \citet{cotter2018numerically} introduced damping and forcing to facilitate statistical equilibrium in the underlying resolved system, leading to the following stochastic partial differential equation (SPDE),
\begin{equation}
\dd q+\mathcal{L}_{\vecv}q\ dt+\sum_{i=1}^{m}\mathcal{L}_{\boldsymbol{\xi}_{i}}q\circ \dd W_{t}^{i}=\left(Q-rq\right)dt\label{eq:stochastic2dEuler}
\end{equation}
where the vector fields $\boldsymbol{\xi}_{i}$ represent scaled eigenvectors of the velocity-velocity correlation tensor $C_{ij}=\boldsymbol{\xi}_{i}\boldsymbol{\xi}_{j}^{T}$.

Equation (\ref{eq:stochastic2dEuler}) arises from a time-scale separation assumption for the deterministic Eulerian transport velocity $\vecu$, leading to the following Stratonovich stochastic differential equation
\begin{equation}\label{eq:StochVF}
\dd \tilde{\vecx{x}}_t (\vecx{x}) = 
\vecv(\vecx{x},t) dt + \sum_{i=1}^m\boldsymbol{\xi}_{i}(\vecx{x})\circ \dd W_{t}^{i}
\end{equation}
where $\vecv$ and $\vecx{\xi}_i$ are divergence free vector fields, from which \eqref{eq:stochastic2dEuler} may be derived.
Here one can intuitively think of $\vecv$ as the ``large'' scale mean part of $\vecu$. In this present work since we are interested in the practicality of \eqref{eq:stochastic2dEuler} for data assimilation, we follow \citet{cotter2018numerically} and make the approximation that the sum in \eqref{eq:StochVF} is finite. Hence the stochastic term in \eqref{eq:stochastic2dEuler} also consists of $m$ terms.

Let $\tilde{\psi}$ denote the stream function of $\vecv$, and let $\zeta_i$ denote the stream function of $\vecx{\xi}_i$, i.e.
\[
\boldsymbol{\xi}_{i}=\nabla^{\perp}\zeta_{i}.
\]
Note that $\zeta_i$ is constant in time. The $\zeta_i$ can be solved for and stored on the computer after the $\vecx{\xi}_i$ are obtained. For this the boundary condition
\begin{equation}
    \quad \left.\zeta_i\right|_{\partial\mathcal{D}}=0
\end{equation}
is enforced for each $i=1,\dots,m.$
Then \eqref{eq:StochVF} can be expressed in terms of $\tilde{\psi}$ and $\zeta_i$,
\begin{equation}
\dd \tilde{\vecx{x}} =\gradperp\big(\tilde{\psi} dt+\sum_{i=1}^{m}\zeta_{i}\circ \dd W_{t}^{i}\big).\label{eq:stochastic_transport_velocity_streamfunction_form}
\end{equation}
Expressing the transport velocity in this form is useful because it allows us to introduce stochastic perturbation (i.e. terms with $\circ \ \dd W_{t}^{i}$) via the stream function when solving the SPDE system numerically, thereby keeping the discretisation of (\ref{eq:stochastic2dEuler}) the same as the deterministic equation (\ref{eq:2DEulerVorticity}), see \citet{cotter2018numerically}. 

The full set of stochastic equations is
\begin{eqnarray}
\dd q+\mathcal{L}_{\vecv}qdt+\sum_{i=1}^{m}\mathcal{L}_{\boldsymbol{\xi}_{i}}q\circ \dd W_{t}^{i}&=&\left(Q-rq\right)dt \label{eq:spdeeuler}\\
\vecv & = & \nabla^{\perp}\tilde{\psi} \\
\Delta\tilde{\psi} & = & q 
\end{eqnarray}
with boundary condition
\begin{equation}\label{eq:spdeboundary}
\tilde{\psi}|_{\partial\mathcal{D}}=0.
\end{equation}
The forcing term is the same as the deterministic case.
\begin{remark}
        The It\^o form of (\ref{eq:stochastic2dEuler}) is obtained from an application of the identity
        \begin{equation}
        \int_{0}^{t}\mathcal{L}_{\boldsymbol{\xi}_{i}}q\left(s\right)\circ \dd W_{s}^{i}=\int_{0}^{t}\mathcal{L}_{\boldsymbol{\xi}_{i}}q\left(s\right)\dd W_{s}^{i}+\frac{1}{2}\left\langle \mathcal{L}_{\boldsymbol{\xi}_{i}}q,W^{i}\right\rangle _{t}\label{eq:ito_stochastic2deuler}
        \end{equation}
        where $\left\langle .,.\right\rangle _{t}$ is the cross-variation bracket and 
        \begin{align*}
        \left\langle \mathcal{L}_{\boldsymbol{\xi}_{i}}q,W^{i}\right\rangle _{t} & =\mathcal{L}_{\boldsymbol{\xi}_{i}}\left\langle q,W^{i}\right\rangle _{t}\\
        & =\mathcal{L}_{\boldsymbol{\xi}_{i}}\left\langle \int\{(Q-rq)dt-\mathcal{L}_{\vecv}qdt-\sum_{j=1}^{\infty}\mathcal{L}_{\boldsymbol{\xi}_{j}}q\circ \dd W_{t}^{j}\},W^{i}\right\rangle _{t}\\
        & =\mathcal{L}_{\boldsymbol{\xi}_{i}}\left\langle -\int_{0}^{.}\mathcal{L}_{\boldsymbol{\xi}_{i}}q\circ \dd W_{s}^{i},W^{i}\right\rangle _{t}\\
        & =\mathcal{L}_{\boldsymbol{\xi}_{i}}\left(-\int_{0}^{t}\mathcal{L}_{\boldsymbol{\xi}_{i}}q\left(s\right)ds\right)
         =-\int_{0}^{t}\mathcal{L}_{\boldsymbol{\xi}_{i}}^{2}q(s)ds
        \end{align*}
        Hence
        \[
        \int_{0}^{t}\mathcal{L}_{\boldsymbol{\xi}_{i}}q\left(s\right)\circ \dd W_{s}^{i}=\int_{0}^{t}\mathcal{L}_{\boldsymbol{\xi}_{i}}q\left(s\right)\dd W_{s}^{i}-\frac{1}{2}\int_{0}^{t}\mathcal{L}_{\boldsymbol{\xi}_{i}}^{2}q(s)ds
        \]
        and (\ref{eq:ito_stochastic2deuler}) is thus
        \begin{equation}
        \dd q+\mathcal{L}_{\vecv}qdt+\sum_{i=1}^{m}\mathcal{L}_{\boldsymbol{\xi}_{i}}q\ \dd W_{t}^{i}=\frac{1}{2}\sum_{i=1}^{m}\mathcal{L}_{\boldsymbol{\xi}_{i}}^{2}q\ dt+(Q-rq)dt\label{eq:ito_spde_euler}
        \end{equation}
        where $\mathcal{L}_{\boldsymbol{\xi}_{i}}^{2}q=\mathcal{L}_{\boldsymbol{\xi}_{i}}\left(\mathcal{L}_{\boldsymbol{\xi}_{i}}q\right)=\left[\boldsymbol{\xi}_{i},\left[\boldsymbol{\xi}_{i},q\right]\right]$
        is the double Lie derivative of $q$ with respect to the divergence free vector field $\boldsymbol{\xi}_{i}$.
\end{remark}
For the damped and forced stochastic system considered in this section, on the torus $\mathbb{T}^2$ a global wellposedness theorem with the solution space $W^{2,2}\left(\mathbb{T}^2\right)$ is proved in \citet{crisan2019well}. In a forthcoming sequel to this work we also show the the wellposedness of the solution on the bounded domain $D$ with no-penetration boundary conditions. We make the following important assumption.
\begin{assumption}\label{ass:spdewellposedness}
The stochastic system \eqref{eq:spdeeuler} -- \eqref{eq:spdeboundary} is wellposed in some solution space denoted by $\statespacespde$ in the sense that a unique global in time, pathwise distributional solution\footnote{We refer to \citet{crisan2019well} for the precise definition of these terms.} exists. 
\end{assumption}
Under assumption \ref{ass:spdewellposedness}, it is useful to introduce the following. Let $G: \statespacespde \times \Omega \rightarrow \statespacespde$ denote the It\^o solution map of the stochastic system \eqref{eq:spdeeuler} -- \eqref{eq:spdeboundary} so that
\begin{equation}\label{eq:spdesolutionop}
	G(q_0, \gamma)_t = q_t, \quad q_0 \in \statespacespde.
\end{equation}
The solution map $G$ is used in the next section where probability measures on the solution space are defined as the push-forward of $\prob$ using $G$.

\section{Nonlinear Filtering }\label{sec:filtering}
        
In this section, we formulate data assimilation as a nonlinear filtering problem in which the aim is to utilise observed data to correct the distribution of predictive dynamics. We describe a particle filter methodology which incorporates three additional techniques that are required to effectively tackle this {high dimensional} data assimilation problem.

In nonlinear filtering terminology the predictive dynamics is often called the \emph{signal }\footnote{ Also known as the forecast model in statistics and meteorology literature, \citep[see][]{reich2015probabilistic}.}. The signal in our setting corresponds to the SALT SPDE. Data is obtained via an \emph{observation process} which represents noisy partial measurements of the underlying true system state. The goal is to determine the posterior distribution $\pi_t$ of the signal at time $t$ given the information accumulated from observations. This is known as the \emph{filtering problem}. This is different to \emph{inversion} problems (also called smoothing problems),  where one is interested in obtaining the posterior distribution of the system's initial condition, see for example \citet{StuartAM2010IpAB}. 

The stochastic filtering framework enables us not just to provide a solution to the data assimilation problem, but also offer a clear language in which to explain the details and the intricacies of the problem. We detail below an elementary introduction to the filtering problem. 

Let $\statespace$ denote a given state space, and let $\mathcal{P}(\statespace)$ denote the set of probability measures on the state space. In what follows the state space will be $\statespace=\mathbb R^{d_x}$, where $d_x$ is the dimension of the space. To avoid technical complications we will assume in the following that time runs discretely $t=0,1,\ldots$ .
We shall work in a Bayesian setting, in other words we will assume that we know the distribution of the signal $X_t$ for $t=0,1,\ldots$, which will be denoted by $p_t$ for $t=0,1,\ldots$. We also assume that partial observations, denoted by $Y_t$, of dimension $\statespaceobs = \reals^{d_y}$ with $d_y \leq d_x$ are available to us at times $t=0,1,\ldots$ and we wish to approximate the signal  $X_t$ given the accumulated observations $Y_1, \ldots , Y_t$. Of course we could aim to approximate $X_t$ using an arbitrary 
$\mathcal{Y}_t$-adapted estimator $\bar{X}_t$, where $\mathcal{Y}_t$ is the $\sigma$-algebra $\mathcal{Y}_t=\sigma(Y_1,\ldots, Y_t)$.  However, the \emph{best} estimator $\hat X_t$ is the conditional expectation of $X_t$ given $\mathcal{Y}_t$, $\hat X _t=\mathbb E[X_t|\mathcal Y _ t]$. In this context, by the best estimator, we mean the minimiser of the mean square error $MSE(\bar X _t)=  \mathbb E[\|X_t-\bar X_t\|^2|\mathcal Y _ t]$, where $\|\cdot\|$ is the standard Euclidian norm on $\mathbb R^{d_x}$. Of course we would not just want to compute/estimate $\hat X _t=\mathbb E[X_t|\mathcal Y _ t]$, but also the error that we would make if we approximate $X_t$ with $\hat X_t$, i.e., for $t=0,1,\ldots$. 
\[
\expect[\|X_t-\hat X_t\|^2|\mathcal Y _ t]=\expect[\|X_t\|^2|\mathcal Y _ t]-\expect[\|X_t\||\mathcal Y _ t]^2.
\]
The quantiles of the approximation error will also be of interest. Therefore, in general, the filtering problem consists in determining the condition distribution of the signal given given $\mathcal{Y}_t$ denoted by $\pi_t$. Once $\pi_t$ is determined, then its first moment (the mean vector) will give us $\hat X _t$, its covariance matrix can be used to compute the mean square error $MSE(\hat X _t)$, etc.
So one can adopt one of two different approaches of estimation the signal given partial observations.

\begin{itemize}

\item Develop a data assimilation algorithm that results in a a point approximation $\bar X_t$ of the signal using the data $Y_1, \ldots , Y_t$. The approximation may or may not be optimal and only, on rare occasions, an estimate of the error $X_t-\bar X_t$ will be available. 

\item Develop a data assimilation algorithm that results in an approximation of $\pi_t$ the conditional distribution of the signal $X_t$ using the data $Y_1, \ldots , Y_t$. This in turn will offer an approximation of the optimal estimator $\hat X_t$ as well as the approximation of the error, quantiles, occupation mesures, etc.   

\end{itemize}

Of course, algorithmically we expect the first problem to be a lot easier than the second. The computation, of an estimator $\bar X_t$ that is an element of $\mathbb R^{d_x}$ would be expected to be a lot easier that that of a probability measure over $\mathbb R^{d_x}.$ The first one is a finite dimensional object the latter is an infinite dimensional one. However, in the exceptional case when the signal is a linear time-series and the observation has linear dependence on the signal and they are driven by Gaussian noise the two approaches more or less coincide. The reason is that, in this case $\pi_t$ is Guassian and one can explicitly write the recurrence formula for the pair $(\hat X _t, P_t)$, where $P_t$ is the covariance matrix of $\pi_t$. So on one hand one can compute directly the optimal estimator $\hat X_t$ and on the other hand the Gaussianity ensure that $\pi_t$ is fully described by 
$(\hat X_t, P_t)$. This is the so-called Kalman-Filter. There are numerous extensions of this method to non-linear filter that attempt a similar methodology for the non-Gaussian conditional distribution. Such approaches are not optimal in the sense that they don't offer a point estimator that is the optimal one and the corresponding ``covariance" matrix that is produced is not the covariance matrix of $\pi_t$. The existing literature in this direction is vast, we cite here \citep{reich2015probabilistic, evensen2009data, ljung1979asymptotic}.

Particle filters are a class of numerical methods that can be used to implement the second approach. They have been highly successful for problems in which the dimension of the state space $d_x$ has been low to medium. However, in recent works \citep{kantas2014sequential, beskos2017stable, beskos2014error} they have been shown to also work in high dimensions $d_x$. In this paper, we tackle a state space with dimension of order $O(10^6)$. For a filtering perspective, we overcome here one other hurdle as we explain below.

Let us denote by $p_t\in \mathcal P (\underset{\left( t+1\right) \mathrm{-times}}
{\underbrace{\mathbb R^{d_x}\times\ldots\times\mathbb R^{d_x}}})$, $t=0,1,\ldots$ the (prior) distribution of the signal on the path space $(X_0,X_1,\ldots, X_t)$. The prior distribution of the signal $p_t$ and the observations $Y_s$, $s=0,1,\ldots,t$ are the building blocks of $\pi_t$, $t=0,1,\ldots$. To be more precise, one can show that there exists a mapping     
\begin{equation}\label{eq:filtering_density}
(p,y_{0},\ldots y_t)\mapsto\Xi(p,y_{0},\ldots y_t):\mathcal P (\underset{\left( t+1\right) \mathrm{-times}}
{\underbrace{\mathbb R^{d_x}\times\ldots\times\mathbb R^{d_x}}})\times \underset{\left( t+1\right) \mathrm{-times}}
{\underbrace{\mathbb R^{d_y}\times\ldots\times\mathbb R^{d_y}}} \mapsto \mathcal P (\mathbb R ^{d_x}),
\end{equation}
such that $\pi_t=\Xi(p_t,Y_{0},\ldots Y_t)$. Under very general \MODIFY{conditions (for example, it is enough to assume that the likelihood functions are continuous in the $y$ variable and apply Lemma 2.4 from \citep{crisan2018stable})} on the signal and the observation, this mapping is jointly continuous on the product space $\mathcal P (\underset{\left( t+1\right) \mathrm{-times}}
{\underbrace{\mathbb R^{d_x}\times\ldots\times\mathbb R^{d_x}}})\times \underset{\left( t+1\right) \mathrm{-times}}
{\underbrace{\mathbb R^{d_y}\times\ldots\times\mathbb R^{d_y}}}.$ 
This would mean that $\pi_t$ will give a reasonable approximation of the conditional distribution of the signal as long as the distribution $(X_0,X_1,\ldots, X_t)$ does not differ significantly from the one used to construct $\pi_t$. The same will happen when the true law of the observation does not differ significantly from the chosen model. This property of the posterior distribution is crucial, see \MODIFYrev{\emph{imperfect model} in section \ref{sec:numericalresults} for details}.

In the rest of this section we consider only the space-time discretised SPDE signal, of spatial dimension $d_x$. The observation process is given by noisy spatial evaluations of an underlying true system state at discrete time steps. We consider two scenarios for the underlying true system state, henceforth called the \emph{truth}.

In the first scenario, we aim to compute the conditional distribution of the signal given partial observations of a single realised trajectory of the SPDE system. In this case the predictive dynamics and the truth are from the same dynamical system.
We call this the \emph{perfect model scenario} (or twin experiment, see \citet{reich2015probabilistic}). 

In the second scenario, we use instead noisy spatial evaluations of a space-time discretised solution corresponding to the PDE system \eqref{eq:2DEulerVorticity} -- \eqref{eq:boundary_condition_streamfunction}. We call this the \emph{imperfect model scenario}. The truth in this case is computed on a more refined grid than solutions of the SPDE. Nevertheless the solution of the SPDE converges to that of the PDE as the coarser grid converge, see \citet{cotter2018numerically}. Similarly the corresponding observations will converge (provided the observation noise does not change).  This ensures the  successful assimilation of PDE data into the SPDE model, assurance from the uncertainty quantification tests shown in \citet{cotter2018numerically} is necessary to numerically guarantee that the mis-match between state spaces remains small.

To our knowledge, this is the first application of particle filters to the case where the signal is described by a SALT SPDE system. As we explain below a straight application of the classical bootstrap particle filter algorithm fails. To succeed we implement and incorporate the following procedures.
\begin{itemize}
  \item Model reduction -- approximate a high dimensional system using a low dimensional system via stochastic modelling, the result of which can be further reduced by choosing a projection of the noise process onto a submanifold. This was accomplished in \citet{cotter2018numerically}.

  \item Tempering -- compute a sequence of intermediate measures $\pi^k_t$ parameterised by a finite number of temperatures that control the smoothness of the density of $\pi^k_t$. This procedure eases the problem of highly singular posteriors in high dimensions, which come from the fact that high dimensional observations are too informative. 

  \item Jittering -- a Markov chain Monte Carlo (MCMC) based technique for recovering lost population diversity in particle filter algorithms.
\end{itemize}
These techniques are added to the basic bootstrap particle filter, and are demonstratively necessary, theoretically consistent and rigorously justified.
In addition, we shall pay particular attention to the initialisation of the particle filter, though this is discussed in section \ref{sec:initialdist}. 

Before proceeding to the problem formulation, we insert an important remark.
\begin{remark}\label{rem:ignore_functionspace}
    Our spatial discretisations for the PDE and SPDE fields are defined on appropriate finite element spaces, see \citet{cotter2018numerically} for details of the numerical methods we use for the models under consideration. Under assumption \ref{ass:spdewellposedness}, it is important to understand that instead of the finite state space $\statespace = \reals^{d_x}$, the actual problem involves measures defined on infinite dimensional function spaces, in particular it is highly plausible that in theory the state space for the SPDE is Sobolev $W^{k,2}(D)$ for $k\geq 2$. Discussions of these technical complications are not the focus of this work. And since in practise we work with numerical solutions anyhow, we setup our filtering problem in a finite dimensional setting. However, the methods we use are all theoretically consistent in the limit, see \citet{StuartAM2010IpAB,dashti2017bayesian}.
\end{remark}
In light of remark \ref{rem:ignore_functionspace}, henceforth we drop the word ``discretised" when describing the state space, signal and observation processes.

\subsection{Filtering problem formulation}\label{secfilterproblem}

    Consider discrete times $\Lambda = \lbrace t_0, t_1, \dots, t_n,\dots\rbrace$. Let $X:\Lambda\times\Omega \rightarrow \statespace$ be a discrete  time Markov process called the \emph{signal}.
    Let $Y:\Lambda\rightarrow \statespaceobs$ be a discrete time process called the \emph{observation} process. We assume $Y(t_0)=0$ almost surely (a.s.). We consider \emph{Eulerian data assimilation} where the observations correspond to fixed spatial points $Y_t = (Y_t^1,\dots, Y_t^{d_y}) \in D$, for all $t\in\Lambda.$ As already mentioned in the \MODIFY{ introduction}, we denote the dimensions of $\statespace$ and $\statespaceobs$ by $d_x$ and $d_y$ respectively.
        
    We take $X$ and $Y$ to correspond to the \emph{velocity} {vector field}.  Mathematically we could also consider the vorticity field or the stream function, but in real world scenarios those fields may be difficult to observe directly.
    We denote by $X_{i:j}$ and $Y_{i:j}$ the path of the signal and of the observation process from time $t_i$ to time $t_j$,
    \[
    X_{i:j} = (X_{t_i},X_{t_{i+1}},\dots ,X_{t_j}),\quad Y_{i:j}=(Y_{t_i},Y_{t_{i+1}},\dots ,Y_{t_j}).
    \]
    Let $x_{i:j}$ and $y_{i:j}$ denote particular trajectories of $X_{i:j}$ and $Y_{i:j}$. For notational convenience, we may write in the subscripts $i$ to mean $t_i$. 
    
    It is useful to introduce the following standard notation in the case when $\mu$ is a measure and $f$ is a measurable function, and $K$ is a Markov kernel
    \[
    \mu f \triangleq \int f d\mu,\ \ \ 
    \mu K(A)\triangleq \int K(x, A)\mu(dx),
    \ \ \ Kf(x)\triangleq \int f(z) K(x, dz).
    \] 
    The marginal distribution of the signal changes according to 
    \begin{equation}
    \prob( \left.X_t \in A \ \right| X_{t-1}=x_{t-1} )=\int_{A} k_t(x_{t-1}, dx_t)\label{eq:signalmarginaldistribution}
    \end{equation}
    for $A \in \borel(\statespace)$, and $k_t$ is a probability transition kernel defined by the push-forward of $\prob$ using the (discretised) SPDE solution map $G$ from assumption \ref{ass:spdewellposedness}. 
    
    In standard filtering theory the observation process is defined by
    \begin{equation}\label{eq:obs}
            Y_t = h(X_t) +\epsilon_t ,\quad t\in \Lambda
    \end{equation}
    where $h:\statespace \rightarrow \statespaceobs$ is a Borel-measurable function, and for $t\in \Lambda$, $\epsilon_t:\Omega \rightarrow \statespaceobs$ are mutually independent Gaussian distributed random vectors with mean zero and covariance matrix $\gamma$. Thus
    
    \begin{equation}
    \prob( \left.Y_t \in B \right| X_{t}=x_{t})=\int_{B} g_t(y_t - h(x_{t}))dy_t \label{eq:likelihood}
    \end{equation}
    for $B \in \borel({\statespaceobs})$ and Gaussian density $g_t$. For convenience, define $g_t^{y_t}(x)\triangleq g_t(y_t - h(x_t))$ which is commonly referred to as the \emph{likelihood} function. 
    
    We can now define the filtering problem.
    \begin{problem}[Filtering Problem] \label{prob:filtering_problem}
        For $t \in \Lambda$, we wish to determine the conditional distribution of the signal given the information accumulated from observations, i.e.
        \begin{equation}\label{eq:postcondexpspde}
        \pi_t \varphi \triangleq \mathbb{E}\left[ \left. \varphi(X_t) \right| \sigmay_t \right],\quad \sigmay_t = \sigma(Y_{0:t})
        \end{equation}
        for all bounded measurable functions $\varphi \in B\left(\statespace\right)$, with $\pi_0$ being the given initial probability distribution on the state space $(\statespace,\borel(\statespace))$. In particular when $\varphi = \boldsymbol{1}_{A}$ for $A \in \borel(\statespace)$ we have $\pi_t \boldsymbol{1}_A =\pi_t(A)= \prob(X_t\in A \,| \sigmay_t)$. 
    \end{problem}
    In statistics and engineering literature, $\pi_t$ is often called the Bayesian posterior distribution.  
    Note that $\pi_t$ is a random probability measure. For arbitrary $y_{0:t}$, denote
    \[
    \pi_t^{y_{0:t}} \varphi \triangleq \expect\left[ \left. \varphi(X_t) \right| Y_{0:t}=y_{0:t} \right],\quad
    \pi_t^{y_{0:t}}(A) = \prob(X_t\in A \,| Y_{0:t}=y_{0:t}).
    \]
    We also introduce predicted conditional probability measures $p_t$ and $p_t^{y_{0:t}}$ defined by
    \[
    p_t^{y_{0:t-1}} \varphi \triangleq \expect\left[ \left. \varphi(X_t) \right| Y_{0:t-1}=y_{0:t-1} \right],\quad
    p_t^{y_{0:t-1}}(A) = \prob(X_t\in A \,| Y_{0:t-1}=y_{0:t-1}).
    \]    
    We have $\prob$-almost surely the following \emph{Bayes recurrence relation}, see \citet{bain2009fundamentals}. For $t\in\Lambda$ and $A \in \borel(\statespace)$,
    \begin{align}
           & p_{t}(A) \triangleq \pi_{t-1}k_t(A) = 
           \int k_t(x_{t-1},A) \pi_{t-1}(dx_{t-1})  && \text{prediction} \label{eq:prediction}\\
           & \pi_{t}(A) = C_t^{-1} p_t \, g_t^{Y_t}(A) =  C_t^{-1}\int_{A} g_t^{Y_{t}}(x_t) p_t(dx_{t})  && \text{update} \label{eq:updating}
    \end{align}
    where 
    \[
    C_t \triangleq p_t \, g_t^{Y_t} = \int_{\statespace} g_t^{Y_{t}}(x_t) p_t(dx_{t})
    \] is a normalising constant. Due to \eqref{eq:updating}, we may also write $\frac{d\pi_t}{dp_t} \propto g^{Y_t}_{t}$, thus $\pi_t = p_t \frac{d\pi_t}{dp_t}$.
    
        In the general case for any bounded measurable function $\varphi \in \borel( \statespace)$, we have for problem \ref{prob:filtering_problem} the recurrence relation
    \begin{align}
                & p_t\, \varphi = \pi_{t-1}\, k_t \, \varphi && \text{prediction} \\
                & \pi_t \, \varphi = C_t^{-1}p_t \, g_t^{Y_t}\,\varphi && \text{update}
    \end{align}
        Except for a few rare examples of the signal, it is extremely difficult to directly evaluate $\pi_t$ because there are no ``simple'' expressions. In section \ref{subsec:smc} we shall describe the particle filter methodology that we employ to tackle the filtering problem. Note that in statistics and engineering literature, particle filters are often called sequential monte carlo (SMC) methods.

\subsection{Particle filter}\label{subsec:smc}
 Particle filter methods are among the most successful and versatile methods for numerically tackling the filtering problem. A basic algorithm implements the Bayes recurrence relation by approximating  the measure valued processes $\pi_t$ and $p_t$ by $N$-particle empirical distributions. The position of each particle is updated using the signal's transition kernel. At the same time, individual weights are kept up-to-date in accordance with the updated particle positions. It is in the weights updating step that we take into account the information provided by the observations: particles are reweighted using the likelihood function. A new set of particle positions can be sampled based on the updated weights and the procedure iterates. 
 
 Due to the high dimensional nature of the systems in consideration, additional techniques are necessary in order to make the basic algorithm work effectively. We provide a concise presentation of the algorithms employed, and note that these methods are all mathematically rigorous. For more thorough discussions we refer the reader to \citet{bain2009fundamentals,reich2015probabilistic,dashti2017bayesian,kantas2014sequential,beskos2014error}.
 
 \subsubsection{Bootstrap particle filter}
 The basic algorithm, called the \emph{bootstrap particle filter} or the \emph{sampling importance resampling} (SIR) algorithm, proceeds in accordance with the Bayes recurrence relation \eqref{eq:prediction} -- \eqref{eq:updating} by repeating \emph{prediction} and \emph{update} steps. To define the method, we write an $N$-particle empirical approximation of $\pi_i$. \MODIFY{Thus we have}
 \begin{equation}\label{eq:smc}
  \pi_i \approx \pi_i^N \triangleq \frac{1}{ \sum_{m=1}^{N}w_i^{(m)}} \sum_{n=1}^{N} w_i^{(n)} \delta(x_i^{(n)}) = \sum_{n=1}^{N} \bar{w}_i^{(n)} \delta(x_i^{(n)})
 \end{equation}
 where $\delta$ denotes Dirac measure. The discrete measure $\pi_i^N$ is completely determined by particle positions $x_i^{(n)}\in \statespace$ and weights $w_i^{(n)}\in \reals$, $n=1,\dots,N$.
 We define the update rule \[\{x_i^{(n)},w_i^{(n)}\}_{n=1}^{N} \rightarrow \{x_{i+1}^{(n)},w_{i}^{(n)}\}_{n=1}^{N},\] for advancing $\pi_i^N$ to $p_{i+1}^N$ to be given by the numerical implementation of the SPDE solution map G, see \eqref{eq:spdesolutionop},
  \begin{equation}
  x_{i+1}^{(n)} = G(x_i^{(n)},\omega^{(n)})_{i+1}, \quad \omega^{(n)}\in \Omega.
  \end{equation} Note that each particle position $x_i^{(n)}$ is updated independently.

For the weights, suppose the particles $x_i^{(n)}$, $n=1,\dots,N$ are independent samples from $\pi_i$ then we have equal weighting for each particle
\[
     \pi_{i}^N = \frac{1}{N}\sum_{n=1}^{N} \delta(x_i^{(n)}).
\]
This does not change in the prediction step, thus
 \begin{equation}\label{eq:pfprediction}
         p_{i+1}^N = \frac{1}{N}\sum_{n=1}^{N} \delta(G(x_{i}^{(n)}, \omega^{(n)})_{i+1}) = \frac{1}{N}\sum_{n=1}^{N} \delta(x_{i+1}^{(n)}).
 \end{equation}
To go from $p_{i+1}^N$ to $\pi_{i+1}^N$, the weights $w_{i+1}^{(n)}$ need to be updated to take into account the observation data $y_{i+1}$ at time $i+1$. This is done using the likelihood function \eqref{eq:likelihood},
\begin{equation}
\bar{w}_{i+1}^{(n)} \propto g_{i+1}^{y_{i+1}}(x_{i+1}^{(n)}), \quad \sum_n \bar{w}_{i+1}^{(n)} = 1. \label{eq:pfupdate}
\end{equation} 
Using \eqref{eq:smc} but with the collection of updated particle positions and normalised weights $\lbrace x_{i+1}^{(n)}$, $\bar{w}_{i+1}^{(n)} \rbrace_{n=1}^{N}$ we obtain $\pi_{i+1}^{N}$.

In the above we assumed to have started with independent samples from $\pi_{i}$ before proceeding with prediction and update. Thus after we obtain $\pi_{i+1}^N$ we have to generate independent (approximate) samples from $\pi_{i+1}$ in order to iterate the above prediction and update steps for future times. This is done via \emph{selection} and \emph{mutation} steps. Otherwise the non-uniform weights are carried into future iterations until resampling is required.
 \begin{description}
        \item[Selection]  
        In high dimensions, $\pi_{i+1}^N$ can easily become \emph{singular} due to the observations being too informative. This means after the update step, most of the normalised weights are very small. Thus with a finite support, $\pi_{i+1}^N$ does not have enough particle positions in around the concentration of the true distribution $\pi_{i+1}$. Therefore it is desirable to add a \emph{resampling} step so that particles with low weights are discarded, and replaced with (possibly multiple copies of) higher weighted particles. This selection is done probabilistically; for example, one could draw uniform random numbers in the unit interval and select particles based on the size of $\bar{w}_{i+1}^{(n)}$, see \citet{bain2009fundamentals,reich2015probabilistic}.
        
        \item[Mutation] Since the resampling step can introduce duplicate particle positions into the ensemble, without reintroducing the lost diversity, repeated iterations of resampling will eventually lead to a degenerate distribution (i.e. measures whose support are singletons). To tackle this issue we apply \emph{jittering} after every resampling step. Jittering is based on Markov Chain Monte Carlo (MCMC) whose invariant measure is the target $\pi_{i+1}^N$. The jittering step shifts duplicate particle positions whilst preserving the target distribution. We discuss this in section \ref{subsec:mcmc}.
 \end{description}
 After resampling is applied, we obtain a new ensemble $\hat{x}_{i+1}^{(n)},\ n=1,\dots,N$ with equal weights $1/N$, i.e.
 \[
\pi_{i+1}^N = \frac{1}{N}\sum_{n=1}^{N} \delta(\hat{x}_{i+1}^{(n)}).
 \] When we do not resample, then the particles in the ensemble keep the weights given by $\bar{w}_{i+1}^{(n)}$, and use \eqref{eq:smc} for $\pi_{i+1}^{N}$. 
 
 The resampling step should be done only when necessary to reduce computational cost, because the jittering step requires evaluating the solution map $G$. Therefore we employ a test statistic to quantify the non-uniformity in the weights and only resample when the non-uniformity becomes unacceptable. For this we use the \emph{effective sample size} (ess) statistic. It is defined by the inverse $l^2$-norm of the normalised weights $\bar{\boldsymbol{w}} = (\bar{w}^{(1)},\dots,\bar{w}^{(N)})$,
 \begin{equation}\label{eq:ess}
         \ess\left(\bar{\boldsymbol{w}}\right)\triangleq\left\Vert \bar{\boldsymbol{w}}\right\Vert_{l^{2}}^{-2} = \frac{1}{\sum_n (\bar{w}^{(n)})^2}.
 \end{equation}
 The ess statistic measures the variance of the weights. If the particles have near uniform weights then the $\text{ess}$ value is close to $N.$ On the other hand if only a few particles have large weights then the $\text{ess}$ value is close to $1.$ In practice we resample whenever \eqref{eq:ess} falls below a given threshold 
 \[
         \text{ess} < N_\text{thresh}.
 \]
 Algorithm \ref{alg:smc_part1} summarises the bootstrap particle filter.
 The algorithm starts with an empirical approximation of the initial prior $\pi_{t_0}$ and steps forward in time, assimilates observation data in repeating cycles of prediction-update steps. The ess statistic is employed. When resampling is required, selection-mutation steps are applied.

 \begin{algorithm}\caption{Bootstrap particle filter}\label{alg:smc_part1}
        \MODIFY{Let total number of iterations $M$ be given.} Draw independent samples $x_0^{(n)}\sim \pi_0, \ i=1,\dots,N$ and set weights $\bar{w}_0^{(n)} = 1/N$
        \begin{algorithmic}[1]
                \FOR{$j=1,2,\dots,\MODIFY{M}$}

                \STATE Compute $x_{j}^{(n)} = G(x_{j-1}^{(n)}, \omega^{(n)})_j, \ n=1,\dots,N$ with $t_j,\,t_{j-1} \in \Lambda.$
                
                \STATE Obtain observation data $y_j$ and compute weights $\bar{w}_j^{(n)} \propto \bar{w}_{j-1}^{(n)} g_j^{y_j} (x_{j}^{(n)}), \ n=1,\dots,N$ using \eqref{eq:pfupdate}.
            
                \IF{$\ess < N_{\text{threshold}}$}
                        \STATE Sample $\hat{x}_j^{(n)},\ n=1,\dots,N$ according to the weights $\bar{w}_j^{(n)}$. 
                        
                        \STATE Set the weights to be $\bar{w}_j^{(n)}=1/N$. 
                        
                        \STATE \MODIFY{Apply jittering to the set  $\hat{x}_j^{(n)}, \ n=1,\dots,N$, if there are duplicates. }
                                                
                            \STATE Set $x_j^{(n)} = \hat{x}_j^{(n)}$, $n=1,\dots,N$.
                \ENDIF

                \ENDFOR
        \end{algorithmic}
 \end{algorithm}

\subsection{MCMC and jittering}\label{subsec:mcmc}

In this section we describe an effective Metropolis-Hastings MCMC based method called \emph{jittering} with the proposal step chosen specifically for our signal. Jittering reintroduces lost diversity due to resampling by replacing an ensemble of samples that contain duplicates $x_{t}^{(n)}\sim \pi_t, \ n=1,\dots,N$ with a new ensemble $\hat{x}_{t}^{(n)}, \ n=1,\dots,N$ without duplicates, such that the distribution $\pi_t$ is preserved. 

MCMC is a general iterative method for constructing ergodic time-homogeneous Markov chains $u(m), \ m\geq 0$ with transition kernel $K(u,\cdot)$, that are invariant with respect to some target distribution $\pi$, i.e.
\[ 
\pi K (\cdot) = \int K(u,\cdot) \pi(du)=\pi(\cdot).
\] 
By the Birkhoff's ergodic theorem, we have the following identity
\[
\int f(u)\pi\left(du\right)=\lim_{n\rightarrow\infty}\frac{1}{n}\sum_{k=1}^{n}f\left(u_{k}\right)\qquad{\rm a.s.}
\]
for any integrable and measurable function $f.$ Practically, this means starting from an initial $u(0)$, each $u(m)$ with $m \in \mathbb{N}$ can be treated as samples from the target distribution $\pi$.

A generic Metropolis-Hastings MCMC algorithm is described in algorithm \ref{alg:generic_mcmc}. A Markov transition kernel $K$ defined on the state space is used to generate proposals. Together with the right conditions on the acceptance probability function $a$ to guarantee \emph{detailed-balance}, the algorithm produces a Markov chain with kernel that is reversible with respect to the target measure $\mu$, see \citet{dashti2017bayesian}. In the Gaussian case, a classic and widely used choice for $K$ and $a$ is
\begin{align}\label{eq:mcmcgaussiankernel}
\begin{split}
 \MODIFY{K(u(m),dv)} &= \rho u(m) + \sqrt{1-\rho^2} \zeta,\quad \zeta \sim \mathcal{N}(0,\mathcal{C}) \\
 a(u,v) &= 1\wedge \frac{\exp(-\Phi(v))}{\exp(-\Phi(u))}
\end{split}
\end{align}
for any appropriate covariance operator $\mathcal{C}$ and log likelihood function $\Phi$, see \citet{kantas2014sequential}. The parameter $\rho$ controls the local exploration size of the Markov chain. In practise for high dimensional problems $\rho$ needs to be very close to 1 in order to achieve a reasonable average acceptance probability. For bad choices of $\rho$ the MCMC chain may mix very slowly and would require a burn-in step size that makes the whole algorithm computationally unattractive.

\begin{algorithm}[!ht]\caption{Generic Metropolis-Hastings MCMC, see \citet{dashti2017bayesian}}\label{alg:generic_mcmc}
        Let $\mu$ be a given measure on the state space. Let $u(0) \sim \mu$. Generate a $\mu$-invariant Markov chain $u(m), \ m>0$ as follows
        \begin{algorithmic}[1]
                \STATE Propose
                \begin{equation}
                 \tilde{u} \sim K(u(m),du)
                \end{equation}
                
                \STATE Accept $u(m+1)=\tilde{u}$ with probability \begin{equation} a(u(m), \tilde{u}), \end{equation} otherwise $u(m+1) = u(m)$.   
                
                \STATE $m \rightarrow m+1$ and repeat.  
        \end{algorithmic}
\end{algorithm}

With \eqref{eq:mcmcgaussiankernel}, algorithm \ref{alg:generic_mcmc} is known as the Preconditioned Crank Nicolson (pCN) and is wellposed in the mesh refinement limit, see \citet{dashti2017bayesian,kantas2014sequential}. Thus when applied to discretised problems the algorithm is robust under mesh-refinement. It is commonly applied in Bayesian inverse problems where the posterior is absolutely continuous with respect to a Gaussian prior on Banach spaces. It is important to note that here the design of the algorithm is important because in high dimensions measures tend to be mutually singular, but for Metropolis-Hastings algorithms the acceptance probability is defined as the Radon-Nikodym derivative given by the stationary Markov chain transitions.
 
Our choice \eqref{eq:prior} for the prior is not Gaussian. The distribution \eqref{eq:signalmarginaldistribution} is also not Gaussian for any $t\in(0,T]$. The distribution of the SPDE solution is investigated numerically in \citet{cotter2018numerically}, in which it is noted that non-Gaussian scaling is interpreted as intermittency in turbulence theory.  Therefore it is important to choose $K$ and $a$ such that the following properties hold.
\begin{enumerate}[(i)]
        \item Robustness under mesh refinement. Although we are considering finite dimensional state spaces, in the limit the state spaces under assumption \ref{ass:spdewellposedness} are infinite dimensional function spaces. 
        
        \item The chain should mix and stablise sufficiently quickly so that the number of burn-in steps required is reasonable.
\end{enumerate}
Then with the appropriately chosen $K$ and $a$, we apply algorithm \ref{alg:generic_mcmc} as a jittering step to shift apart duplicate particles introduced into the ensemble by the resampling step.

Given these considerations, we use directly the SPDE solution map $G$ \eqref{eq:spdesolutionop} to define the transition kernel $K$. Let the target distribution be the posterior distribution $\pi_{t_k}$, $t_k \in \Lambda$.
With a slight abuse of the notation introduced in \eqref{eq:spdesolutionop}, we write 
\[
G(u,W_{k-1:k})_{t_k}
\] to mean the solution of the SPDE at time $t_k$ along a realised Brownian trajectory over the time interval $[t_{k-1},t_{k}]$ starting from position $u\in \statespace$. When we consider $u\in\statespace$ and interval $\Delta_k:=t_k-t_{k-1}$ as fixed, then we view $G_{u,\Delta_k }(W) := G(u, W_{k-1:k})$ as a function on $\Omega$.

Let $\pi_{t_k}^N$ be the empirical approximation of $\pi_{t_k}$ with $N$ particles $u^{(n)}_{k}, \ n=1,\dots,N$. We consider each particle $u^{(n)}_{k}$ a child of some parent $u^{(n)}_{{k-1}}$ at time $t_{k-1}\in \Lambda$ for a realised Brownian trajectory $W$ over the interval $[ t_{k-1}, t_k]$, i.e.
\[
u^{(n)}_{{k}} = G_{u^{(n)}_{{k-1}},\Delta_k }(W)_k.
\]
To jitter $u_{k}^{(n)}$, set $W^{(0)} = W$ and $u(0) = u_{k}^{(n)}$ (see algorithm \ref{alg:generic_mcmc}). 
At the $m$-th MCMC iteration, $m\geq 1$, propose
\begin{equation}\label{eq:proposal}
\tilde{u} = G_{u_{k-1}^{(n)}, \Delta_k}(\rho W^{(m-1)} + \sqrt{1 - \rho^2} Z^{(m-1)})_k
\end{equation}
where $Z^{(m-1)}$ is a Brownian trajectory over $[t_{k-1}, t_k]$ generated independently from $W^{(m-1)}$.

We use the canonical Metropolis-Hastings accept-reject probability function
\begin{equation}\label{eq:acceptreject}
a(u(m-1),\tilde{u}) = 1\wedge \frac{g^{y}(\tilde{u} ) }{g^{y}(u(m-1))}
\end{equation}where $g^y$ is the likelihood function, see \eqref{eq:updating}.
The proposal \eqref{eq:proposal} is accepted with probability \eqref{eq:acceptreject} independently of $(\tilde{u}, u(m-1))$. In this case set
\[
 u(m) = \tilde{u}
\quad \text{ and }\quad
W^{(m)} = \rho W^{(m-1)} + \sqrt{(1 - \rho^2)} Z^{(m-1)}. 
\] 
Otherwise the proposal is rejected, in which case set
\[u(m) = u(m-1) \quad \text{ and }\quad W^{(m)} = W^{(m-1)}\]and go to the next iteration in algorithm \ref{alg:generic_mcmc}.

Algorithm \ref{alg:mcmc_jittering} summarises our MCMC procedure. The algorithm includes tempering scaling $\phi_{k}$ of the accept-reject function \eqref{eq:acceptreject}. \emph{Tempering} is explained in the next subsection. Practically, to save computation, we may apply jittering to just the duplicated particles after resampling, and run each jittering procedure for a fixed number of steps. 

\begin{prop}\label{thm:reversible}
        With the proposal \eqref{eq:proposal} and accept-reject function \eqref{eq:acceptreject}, the Markov chain generated by algorithm \ref{alg:generic_mcmc} is reversible with respect to $\pi_t^{y_{0:t}}$.
\end{prop}
\begin{proof} The generic Metropolis-Hastings algorithm \ref{alg:generic_mcmc} defines the following Markov transition kernel
	\begin{equation}
	Q(u,dv) = K(u,dv)a(u,v) + \delta_u (dv) \left(\int (1 - a(u,w))K(u,dw)\right).
	\end{equation} Since $\pi_t = p_t \frac{d\pi_t}{dp_t},$ if $K$ is such that it satisfies the detailed balance condition with respect to $p_t$
	\begin{equation}
	p_t(du)K(u,dv) = p_t(dv)K(v,du), \label{eq:detailebalance}
	\end{equation} then using the accept-reject function \eqref{eq:acceptreject} 
	\[
	a(u,v) = 1\wedge \frac{g_t^{y_t}(v)}{g_t^{y_t}(u)}
	\]
	we have $Q(u,\cdot)$ is a Markov kernel that is $\pi_t$-invariant, see \citet{dashti2017bayesian}.

	Let $\gamma$ denote the Wiener measure. Note that for a Brownian path $W\sim \gamma$ it is standard to show that the noise proposal in \eqref{eq:proposal}
	\[
	W' := \rho W + \sqrt{1-\rho^2} Z \sim \gamma 
	\] for $Z \sim \gamma$ independent of $W$. Thus due to the prediction formula \eqref{eq:prediction} and Markov transition \eqref{eq:signalmarginaldistribution} for the signal, conditioned on the value $u_{t-1}\in \statespace$, we have for $u = G_{u_{t-1}, \Delta}(W)_t\sim \pi_{t-1}k_t = p_t$, a sample $v$ obtained using the proposal \eqref{eq:proposal} is thus \[
	v = G_{u_{t-1},\Delta}(	W')_t\sim \pi_{t-1}k_t = p_t, \quad \text{ with } u_{t-1} \sim \pi_{t-1}
	\] i.e. conditioned on $u_{t-1}$, we have
	\[
		p_t(du) K(u, dv ) = p_t(du) p_t(dv) 
	\] which is symmetric in the pair $(u, v)$ giving us detailed balance \eqref{eq:detailebalance}.	
	
\end{proof}

\begin{remark}
The purpose of the jittering procedure is to introduce diversity in the samples, through a small perturbation, that is achieved via the Metropolis-Hastings algorithm. The perturbation is small as it is controlled by the parameter $\rho$ which we take it to be 0.9995 (if $\rho$ is equal to $1$ then no \MODIFYrev{bias} occurs). Therefore the error incurred in this step will be small, and this is the intention here rather than trying to preserve the underlying posterior distribution.
\MODIFYrev{
The reason for this error is as follows. If one operates directly on the target posterior distribution, the jittering step preserves the posterior distribution and therefore no bias occurs. However, the jittering procedure is applied to the approximating measure (the one given by the particle filter combined with the tempering methodology). As a result, there will be a ``local" error induced by the jittering step proportional to $(1-\rho)$. In the rigorous analysis of the rate of convergence of the particle filter, this appears as a separate term. This is covered, for example in \citet[lemma~4]{crisanDoucet2002}, and in \citet[section~4.2]{crisanmiguez2018} (albeit in a different context). 
}
\end{remark}

\begin{algorithm}[H]\caption{MCMC jittering for the 2D damped and forced SALT Euler dynamics}\label{alg:mcmc_jittering}
        Let $t_i\in \Lambda$, $\Delta_i = t_i - t_{i-1}$. \MODIFY{Let the number of MCMC steps $M$ be fixed.} Given the ensemble of equal weighted particle positions $ x^{(n), k}_{t_i},\ n=1,\dots,N$, corresponding to the $k$'th tempering step with temperature $\phi_k$, and proposal step size $\rho\in\left[0,1\right]$, repeat the following steps.       
        \begin{algorithmic}[1]
	        	\FOR {$n = 1,\dots,N$}

	        	\STATE Let particle $x^{(n),k}_{t_i}$ be such that $x^{(n),k}_{t_i}=G_{x^{(n)}_{i-1},\Delta_i} (W),$ for an initial condition $x^{(n)}_{t_{i-1}}\in \statespace$ and  a realised Brownian path $W(t_{i-1}:t_i)$ over the time interval $[t_{i-1},t_i]$.
	        	
	        	\STATE Set $u^0 = x_{t_i}^{(n),k}$ and $W^0(t_{i-1}:t_i) = W(t_{i-1}:t_i)$.

	        	\FOR {\MODIFY{$m =1,2,\dots,M$}}
	        	
                \STATE Propose $v \sim K(u{(m-1)}, \cdot)$ given by
                \[
                 v = G_{x^{(n)}_{{i-1}},\Delta_i}(W')
                \]
                where \[W' =\rho W^{m-1}\left({t_{i-1}}:t_i\right)+\sqrt{1-\rho^{2}}Z\left(t_{i-1}:t_i\right)\] for $Z$ a Brownian path independent of $W^{m-1}$.
                
                \STATE Accept $v$ with probability
                \begin{equation} \label{eq:tempered_acceptreject}
                  a(u(m-1), v) = 1\wedge \left(\frac{g_{t_i}^{y_{t_i}}(v)}{g_{t_i}^{y_{t_i}}(u(m-1))}\right)^{\phi_k}
                \end{equation}
                where $g_{t_i}^{y_{t_i}}$ is the likelihood function and $y_{t_i}$ is the observation at time $t_i$, in which case set $u^m = v$ and $W^m = W'$. Otherwise set $u^m = u^{m-1}$ and $W^m = W^{m-1}$.
                \ENDFOR
                
                \ENDFOR
        \end{algorithmic}
\end{algorithm}

\subsection{Tempering}

Empirical approximations of $\pi_t$ defined on high dimensional space can very quickly become degenerate, which is indicated by low effective sample size ($\text{ess}$) statistic. In order to facilitate smoother transitions between posteriors, so that ensemble diversity is improved, we employ the \emph{tempering} technique, see \citet{neal2001annealed,kantas2014sequential,beskos2017stable,beskos2014error}. Use of other techniques such as \emph{nudging} and \emph{space-time particle filter} (see \citet{beskos2017stable}) will be explored in future research work. 

We employ tempering when the $\text{ess}$ value for an posterior ensemble, falls below an apriori threshold $N_{\text{threshold}}$. The idea of tempering is to artificially scale the log likelihoods by a number $\phi \in(0,1]$ called the \emph{temperature}, which in effect increases the variance of the distribution so that the apriori $\text{ess}$ threshold is attained. Once this done resampling can be applied (with MCMC if required) which leads to a more diverse ensemble. Of course particles in this new ensemble are samples of the altered distribution which is not what we desire, therefore the procedure is repeated by finding the next temperature value in the range $(\phi,1]$. This is repeated until the temperature scaling is 1 so that the original distribution is recovered. 

More precisely, let
\begin{equation}
        0=\phi_0 < \phi_1 <\dots < \phi_R = 1
\end{equation}
be a sequence of \emph{temperatures}. Let
\begin{equation}\label{eq:tempering_transition}
        \pi_{t,r}(A) \triangleq C_{t,r}^{-1}p_{t} (g_t^{Y_{t}})^{\phi_r}(A)
\end{equation}
be called the tempered posterior at the $r$-th tempering step or simply the $r$-th tempered posterior, where $C_{t,r} = p_t (g_t^{Y_t})^{\phi_r}$ (compare with the recurrence formula \eqref{eq:updating}). Note that $\pi_{t,R} = \pi_{t}$ and $\pi_{t,0} = p_t$.
Thus with
\[
\frac{d\pi_{t,r}}{d\pi_{t,r-1}} \propto (g_t^{Y_t})^{\phi_r - \phi_{r-1}}
\] we have
\[
\pi_t= p_t \frac{d\pi_{t, 1}}{d\pi_{t,0}} \dots \frac{d\pi_{t,R}}{d\pi_{t,R-1}}
\] which suggests the iterative procedure,
\begin{equation}
	\pi_{t,r-1} \longmapsto \pi_{t,r} \propto \pi_{t,r-1} (g_t^{Y_t})^{\phi_r - \phi_{r-1}}, \quad r=1,\dots,R. \label{eq:iterative_tempering}
\end{equation}
Empirically this means, for each $r = 1, \dots, R$, assume we have equal weighted particle positions $\lbrace x_t^{(n)}\rbrace_{n=1,\dots,N}$ that give us the empirical $(r-1)$-th tempered posterior \[\pi_{t,r}^N = \frac{1}{N}\sum_{n=1}^{N}\delta(x_t^{(n)}),  \]
we compute unnormalised tempered weights 
\begin{equation}\label{eq:tempered_weights}
w^{(n),r}_t\left(\phi_r\right)\triangleq (g_t^{y_t}(x^{(n)}_t))^{\phi_r - \phi_{r-1}}, \quad n=1,\dots,N
\end{equation}
to obtain the empirical $r$-th tempered posterior \[
\pi_{t,r}^N = \sum_{n=1}^{N} \bar{w}_{t}^{(n),r}(\phi_r) \delta(x_t^{(n)}).
\] Then we resample according to $\pi_{t,r}^N$ and apply the MCMC jittering  algorithm \ref{alg:mcmc_jittering} (see remark \ref{rem:tempered_jittering}) to separate apart any duplicated particles before going to the $r+1$'th iteration step.

 The sequence of temperatures $\phi_r$ is chosen so that at each tempering iteration $r$, the empirical tempered distribution $\pi_{t,r}^N$ attains the apriori $\text{ESS}$ threshold $N_{\text{threshold}}$, i.e.
\begin{equation}
\ess(\bar{\boldsymbol{w}}^{r}(\phi_r))\geq N_{\text{threshold}} \label{eq:ess_threshold}
\end{equation} where $\bar{\boldsymbol{w}}^{r}(\phi_r) = (\bar{w}^{(1),r},\dots,\bar{w}^{(N),r})(\phi_r)$ are the normalised weights corresponding to \eqref{eq:tempered_weights}.
This way the choice for the temperatures can be made on-the-fly by using search algorithms such as \emph{binary search} at each tempering iteration. 

\begin{remark}\label{rem:tempered_jittering}
Proposition \ref{thm:reversible} shows the MCMC jittering algorithm preserves the target distribution $\pi_t$ with the accept-reject function \eqref{eq:acceptreject}. The same argument shows the algorithm preserves the tempered posteriors as long as the accept-reject function is chosen to be \eqref{eq:tempered_acceptreject}. The Markov transition kernel $K$ satisfies the detailed balance condition with respect to $p_t$ independent of tempering.
\end{remark}

Using tempering to smooth out the transition between consecutive filtering measures (i.e. from $\pi_{t_k}$ to $\pi_{t_{k+1}}$) ensures that the importance weights in \eqref{eq:tempering_transition} exhibit low variance, so that no small group of particles are favoured much more than the rest when resampling, see \citet{kantas2014sequential}, thus leading to a more diverse population. 

In algorithm \ref{alg:pf_temperingjittering} we summarise the complete procedure for one filtering step, i.e. from $\pi_{t_{i-1}}^N$ to $\pi_{t_{i}}^N$, incorporating adaptive tempering and MCMC jittering for SALT into the bootstrap particle filter.

\begin{algorithm}[H]\caption{One step particle filter for SALT with adaptive tempering and MCMC jittering}\label{alg:pf_temperingjittering}
        
        Consider the $i$'th filtering step corresponding to $t_i \in \Lambda$. Given the ensemble of equal weighted particle positions $\lbrace x^{(n)}_{i-1}\rbrace _{n=1,\dots,N}$ that define the empirical posterior $\pi_{i-1}^N$, we wish to assimilate observation data $y_i$ at time $t_{i}$ to obtain a new equally weighted ensemble $\lbrace x^{(n)}_i \rbrace_{n=1,\dots,N}$ that define $\pi_i^N$. Define \[
        \ess(\phi_{\cdot}, \{x_i^{(n)}\}_{n=1,\dots,N}) \triangleq \vert\bar{w}^{(n),\cdot}_i(\phi_{\cdot}) \vert_{l^2}^{-1}
        \] for $\phi \in (0,1]$, and \MODIFY{$\bar{w}_i^{(n),\cdot}$} are the normalised values of the unnormalised tempered weights \eqref{eq:tempered_weights}.
        
        \begin{algorithmic}[1]
            \STATE Compute $x^{(n)}_i = G(x_{i-1}^{(n)}, W)_{i-1},$ $n = 1, \dots,N$.
            
            \STATE Set $\MODIFY{\phi_0=0}, $  $r=1$, $x_i^{(n),0} = x_i^{(n)}$ each $n$.
            
            \WHILE{$\MODIFY{\ess}(\phi_{r-1},\{x_i^{(n),r-1}\})<N_{\text{threshold}}$}
            
            \STATE  Find (using e.g. binary search) the largest $\phi_{r}\in(\phi_{r-1},1)$ such that $\MODIFY{\ess}(\phi_r,\{x_i^{(n),r-1}\})\geq N_{\text{threshold}}.$ 
            
            \STATE Resample according to $\bar{w}^{(n),r}\left(\phi_r\right)$ to obtain a new set $\{x^{(n), r}_i\}$. 
            
            \STATE Apply jittering algorithm \ref{alg:mcmc_jittering} to any duplicated particles.
            

            \STATE r =  1
            \ENDWHILE
			
			\STATE Set $R = r$. Do steps 5. and 6. with $\phi_R = 1$ and set $x^{(n)}_i = x_i^{(n),r}$, $n=1,\dots,N$ to obtain $\pi_{t_i}^N$.
            
        \end{algorithmic}
\end{algorithm}


\section{Numerical setup and experiment results}\label{sec:numericalresults}

The setup for the numerical experiments follow on from the $\boldsymbol{\xi}$ calibration and uncertainty quantification work presented in \citet{cotter2018numerically}. Thus the parameter choices for the models are as follows: forcing strength \MODIFY{$a = 0.1$, number of gyres $b = 8$} and damping rate $r = 0.01$. 

The PDE \eqref{eq:2DEulerVorticity} and SPDE \eqref{eq:stochastic2dEuler} are prescribed on mesh of size $512\times512$ cells and $64\times64$ cells respectively for the spatial domain. We use a Galerkin finite element discretisation for the spatial variable and a third order stability preserving Runge-Kutta for the time stepping, see \citet{cotter2018numerically} for details.  This means spatially each mesh cell contains six grid points. Thus the PDE and SPDE velocity fields are of $3145728$ and $49152$ degrees of freedom respectively. Henceforth we shall refer to the PDE spatial dimension as \emph{fine resolution} and the SPDE spatial dimension as \emph{coarse resolution}\footnote{However since we are using an explicit in time method for solving the SPDEs, the coarse time step may need to be smaller to accommodate the fact that Brownian increments are unbounded.}.

The time step for the fine resolution is chosen in accordance with the CFL condition and in this case is $\Delta_{\text{f}} = 0.0025$. The CFL time step for the coarse resolution is $\Delta_{\text{c}} = 0.02$. 

The initial reference fine resolution PDE trajectory was spun-up from the configuration
\begin{equation}
\begin{aligned}
\omega_{\text{spin}} & =\sin(8\pi x)\sin(8\pi y)+0.4\cos(6\pi x)\cos(6\pi y)\\
& \qquad+0.3\cos(10\pi x)\cos(4\pi y) +0.02\sin(2\pi y)+0.02\sin(2\pi x) \label{eq:omega_spin}
\end{aligned}
\end{equation}
until some energy equilibrium state, see \citet{cotter2018numerically}. We call the equilibrium state's corresponding time point the initial time $t_0$.

We use \emph{eddy turnover time} (abbrev. ett) as the time dimension for the PDE system. It describes the time scale of flow features correponding to a given length scale, and is defined by \begin{equation}
 \tau_{l}\triangleq \frac{l}{|\bar{\vecu}|} \label{eq:ett}
 \end{equation}
where $|\bar{\vecu}|$ is the magnitude of the \MODIFY{stabilised} mean velocity\footnote{Our PDE system is spun-up from \eqref{eq:omega_spin} to an energy stable state. \MODIFY{By \emph{stabilised mean velocity}, we mean the $L1$ norm of the velocity field that corresponds to the energy stable state.} Thus $|\bar{\vecu}|$ is constant in time.}, and $l\in[0,L]$ a length scale. Here $L=1$ corresponds to the axis length of the domain $D$. For our experiments, we choose $l=\frac{1}{2}$. It is estimated that $1$ ett roughly equals to $2.5$ numerical time units, or $1000$ (fine resolution) CFL numerical time steps. Since the SPDE is thought of as a stochastic parameterisation for the PDE, we shall use the same eddy turnover time dimension for the SPDE. Thus 1 ett is $125$ coarse resolution CFL numerical time steps.  
 
 For the SPDE model, we use the calibrated EOFs $\vecx{\xi}_i, i=1,\dots,N_{\xi},$ from \citet{cotter2018numerically} with $N_{\xi}$ corresponding to $50\%$ of the total spectrum. This choice is informed by uncertainty quantification tests and amounts to $N_{\xi} = 51$ when the SPDE is prescribed on a mesh of size $64\times64$ cells.

\MODIFY{For the numerical filtering experiments, we consider two scenarios for the observations.}
    \begin{enumerate}
    
        \item \MODIFY{ Perfect model: the observations correspond to a single path-wise solution of the SPDE. Thus there is no discrepancy between the model and the true state. In this scenario the theoretical filtering formulation applies directly. We treat this scenario as a test case for the filtering algorithm.}
        
        \item \MODIFY{ Imperfect model: the observations correspond to the solution of the PDE, i.e. \eqref{eq:obs} is changed to
        \[
            Y_t = h(X_t^\dagger) + \epsilon_t
        \]
        where $X_t^\dagger$ corresponds to the coarse grained PDE velocity field, see remark \ref{remark: coarse graining}. 
        Here there is mismatch between the truth  and the signal. As shown in  \citet{cotter2018numerically}, the law of the SPDE discretised on the chosen grid converges to the law of the PDE as the discretisation grid gets refined. Implicity also the law of the sequence of true observations $(Y_{0},\ldots Y_t)$ is close to the law of the model observations. As stated in  \eqref{eq:filtering_density}, $\pi_t$ is a continuous function of the law of the signal and the observations $(Y_{0},\ldots Y_t)$ so we expect a reasonable approximation of $\pi_t$ even when we don't use the true law of the signal\footnote{The true law of the signal is the push-forward of $\pi_0$, the initial distribution of the signal $X_0$. In the case when $X_0$ is deterministic then the distribution of the signal is a Dirac delta distribution.} but the model law.\footnote{For continuous time models, this property is called the robustness of the filter. See \citet{ClarkCrisan} for results in this direction.}   
    }
    \end{enumerate}

\MODIFY{\begin{remark}\label{remark: coarse graining}
In the imperfect model scenario, since the SPDE solution is meant to capture the large-scale features of the deterministic fine resolution dynamics that are resolvable at the coarse resolution, we should obtain observations from the coarse grained PDE solution. For coarse graining, we use the inverse Helmholtz operator \begin{equation}\label{eq:coarse_graining}
H \triangleq \big(\identity - \frac{1}{k^2}\Delta\big)^{-1}
\end{equation} 
and apply $H$ to the PDE stream function \eqref{eq:laplace_stream_is_vorticity} to average out its small scale features. The boundary condition we impose on the coarse grained stream function is the same Dirichlet condition as for \eqref{eq:laplace_stream_is_vorticity}. The value $k$ in the definition of $H$ corresponds to the coarse resolution, in this case $k=64$.  To obtain the coarse grained PDE velocity field, apply the linear operator $\nabla^\perp$ to the coarse grained stream function. The coarse grained PDE velocity field is then used to generate the observation data in the imperfect model filtering scenario.  It is important to note that this coarse graining procedure is only applied when we obtain observation data, the underlying fine resolution dynamics is unchanged. 
\end{remark}
}

In both scenarios the observations are defined as noisy point measurements of the truth's velocity field. The observation locations (thought of as ``weather stations'') are given by a uniform regular grid of dimension \MODIFY{$d_y$}; see section \ref{secfilterproblem} for the problem's mathematical formulation. \MODIFY{We investigate the impact of the number of weather stations using $d_y = 289$, $d_y = 1089$ and $d_y=81$ in some experiments}. For this paper we only consider fixed uniform geometry for the weather stations. Further the $81$ weather stations are a subset of the $289$ weather stations\MODIFY{, and the $289$ weather stations in turn are a subset of the $1089$ weather stations}.
Figure \ref{fig:obsillustration} visually illustrates a snapshot of the coarse grained numerical PDE solution velocity vector field overlaid with the positions of the $81$ weather stations. 

\begin{remark} \label{rem:computation_limitations}
	The dimension of the observation space compared to the dimension of the underlying truth is very small. Using $289$ weather stations amounts to $1.18\%$ of the overall degrees of freedom in the perfect model scenario, and $0.01837\%$ of the overall degrees of freedom in the imperfect model scenario. \MODIFY{Observation error size and/or observation data dimension affects the number of tempering and jittering steps.} \MODIFY{For our observation error size choice,} these parameter choices are the best we can do
	given our computational hardware, so that we can obtain numerical results in a reasonable amount of time.  \MODIFY{Figures \ref{figure:spde_statistics} and \ref{figure:pde_statistics} provide computation cost estimates for the numerical experiments, measured in terms of number of equation evaluations.} For reference, all numerical experiments for this paper  were run on a workstation equipped with two Intel Xeon CPUs totalling $32$ logical processors, and $64$GB of memory. \MODIFY{To evaluate SPDE ensembles, the ensemble members were run in parallel, in batches of 25.}
\end{remark}

\begin{remark}
\MODIFY{To illustrate the difference in computational cost between the fine resolution and the coarse resolution models, we ran a benchmark test. For a time interval of 0.1 time units, the fine resolution PDE (time step 0.0025 time units) took 24 seconds to run on our workstation. The coarse resolution SPDE (time step 0.02 time units) took 0.6 seconds to run on our workstation.}
\end{remark}

The observation error covariance \MODIFY{$\gamma$} (see \eqref{eq:obs} for definition) is calibrated by computing the standard deviation of the fine resolution PDE velocity field within coarse cells, and then averaged along the time axis. More precisely, let $\statespacepde$ denote the discretised PDE state space. Let $X_t \in \statespacepde$ denote a snapshot in time of the PDE velocity field. Let superscript indices denote vector component. Define $\bar{X}_t \in \statespacepde$ by
\[
\bar{X}_{t}^{i_k} = \frac{1}{\card(j : X_t^j \in \text{coarse cell $k$})}\sum_{j : X_t^j \in \text{coarse cell $k$}} X_t^j,\qquad i_k \in \{ i : X_t^i\in \text{coarse cell $k$}\}
\]
for coarse cells corresponding to the coarse resolution mesh. Thus $\bar{X}_t^i$ are the local coarse cell averages of $X_t$. Then we define \MODIFY{${\gamma(\lambda)}$} by
\begin{equation}\label{eq:obs_error_calib}
\gamma(\lambda)=\lambda \frac{1}{M}\sum_{i=1}^{M} \left( \left(X_{t_i}-\bar{X}_{t_i}\right)\otimes\left(X_{t_i}-\bar{X}_{t_i}\right)
\right)^{1/2}
\end{equation}
where $\lambda$ is a \MODIFY{scaling we use to control the magnitude of the observation error.} \MODIFY{We choose $\lambda = 0.6$ for our numerical experiments.}
The idea is $\gamma(\lambda)$ at the observation locations represent the local variability of the truth at the observation locations. 
The $\gamma(\lambda)$ computed \MODIFY{using \eqref{eq:obs_error_calib}} is a vector field defined on the fine resolution grid. It is evaluated at the observation locations when used in \eqref{eq:obs}. Figure \ref{fig:obserrillustration} visually illustrates the magnitude of $\gamma(1)$ overlaid with the observation locations. 
We use the same calibrated $\gamma(\lambda)$ in both problem scenarios.

\begin{figure}[!htbp]
	\centering
	\begin{minipage}{0.9\textwidth}
	\begin{subfigure}{0.5\textwidth}
		\includegraphics[width=\textwidth]{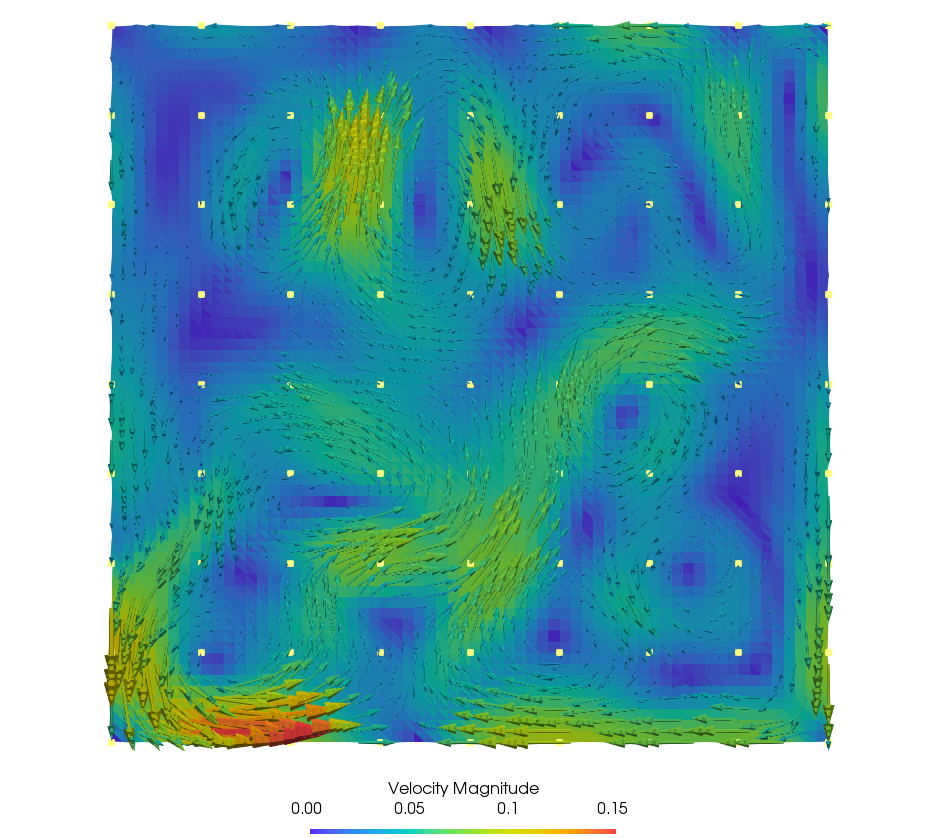}
		\caption{Snapshot of the coarse grained PDE velocity field at a $t\in\Lambda$ overlaid with observation locations (yellow dots) which are defined by a grid of $(8\times8)$ cells.}
		\label{fig:obsillustration}
	\end{subfigure}
	\quad
	\begin{subfigure}{0.5\textwidth}
		\includegraphics[width=\textwidth]{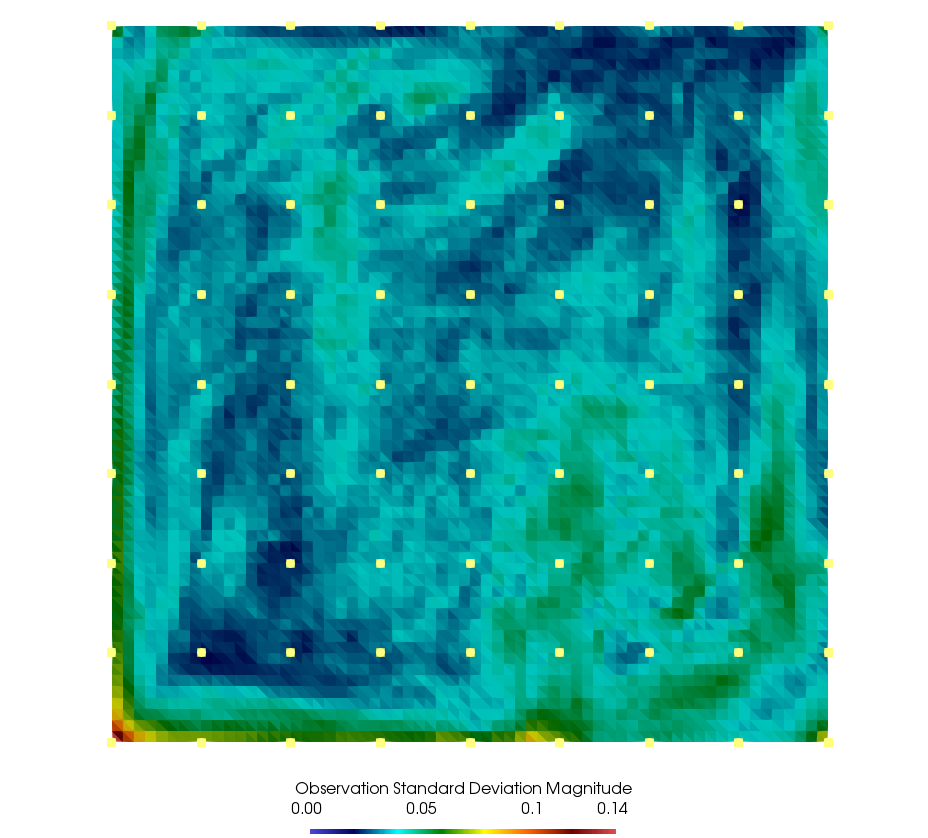}
		\caption{Magnitude of the calibrated observation error $\gamma(1)$ (see \eqref{eq:obs_error_calib}) overlaid with observation locations (yellow dots) which are defined by a grid of $(8\times 8)$ cells.}
		\label{fig:obserrillustration}
	\end{subfigure} 
	\caption{Observation locations and observation error magnitude.} 
\end{minipage}
\end{figure}

In the perfect model scenario, the truth from which we obtain the observations is a single simulated realisation of the SPDE. The initial condition for the SPDE truth is a particular sample from $\pi_0$. See figure \ref{fig:initialSPDEtruth} for a visualisation of the SPDE truth without observation noise at the initial time $t_0$. We discuss what $\pi_0$ is and how we sample from it in section \ref{sec:initialdist}.
In the imperfect model scenario the truth is the coarse grained PDE velocity field. Figure \ref{fig:initialPDEtruth} shows a visualisation of the PDE truth without observation noise at the initial time $t_0$.
\begin{figure}[!htbp]
    \centering
	\begin{minipage}{0.9\textwidth}
		\begin{center}
			\includegraphics[width=\textwidth]{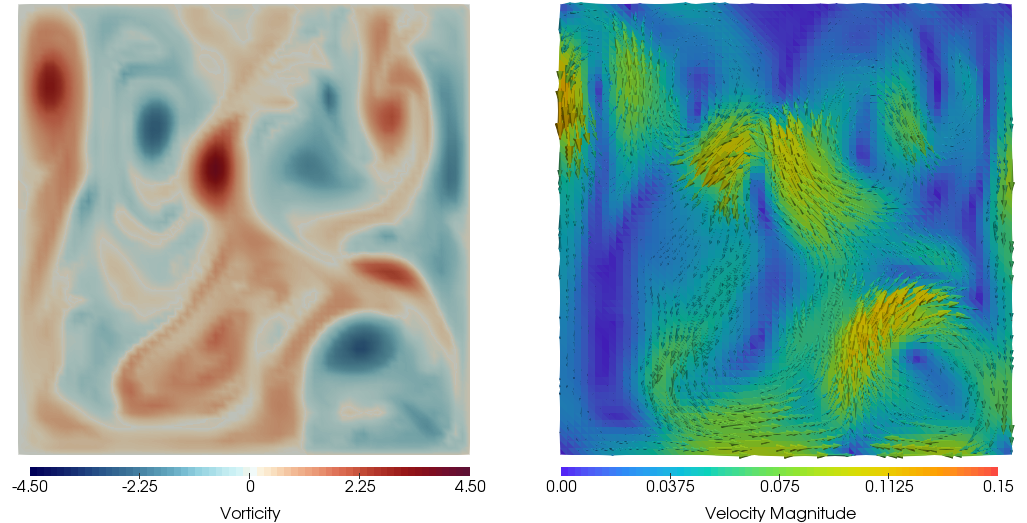}\par
		\end{center}
		\caption{Perfect model scenario. Visualisation of the  \emph{vorticity} scalar field (left) and the \emph{velocity} vector field (right) of the SPDE truth at the initial time $t_0$. The different colours of the vorticity field reflect the clockwise and anti-clockwise directions of the velocity vectors. The velocity vectors are visualised using arrows, the size of which reflect their magnitude. This initial condition is obtained by applying the \emph{deformation} procedure described in section \ref{sec:initialdist} to the PDE truth shown in figure \ref{fig:initialPDEtruth}, for $104$ fine resolution numerical time steps.}
		\label{fig:initialSPDEtruth}
	\end{minipage} 
\end{figure}
\begin{figure}[!htbp]
	\centering
	\begin{minipage}{0.9\textwidth}
		\begin{center}
		\includegraphics[width=\textwidth]{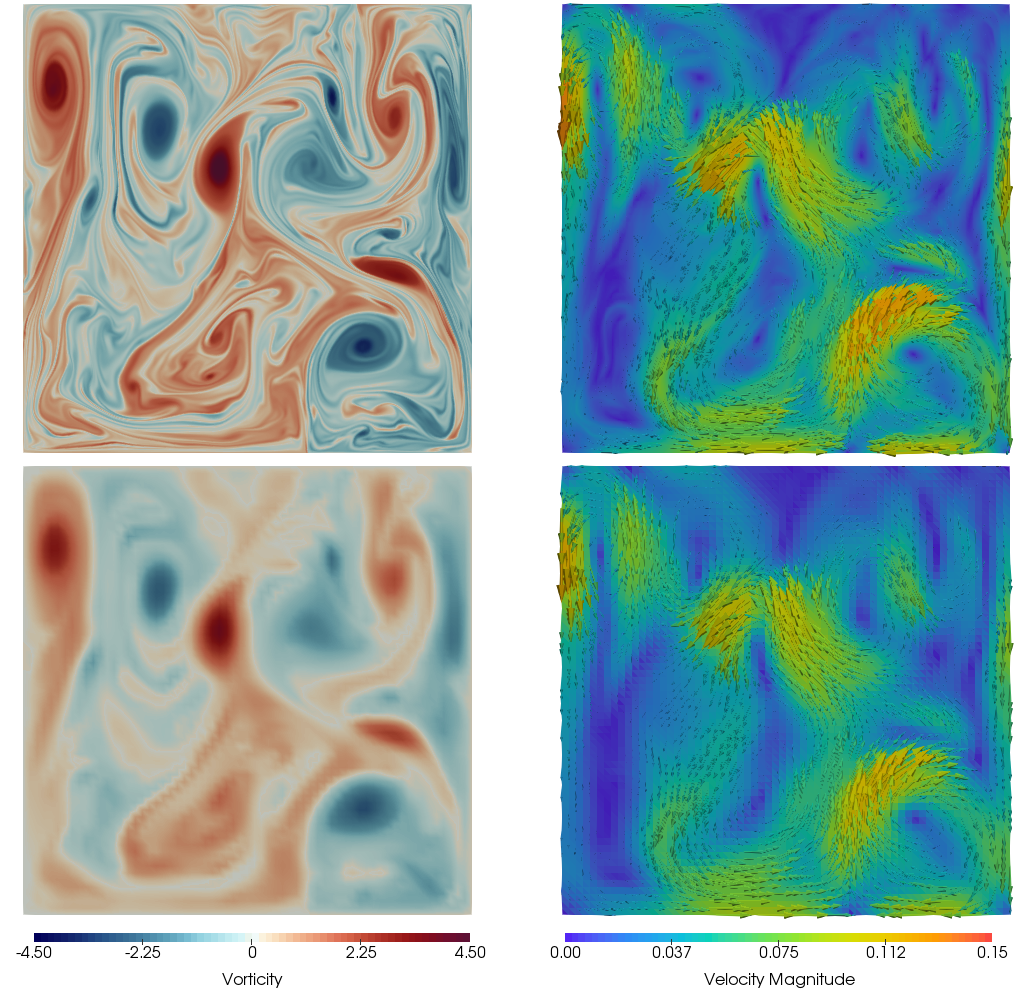}
		\par
		\end{center}
		\caption{Imperfect model scenario. Visualisation of the PDE truth at the initial time $t_0$. Here the top row correspond to the fine resolution visualisation of the \emph{vorticity} scalar field (left) and \emph{velocity} vector field (right). The bottom row correspond to the \emph{coarse grained} version of the top row using the coarse graining operator \eqref{eq:coarse_graining}. The different colours of the vorticity fields reflect the clockwise and anti-clockwise directions of the velocity vectors. The velocity vectors are visualised using arrows, the size of which reflect their magnitude.}
		\label{fig:initialPDEtruth}
	\end{minipage}
\end{figure}

We use an ensemble size of $N=100$ particles. Each particle's initial condition is a sample from the initial distribution $\pi_0$. Also we do not assimilate at $t_0$ since the initial distribution is assumed given, see section \ref{prob:filtering_problem}. 

\MODIFY{
Additionally, we introduce the following two distributions, whose ensemble approximations are utilised in our numerical experiments. We let
\begin{equation}
\pp_{i,j} := \mathbb{P}(X_{i+j} \in \cdot | Y_{0:i})
\end{equation}
be the \emph{$j$ step forecast distribution}. We let 
\begin{equation}
\qq_i := \mathbb{P}(X_i \in \cdot | X_{0})
\end{equation}
denote the \emph{prior distribution}.
In section \ref{subsec:smc}, we have used superscript $N$ to denote the $N$ particle approximation of a distribution. For notational convenience, throughout this section, we drop the superscript $N$ when referring to the ensembles.}
	
\MODIFY{To analyse the numerical results, we evaluate the following statistics.}
\begin{itemize}
	\item \MODIFY{The \emph{root mean square error}} (\MODIFY{rmse}) between the ensemble mean \MODIFY{of a $N$ particle measure and a verification. For example, the rmse between the $j$ step forecast ensemble mean $\bar{\pp}_{i,j}$ and the true system state $X^{\dagger}_j$} at time index $j$ is given by
	\MODIFY{\begin{equation}\label{eq:rmse_statistic}
	\rmse(\bar{\pp}_{i,j}, X^{\dagger}_j ) := \| \bar{\pp}_{i,j} - X^{\dagger}_j \|_{L^2}.
	\end{equation}}
	
	\item \MODIFY{The \emph{\MODIFY{ensemble spread}} ($\sprd$) of a $N$ particle measure. For example, the \MODIFY{ensemble spread} of the forecast distribution $\pp_{i,j}$ is given by} 
	\MODIFY{\begin{equation}
	\sprd(\pp_{i,j}) = \sqrt{\frac{1}{N-1}\sum_{X \in \pp_{i,j}} \| X - \bar{\pp}_{i,j} \|^2_{L^2}}. \label{eq:stddev_statistic}
	\end{equation}}
	
	\item The effective sample size (ess) statistic \eqref{eq:ess} for measuring the variance of the ensemble weights. \MODIFY{Throughout, We choose the $\ess$ threshold to be $80\%$ of the ensemble size, i.e. \[N_{\text{threshold}}=80\% \times N.\]}
	
	\item \emph{Rank histograms} for assessing the reliability of the particle filter, see \citet{brocker2018assessing,reich2015probabilistic}. This is a standard measure of ensemble reliability.  At any reference grid location, given the ensemble values $\{x_t^i\}_{i=1,\dots,N}$ that corresponds to the forecast distribution $p_t^{N}$ \eqref{eq:prediction}, and an observation value $y_t$, define the \emph{rank} function
	\begin{equation}\label{eq:rank}
	R(y_t, \{x^{i}_t\}_{i=1,\dots,N}) = 
	k \quad \text{if } x^j_t \leq y_t \text{ for }j < k, \text{ and } x^j_t > y_t \text{ for } j \geq k.
	\end{equation}
	The rank function $R$ takes values in $\{0,1,\dots,N\}$. If the ensemble forecast is reliable then $R$ is a uniform random variable, meaning the verification and the ensemble members are indistinguishable. Thus collecting the rank values over time $t\in \Lambda$, we should obtain a ``flat'' histogram plot if the particle filter gives reliable results.	
	Further it is shown in \citet{brocker2018assessing} that the rank statistic $R$ is of $\chi^2$ distribution with $N$ degrees of freedom.
\end{itemize}

\begin{table}[!htbp]
	\centering
	\begin{minipage}{0.9\textwidth}
		\begin{center}
\begin{tabular}{ll|l|l|}
\cline{3-4}
                                                                                                    &                       & $\rmse$ & $\sprd$              \\ \hline
\multicolumn{1}{|l|}{\multirow{2}{*}{\begin{tabular}[c]{@{}l@{}}Initial ens.\\ set 1\end{tabular}}} & Perfect model scen.   &   $1.8809\times 10^{-4}$   & \multirow{2}{*}{$1.23\times10^{-3}$} \\ \cline{2-3}
\multicolumn{1}{|l|}{}                                                                              & Imperfect model scen. &  $1.4614\times 10^{-4}$    &                   \\ \hline
\multicolumn{1}{|l|}{\multirow{2}{*}{\begin{tabular}[c]{@{}l@{}}Initial ens.\\ set 2\end{tabular}}} & Perfect model scen.   &  $1.2020\times 10^{-2}$    & \multirow{2}{*}{$2.46\times10^{-2}$} \\ \cline{2-3}
\multicolumn{1}{|l|}{}                                                                              & Imperfect model scen. &  $1.2017\times10^{-2}$    &                   \\ \hline
\end{tabular}

	\end{center}
	\caption{\MODIFY{Two sets of initial ensemble were generated using the deformation procedure described in section \ref{sec:initialdist}. Each set contains $100$ ensemble members. In this table we report the root mean square error and ensemble spread statistics. For each set, the rmse is computed using the true state without observation noise for the perfect model scenario and the imperfect model scenario. The ensemble spreads are the same in each case for both scenarios since  \eqref{eq:stddev_statistic} does not depend on the true state.}}
	\label{tab:initial_ensemble}	
	\end{minipage}
\end{table}

	\subsection{Initial distribution}\label{sec:initialdist}
	
	The initial distribution $\pi_0$ comes from the following construction which we call \emph{deformation}, see \citet{cotter2018numerically}. \MODIFY{The procedure can be understood as applying a random temporal scaling to a given field.}
	Let $\omega_\text{truth}$ be a fine resolution PDE vorticity field.
	Using the coarse graining operator $H$ (defined in \eqref{eq:coarse_graining}), define operator $\vortproj: \statespacepde \times \reals \rightarrow \statespacespde$ by
	\begin{equation}
	    \vortproj(\vecu, \beta) = \nabla^\perp H( \Delta^{-1} \omega^{\vecu, \beta})
	\end{equation}
	where $\omega^{\vecu, \beta}$ is the (vorticity) solution of the linear PDE
	\begin{align}
	\partial_t\omega+ \beta\vecu \cdot\nabla\omega&=0 \label{eq:initial_cond_deformation}\\
	\omega_0 &= \omega_\text{truth}
	\end{align}
	$\beta\sim\mathcal{N}(0,\epsilon),$ is a centered Gaussian weight with an apriori variance parameter $\epsilon,$ and $\vecu \sim \mathcal{U} \left(\statespacepde\right)$ is random draw from a uniform distribution on $\statespacepde$. $\beta$ and $\vecu$ are independent.
	Then 
	\begin{equation}\label{eq:prior}
	\pi_0(A ) \triangleq \prob \left(\vortproj(\vecu,\beta) \in A \right)\qquad A \in \mathcal{B}(\statespacespde)
	\end{equation}
	
	\begin{remark}\label{rem:initial_distribution}
		Practically, we randomly draw a vorticity field from the energy stable period prior to the initial data assimilation time point $t_0$. The drawn vorticity state is then used to compute its corresponding stream function by inverting the Laplacian and using the same Dirichlet condition as for \eqref{eq:laplace_stream_is_vorticity}. The  velocity field $\vecu$ in \eqref{eq:initial_cond_deformation} is then obtained from the stream function. Thus for the linear system \eqref{eq:initial_cond_deformation} the boundary condition is supplied via the sampled $\vecu$.
	\end{remark}
	
	In Hamiltonian mechanics, the conservation laws associated with relabelling symmetries are called Casimirs. In lemma \ref{lem:casimir_preserving} we show our choice for the prior distribution is physical in the sense that any sample generated by the procedure $K(\vecu,\beta)$ preserves the Casimirs of the truth $\omega_\text{truth}$. 
	\begin{definition}[Casimir, see \citet{gay2013selective}]
		For 2D incompressible ideal fluid motion, the Casimirs are 
		\[
		C_{\Phi} = \int_{\mathcal{D}} \Phi(\omega) dx
		\]for any $\Phi \in C^{\infty} (\reals, \reals) $.
	\end{definition}
	
	\begin{lemma}[Preservation of Casimirs]\label{lem:casimir_preserving} Let the domain $\mathcal{D}$ be bounded with piecewise smooth boundary. Assume the sampled vector field $\vecu \in \statespacepde$ is divergence free and $\vecu \cdot \hat{n} = 0$ with $\hat{n}$ being the normal to the boundary $\partial \mathcal{D}$, then $\omega^{\vecu,\beta}$ preserves the Casimir values of $\omega_{\text{truth}}$.
	\end{lemma}
	
	\begin{proof}
		We have
		\begin{align*}
		\frac{d}{dt} C_{\Phi} 
		& = \int_{\mathcal{D}} \frac{d}{dt} \Phi(\omega) dx \\
		& = \int_{\mathcal{D}} \Phi'(\omega)\partial_t\omega \ dx 
		= - \int_{\mathcal{D}}\Phi'(\omega) \beta\vecu \cdot\nabla\omega \ dx \\
		& =  - \int_{\mathcal{D}} \beta\vecu \cdot\nabla \Phi(\omega) \ dx
		= 0
		\end{align*}
		where the last equality follows from integration by parts and the conditions assumed on $\vecu$.
	\end{proof}
	
\begin{figure}[!htbp]
	\centering
	\begin{minipage}{0.9\textwidth}
		\begin{center}
			\includegraphics[width=0.65\textwidth]{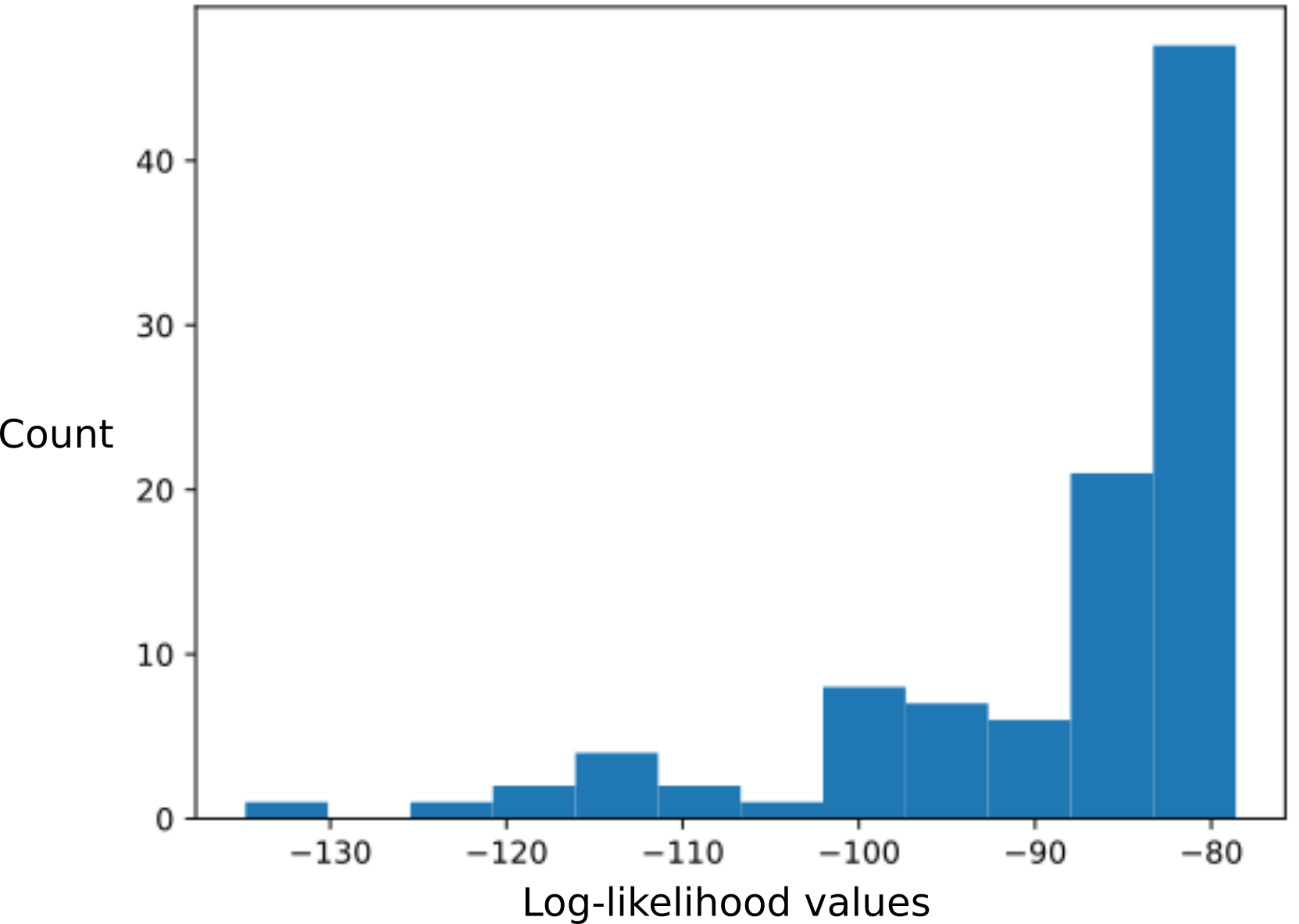}\par
		\end{center}
		\caption{Histogram showing log-likelihood $\ln g^{Y_t}(\cdot)$ weights for an ensemble of $100$ particles that define the forecast distribution at time $t_1$ with $t_1 - t_0 = 1$ eddy turnover time. The highly right skewedness of the bins and the numerical range of log-likelihood values demonstrate the singular nature of $\pi_{t_1}^{N}$ and the requirement of tempering and MCMC jittering techniques to make the basic bootstrap particle filter work. Otherwise resampling would lead to degenerate posteriors $\pi_t^{N}$.}
		\label{fig:loglikelihoodhistogram}
	\end{minipage} 
\end{figure}

\MODIFY{
To investigate the impact the initial distribution can have on the filtering experiment, 
we generate two different sets of initial ensemble. Each set contains $100$ ensemble members.
For both sets, $\omega_\text{truth}$ is taken to be the imperfect model observation's initial condition (see figure \ref{fig:initialPDEtruth} for visualisations), and we choose $\epsilon = 0.25$ for the random scaling parameter $\beta$. 
For the first set, equation \eqref{eq:initial_cond_deformation} is solved for $104$ fine resolution CFL time steps. For the second set,  equation \eqref{eq:initial_cond_deformation} is solved for $3000$ fine resolution CFL time steps. This way, we obtain two initial ensembles whose ensemble average rmse differ by two orders of magnitude, see table \ref{tab:initial_ensemble}.}


Before discussing experiment results, figure \ref{fig:loglikelihoodhistogram} shows a histogram of the log-likelihood $\ln g^{Y_t}(\cdot)$ values for an ensemble of $100$ particles that defines the forecast distribution $p_{t_1}^{N}$, with $t_1 - t_0 = 1 $ ett. It shows  straightforwardly the singular nature of $\pi_{t_1}^{N}$ and that without tempering and MCMC jittering, a plain bootstrap particle filter algorithm would fail in the sense that particle diversity would be lost very quickly, leading to degenerate posteriors $\pi_t^{N}$.

\subsection{Perfect model scenario} \label{sec:perfectmodelscenario}
\begin{figure}[!htbp]
	\centering
	\begin{minipage}{0.9\textwidth}
		\includegraphics[width=\textwidth]{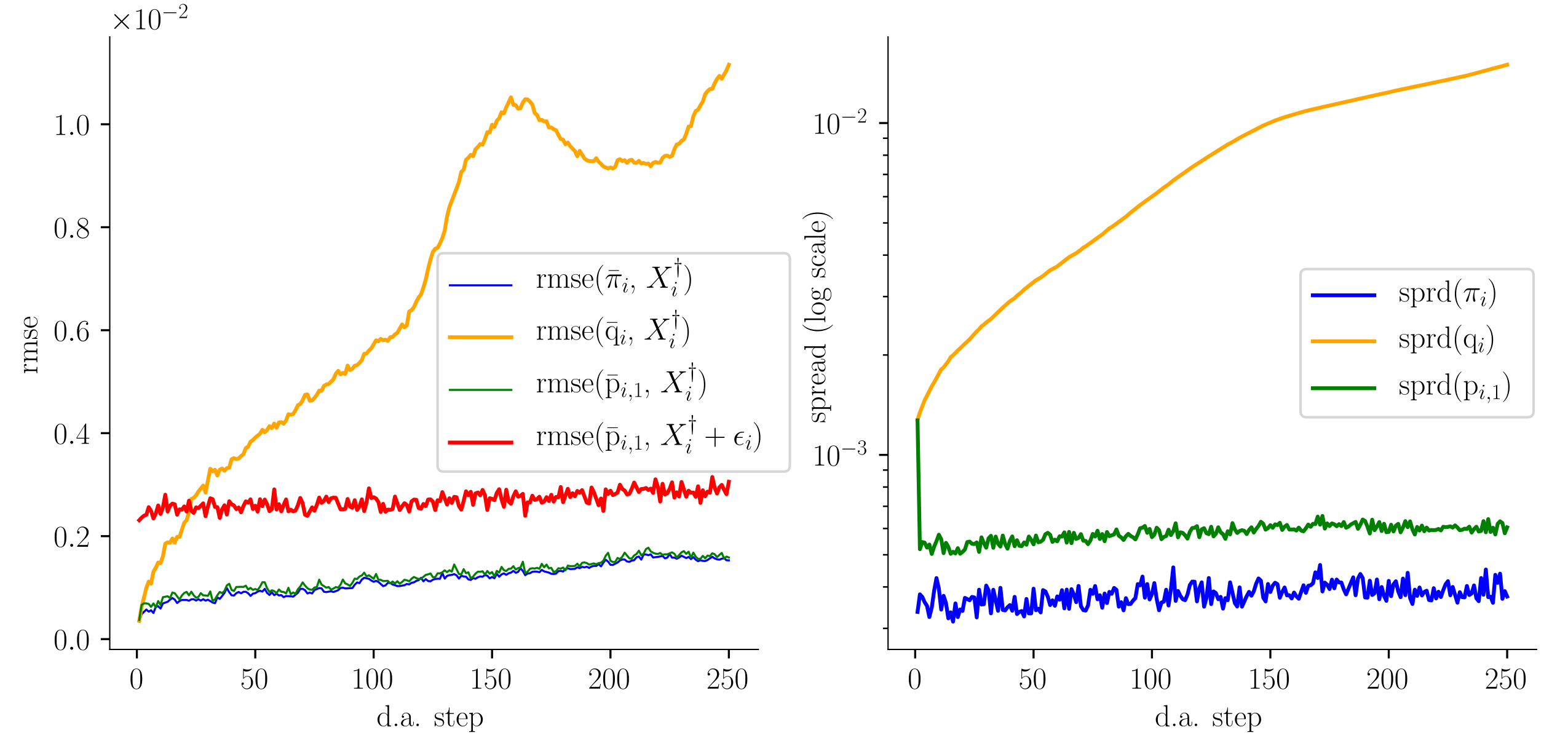}
		\caption{\MODIFY{Perfect model scenario (experiment using initial ensemble set 1). The left figure compares rmse values and the right figure compares the ensemble spreads. In the rmse figure, the blue plot shows the rmse between the posterior ensemble mean and true state; the green plot shows the rmse between the one step forecast ensemble mean and the true state; the red plot shows the rmse between the one step forecast ensemble mean and the true state plus observation noise; the orange plot shows the rmse between the prior ensemble mean and the true state. In the spread figure, the blue plot shows the spread of the posterior ensemble; the green plot shows the spread of one step forecast ensemble; the orange plot shows the spread of the prior distribution ensemble.} }
		\label{figure: rmse and spread initial ensemble set 1}
	\end{minipage}
\end{figure}

\begin{figure}[!htbp]
	\centering
	\begin{minipage}{0.9\textwidth}
		\includegraphics[width=\textwidth]{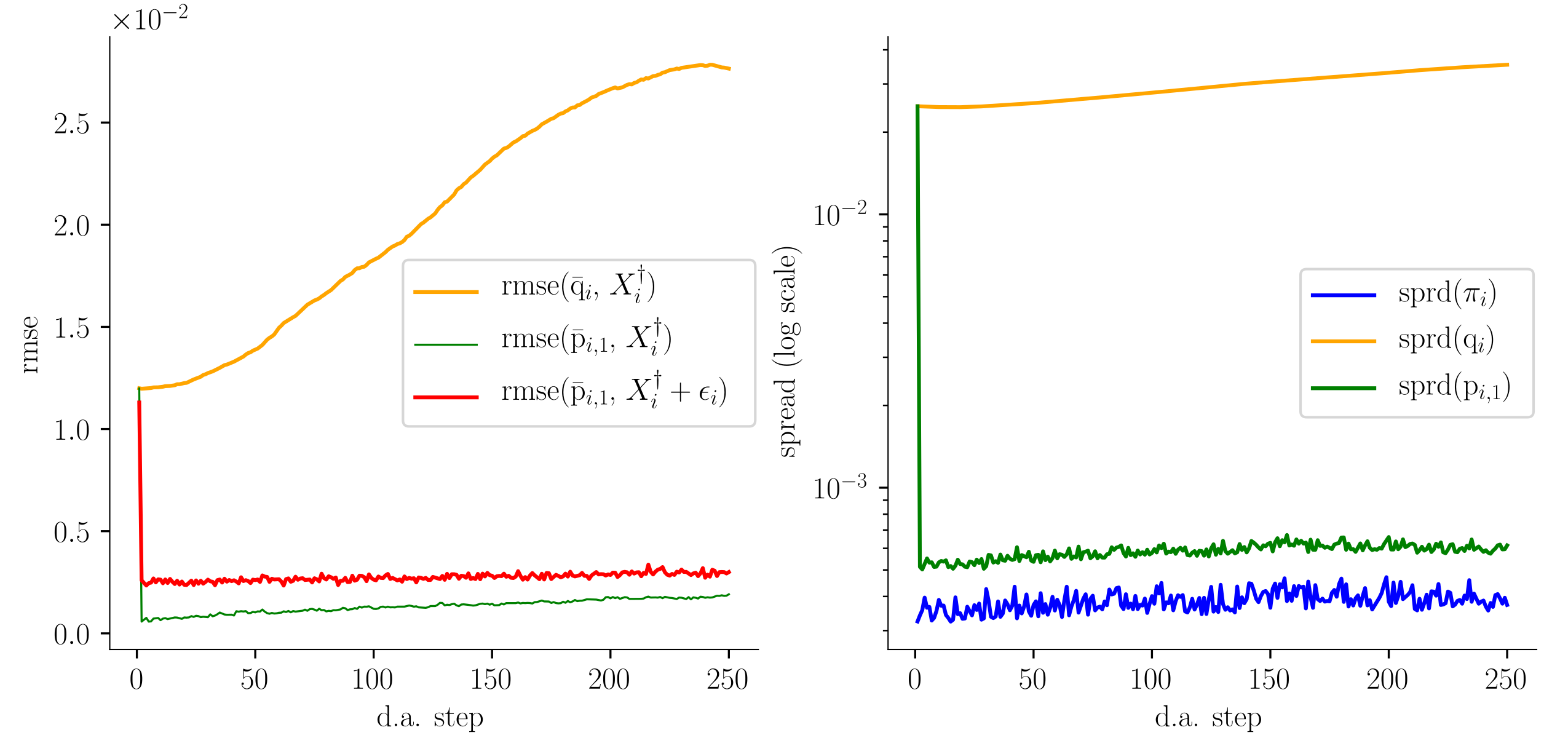}
		\caption{\MODIFY{Perfect model scenario (experiment using initial ensemble set 2). The left figure compares rmse values and the right figure compares the ensemble spreads. In the rmse figure, the green plot shows the rmse between the one step forecast ensemble mean and the true state; the red plot shows the rmse between the one step forecast ensemble mean and the true state plus observation noise; the orange plot shows the rmse between the prior ensemble mean and the true state. In the spread figure, the blue plot shows the spread of the posterior ensemble; the green plot shows the spread of one step forecast ensemble; the orange plot shows the spread of the prior distribution ensemble.}} 
		\label{figure: rmse and spread initial ensemble set 2}
	\end{minipage}
\end{figure}

\begin{figure}[!htbp]
	\centering
	\begin{minipage}{0.9\textwidth}
	\includegraphics[width=\textwidth]{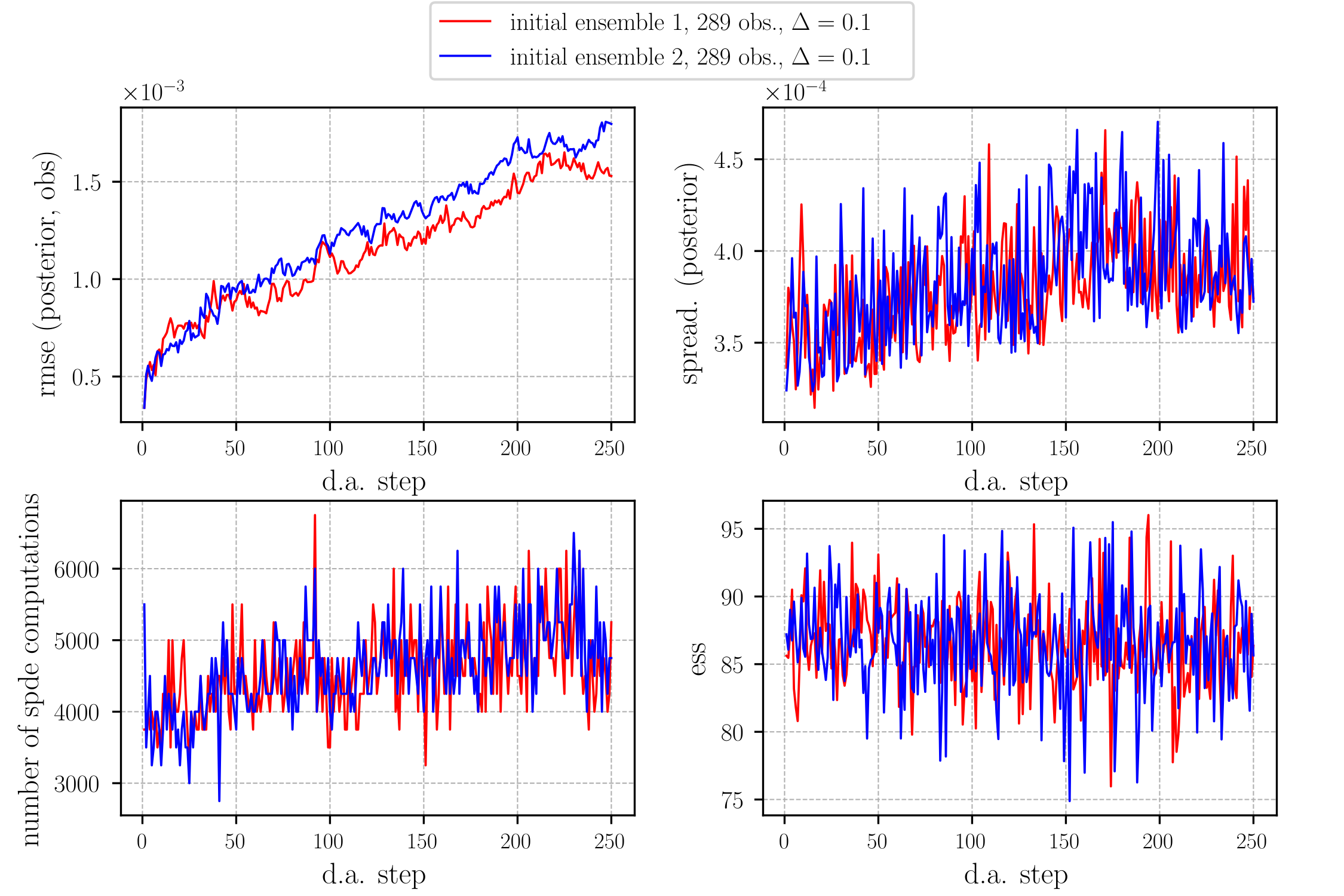}
	\caption{\MODIFY{Perfect model scenario. The four figures show the root mean square error between the posterior ensemble mean and true state, ensemble spread of the posterior, number of SPDE evaluations at each assimilation step, and ess. The red plots correspond to the experiment using initial ensemble set 1. The blue plots correspond to the experiment using initial ensemble set 2. For both experiments the assimilation interval of $\Delta=0.04$ ett ($0.1$ time units) and $289$ weather stations were used. Both experiments were run for a total of $10$ ett, which amounted to $250$ data assimilation steps.} }
	\label{figure:spde_statistics}
	\end{minipage}
\end{figure}

\begin{figure}[!htbp]
	\centering
	\begin{minipage}{0.9\textwidth}
	\includegraphics[width=\textwidth]{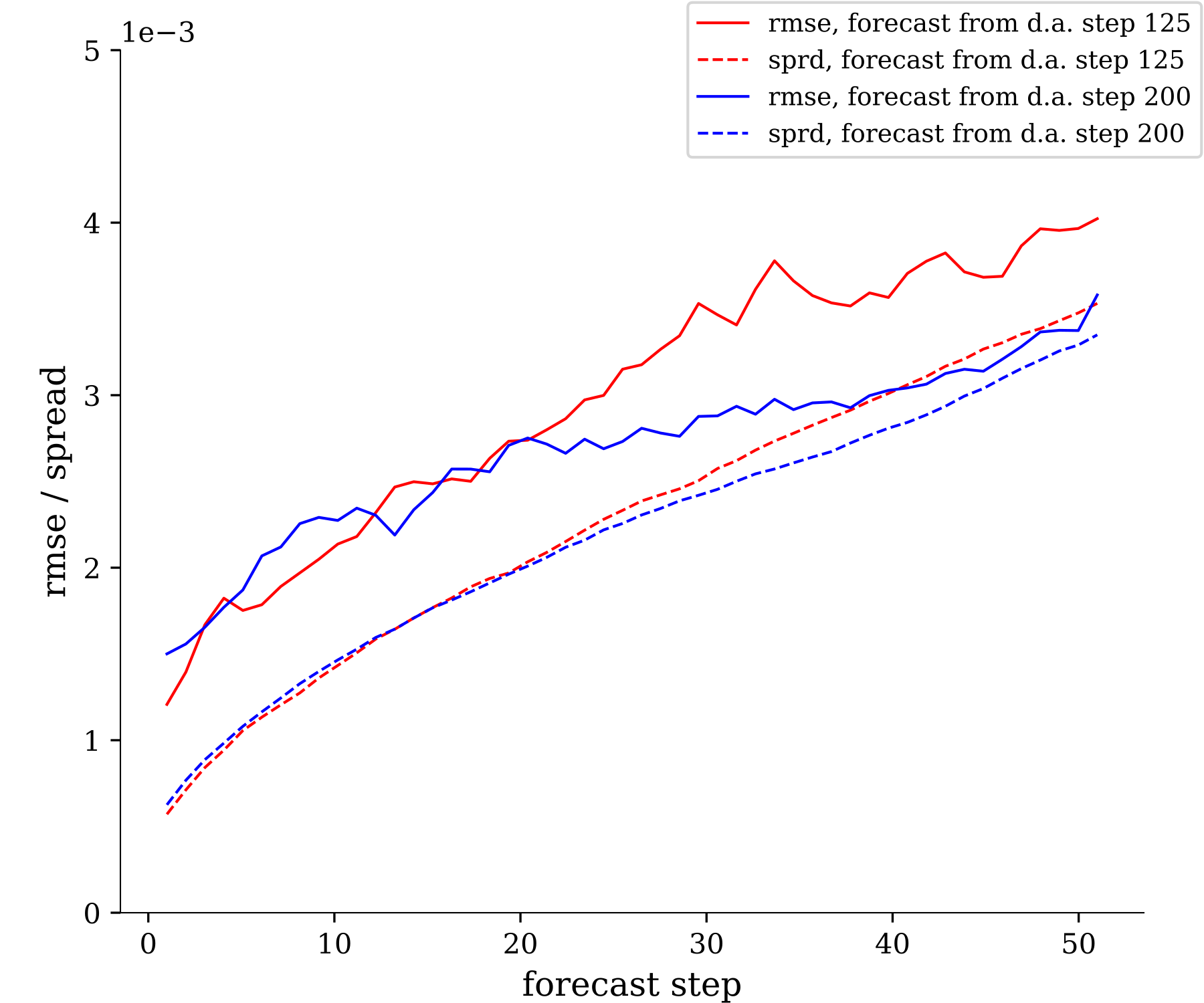}
	\caption{\MODIFY{Perfect model scenario (experiment using initial ensemble set 1). In this figure, we compare a 50 step forecast ensemble rmse with its spread, starting from two different posterior distributions. For red plots, we start from the posterior distribution at d.a. step $125$ and compute the forecast distribution for $50$ steps, i.e. $\pp_{125,j}$, for $j=1,\dots,50$, with $\pp_{125,0} = \ppi_{125}$. For the blue plots, we start from $\pp_{200,0} = \ppi_{200}$ and compute $\pp_{200,j}$ for $j=1,\dots,50$.} }
	\label{figure: rmse vs spread initial ensemble set 1}
	\end{minipage}
\end{figure}

\begin{figure}[!htbp]
	\centering
	\begin{minipage}{0.95\textwidth}
\begin{subfigure}{\textwidth}
		    \includegraphics[width=\textwidth]{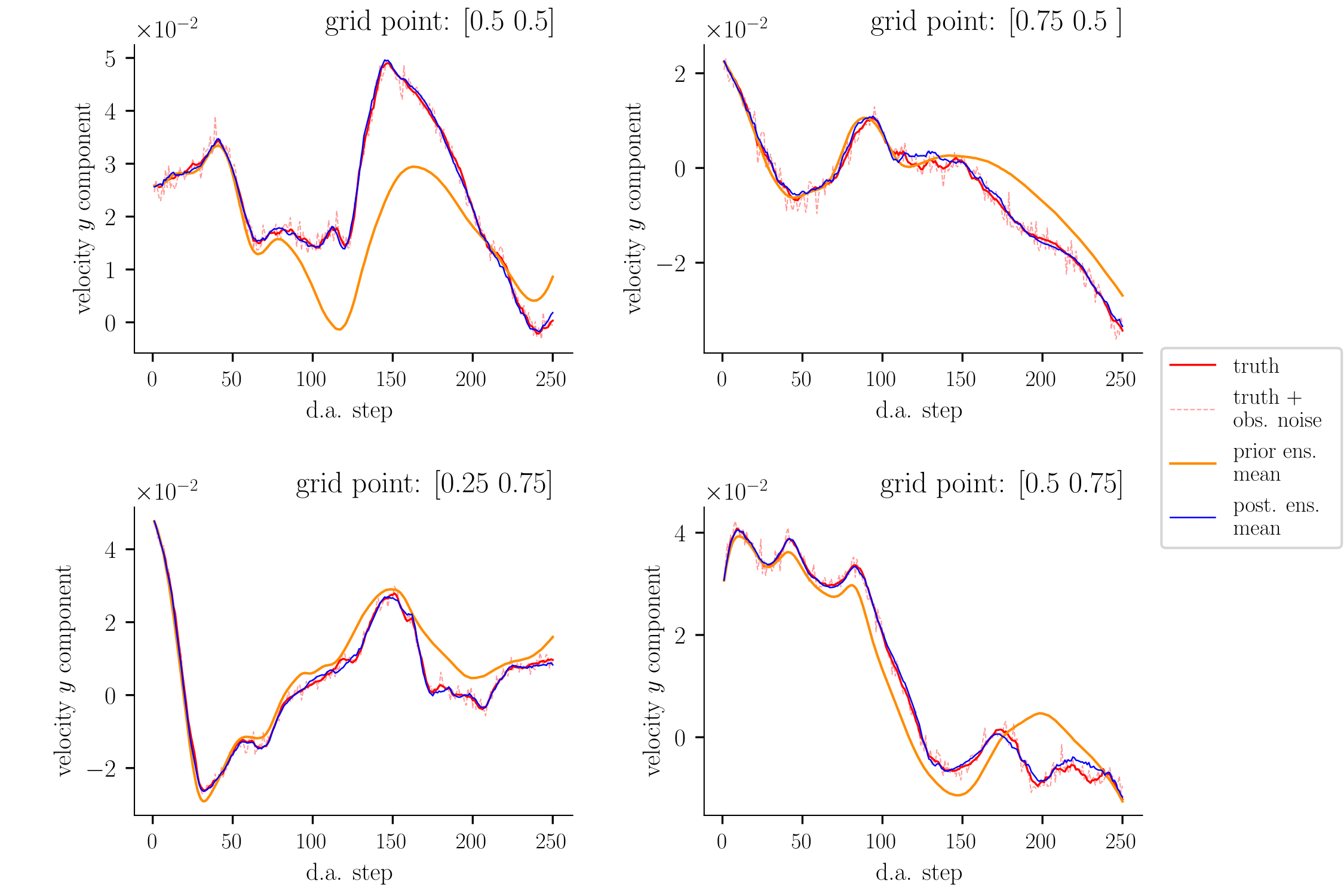}
    	    \caption{Ensemble mean trajectories} 
	        \label{figure:spde_trajectory initial ensemble 1}
	    \end{subfigure}
	    
	    \vspace{0.3cm}
	    
	    \begin{subfigure}{\textwidth}
    		\includegraphics[width=\textwidth]{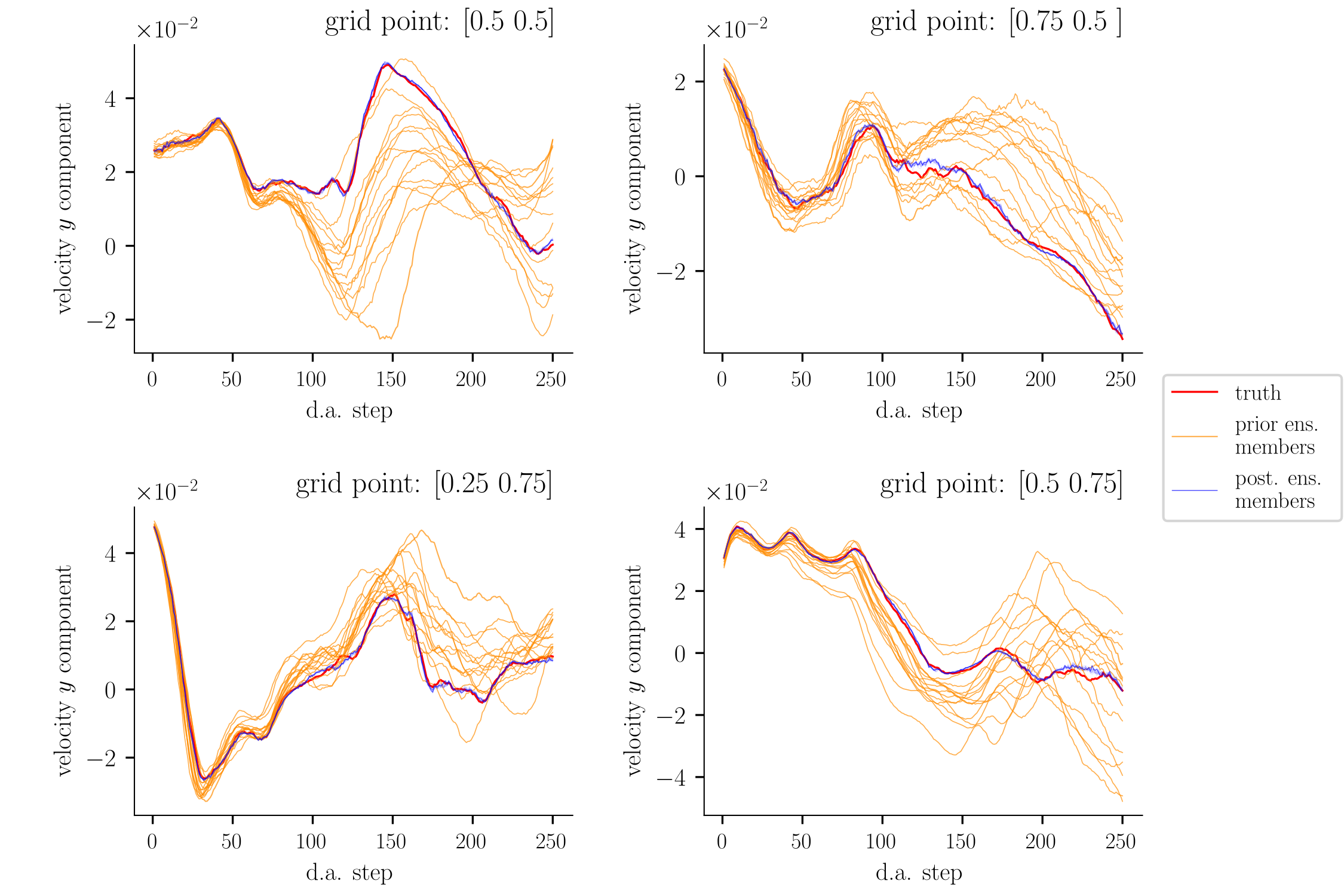}
    	    \caption{Ensemble member trajectories}
    	    \label{figure:spde spread initial ensemble 1}
	    \end{subfigure}

		\caption{\MODIFY{Perfect model scenario (experiment using initial ensemble set 1). In the first subfigure, we show the Eulerian trajectories of the truth (red), truth plus observation noise (dashed pink), prior ensemble mean (orange) and posterior ensemble mean (blue), at four grid points. In the second subfigure, Eulerian trajectories of 15 individual ensemble members are plotted.} }
		\label{figure:spde trajectory and spread initial ensemble 1}
	\end{minipage}
\end{figure}

\begin{figure}[!htbp]
	\centering
	\begin{minipage}{0.95\textwidth}
\begin{subfigure}{\textwidth}
		    \includegraphics[width=\textwidth]{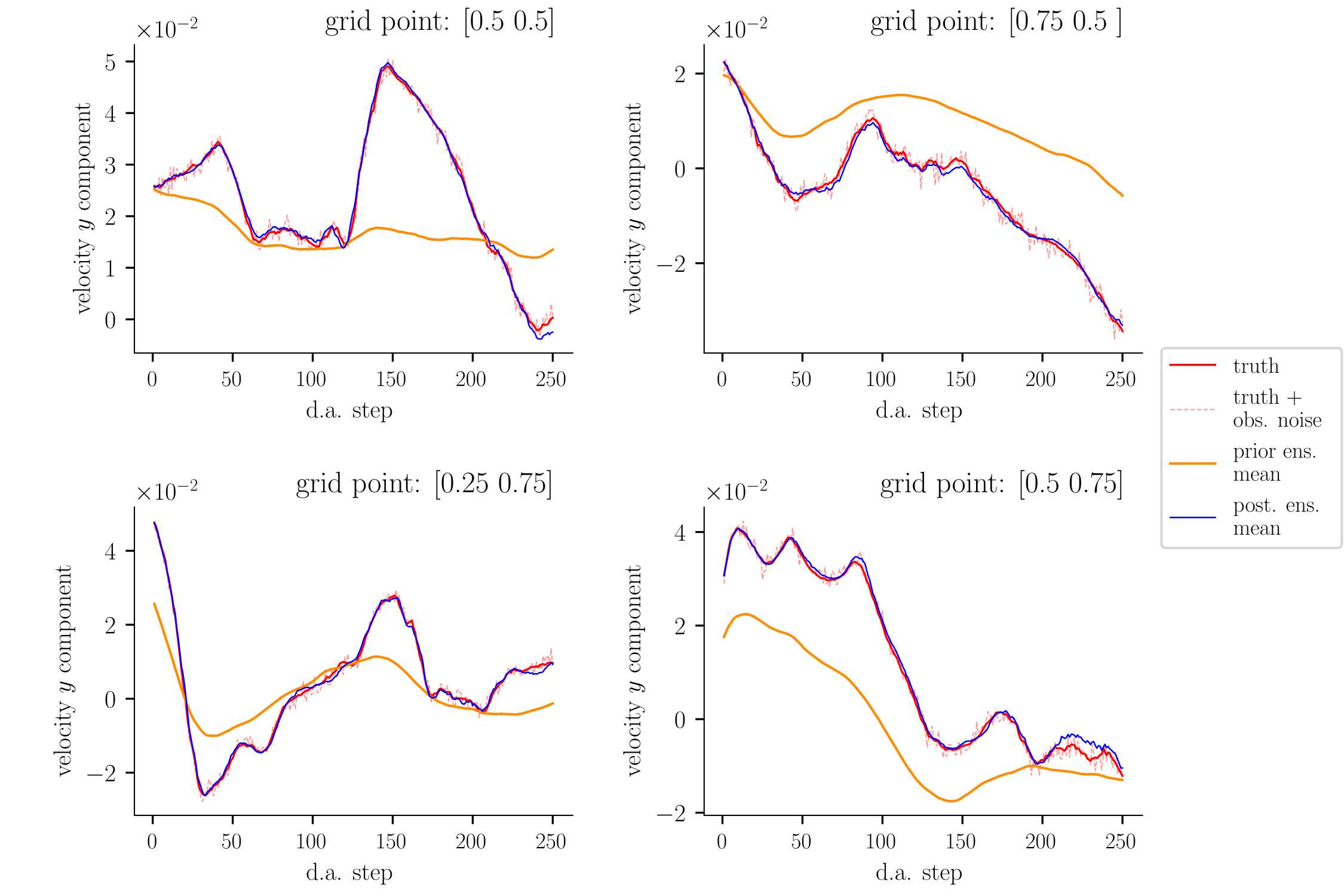}
    	    \caption{Ensemble mean trajectories} 
	        \label{figure:spde traj initial ensemble 2}
	    \end{subfigure}
	    
	    \vspace{0.3cm}
	    
	    \begin{subfigure}{\textwidth}
    		\includegraphics[width=\textwidth]{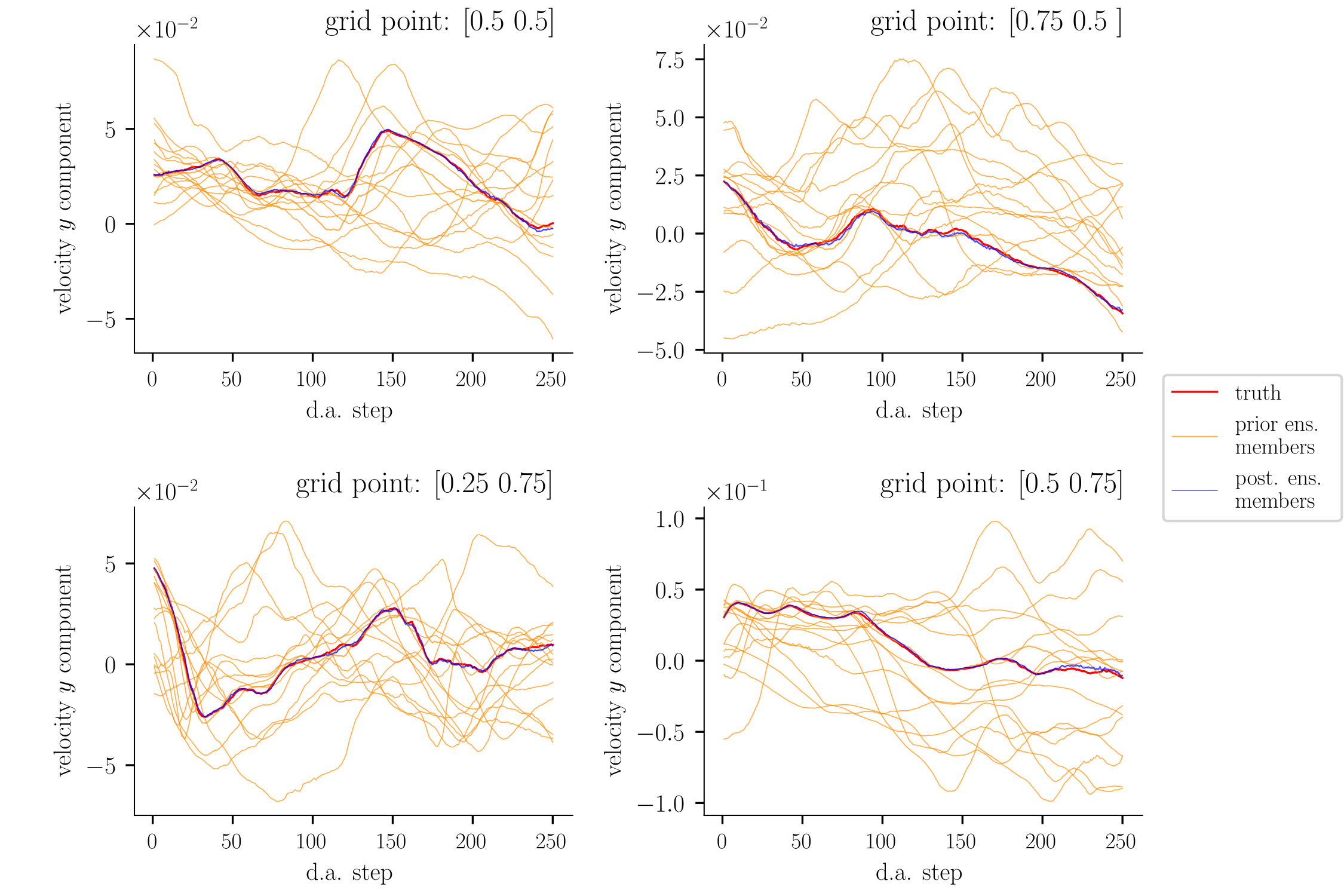}
    	    \caption{Ensemble member trajectories}
    	    \label{figure: spde spread initial ensemble 2}
	    \end{subfigure}
	\caption{\MODIFY{Perfect model scenario (experiment using initial ensemble set 2). In the first subfigure, we show the Eulerian trajectories of the truth (red), truth plus observation noise (dashed pink), prior ensemble mean (orange) and posterior ensemble mean (blue), at four grid points. In the second subfigure, Eulerian trajectories of 15 individual ensemble members are plotted.} }
	\label{figure:spde trajectory spread initial ensemble set 2}
	\end{minipage}
\end{figure}

\begin{figure}[!htbp]
	\centering
	\begin{minipage}{0.9\textwidth}
		\includegraphics[width=\textwidth]{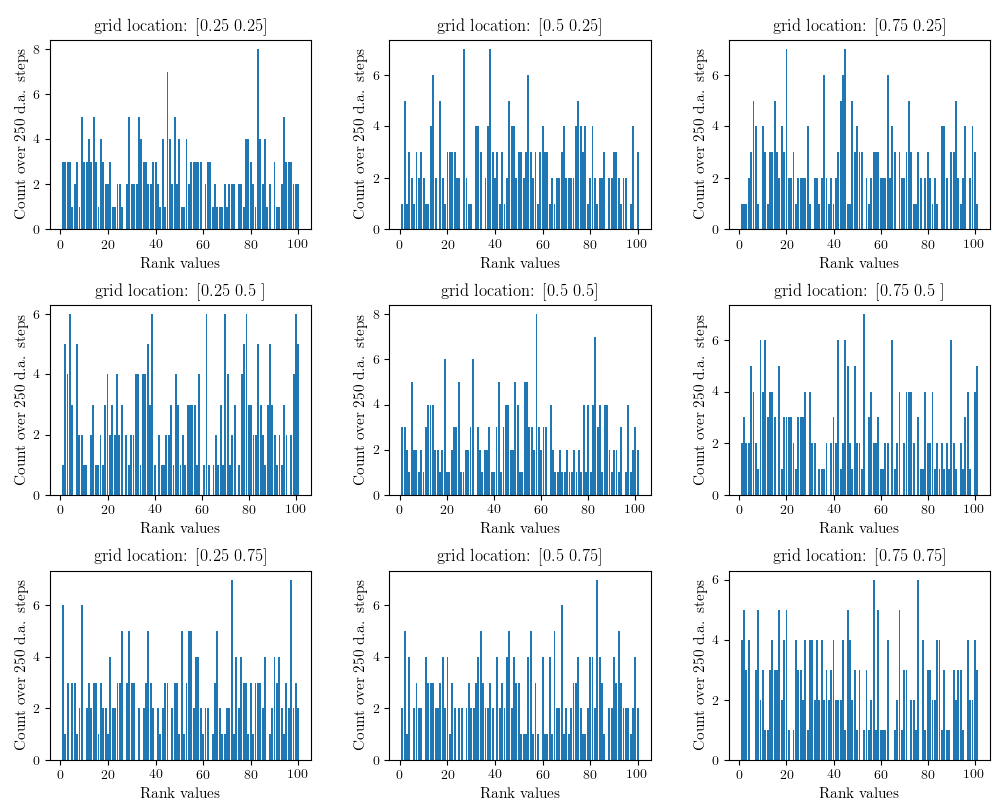}
		\caption{\MODIFY{Perfect model scenario (initial ensemble set 1). Forecast reliability rank histogram plots at nine grid locations, for a single run using the parameters: $100$ particles, assimilation period $\Delta = 1/25$ ett, observation noise scaling $\lambda=10$ and $289$ weather stations. Experiment period: $10$ ett. Grid locations are shown above the plots.}}
		\label{figure:spde_rankhistorgram_ux}
	\end{minipage}
\end{figure}

\begin{figure}[!htbp]
	\centering
	\begin{minipage}{0.9\textwidth}
		\includegraphics[width=\textwidth]{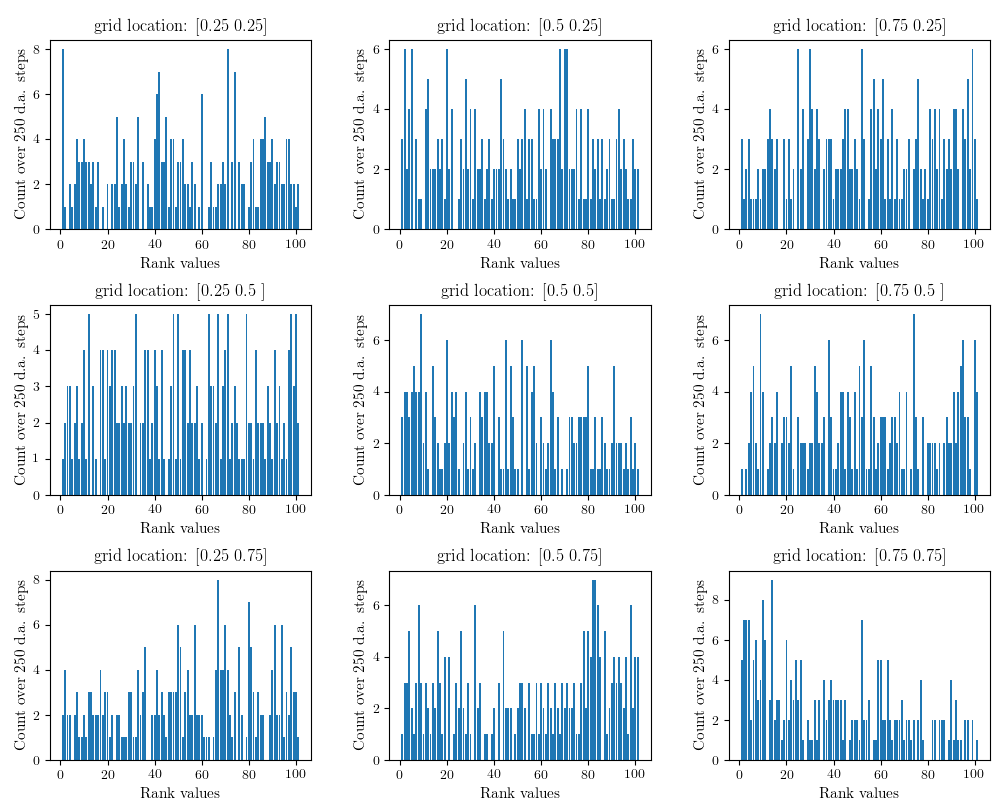}
		\caption{\MODIFY{Perfect model scenario (initial ensemble set 1). Velocity x-component rank histogram plots at nine grid locations, for a single run using the parameters: $100$ particles, assimilation period $\Delta = 1/25$ ett, observation noise scaling $\lambda=10$ and $81$ weather stations. Experiment period: $10$ ett. Grid locations are indicated above the plots.}}
		\label{figure:spde_rankhistorgram_ux_msh8}
	\end{minipage}		
\end{figure}

\MODIFY{A single realisation of the SPDE was used as the truth for the experiments in this scenario. The data assimilation experiments are defined by the following parameters: time interval between assimilations $\Delta=0.04$ ett (every $5$ coarse time steps), number of weather stations $d_y=289$, observation error scaling $\lambda=0.6$.
We ran the same experiment setup for each of the two initial ensembles (see table \ref{tab:initial_ensemble}), using ensemble size $100$ and for a total experiment period of $10$ ett.
Note that $10$ ett is equivalent to $1250$ coarse resolution time steps. For our assimilation interval choice, $10$ ett amounts to $250$ data assimilation steps.
The experiments were run independently of each other.}

\MODIFY{This scenario serves as an important test case for the filtering algorithm because there is no discrepancy between the model and the true state. We want to see a stable rmse error between the posterior ensemble mean and the truth. This is an important indicator to show that the filter does not lose track of the signal over the experiment period. If the results do not show this, then it would be very unlikely that the filtering algorithm can be made to work  with the PDE to SPDE model reduction.}

\MODIFY{The left subplot in figure \ref{figure: rmse and spread initial ensemble set 1} shows comparisons of the rmse between the posterior ensemble mean and the true state (in blue) \[\rmse(\bar{\ppi}_{i}, X^\dagger_i),\] the rmse between the one step forecast ensemble mean and the true state (in green) \[\rmse(\bar{\pp}_{i,1},X^\dagger_i),\] the rmse between the prior ensemble mean and the true state (in orange) \[\rmse(\bar{\qq}_i, X_i^\dagger),\] and lastly the rmse between the one step forecast ensemble mean and the true state plus observation noise (in red) \[\rmse(\bar{\pp}_{i,1},X_i^\dagger + \epsilon_i).\] 
In figure \ref{figure: rmse and spread initial ensemble set 2}, the rmse subplot shows comparisons of $\rmse(\bar{\pp}_{i,1},X^\dagger_i)$ (in green), $\rmse(\bar{\pp}_{i,1},X^\dagger_i)$ (in orange) and
$\rmse(\bar{\pp}_{i,1},X_i^\dagger + \epsilon_i)$ (in red).
}

\MODIFY{In both rmse subplots of figure \ref{figure: rmse and spread initial ensemble set 1} and figure \ref{figure: rmse and spread initial ensemble set 2}, the rmse between the one step forecast ensemble mean and the true state plus observation noise $\rmse(\bar{\pp}_{i,1},X_i^\dagger + \epsilon_i)$ (in red) is \emph{stable}. It is (roughly) a few times larger than the rmse between the one step forecast ensemble mean and the truth without observation noise $\rmse(\bar{\pp}_{i,1},X_i^\dagger)$ (in green). This suggests that the size of the observation noise is dominating. 
Because of this, we treat the positive trend in $\rmse(\bar{\ppi}_i, X^\dagger_i)$ (in blue) and in $\rmse(\bar{\pp}_{i,1},X_i^\dagger)$ (in green) as minimal -- the positive trend is within the ``accuracy tolerance" measured by the likelihood function. Thus  $\rmse(\bar{\ppi}_i, X^\dagger_i)$ is sufficiently stable.
This is further supported by a comparison with the increase in the rmse between the prior ensemble mean and the truth (in orange).  Therefore, the data we assimilated, albeit low dimensional relative to the SPDE degree of freedom, gives sufficient information to be able to offer a reasonably accurate approximation of the signal.}

\MODIFY{Another feature to note in the rmse subplots in figure \ref{figure: rmse and spread initial ensemble set 1} is, the rmse between the one step forecast ensemble mean and the truth (in green) is slightly larger than the rmse between the posterior ensemble mean and the truth (in blue). 
This feature is due to the resampling step at each assimilation time. 
The same reason applies to the differences between the ensemble spreads of the posterior (in blue) and one step forecast (in green), shown in the right subplot in figures \ref{figure: rmse and spread initial ensemble set 1} and \ref{figure: rmse and spread initial ensemble set 2}. 
One can also see that the posterior ensemble spreads are reasonably stable. However, in the absence of data assimilation corrections, for the prior distribution (in orange) we see a continuous increase in both the rmse and spread.}

\MODIFY{Figure \ref{figure:spde_statistics} compares the effect of the two initial ensembles on the experiments.
In the figure, the four subplots compare $\rmse(\bar{\ppi}_i, X^\dagger_i)$, $\sprd(\ppi_i)$, number of SPDE computations and ess. The two different experiments produced very close results. }

\MODIFY{As shown in table \ref{tab:initial_ensemble}, the two initial ensembles produce very different initial rmse and ensemble spread. The rmse values of initial ensemble 1's ensemble mean are two orders of magnitude smaller than initial ensemble 2's rmse values. The ensemble spread of initial ensemble 1 is also one order of magnitude smaller than the ensemble spread of initial ensemble 2. Despite these relatively large initial differences, after one data assimilation step is completed, the spread and rmse of the corresponding ensembles become comparable. The reason is that all unlikely particles are immediately eliminated whilst the diversity of the ensemble is kept high through the tempering procedure. This is also visualised in figure \ref{figure: rmse and spread initial ensemble set 2}, indicated by the initial sharp decrease in the one step forecast rmse, $\rmse(\bar{\pp}_{i,1},X^\dagger_i)$, and one step forecast ensemble spread, $\sprd(\pp_{i,1})$ (both plotted in green).}

\MODIFY{The bottom left subfigure of figure \ref{figure:spde_statistics} shows the amount of computation taken at each assimilation time, measured in terms of the number of SPDE evaluations. The values used to obtain these plots can be accurately estimated from the number of tempering steps. 
For each tempering step, we have to solve $100$ SPDEs (number of ensemble members), followed by a fixed number of jittering steps (we use a fixed value of $5$ jittering steps, see step 4 of algorithm \ref{alg:mcmc_jittering}) for each duplicate resampled ensemble member. We assume a duplicate rate of $30\%$\footnote{The duplicate rate is more or less the average number of duplicates per assimilation step from our numerical experiments.}. Thus the computational cost in terms of number of SPDE evaluations can be estimated by
\[
\text{(no. of tempering steps)}\times(N+5*30\% N).
\]}

\MODIFY{Finally, the ess value in figure \ref{figure:spde_statistics} shows that the tempering procedure is successful in keeping the ess values near the chosen threshold of $80\%$.}

\MODIFY{The uncertainty quantification results in \cite{cotter2018numerically} focused on the prior ensemble. In figure \ref{figure: rmse vs spread initial ensemble set 1}, we compare forecast rmse with forecast ensemble spread, i.e.
\[
\rmse(\bar{\pp}_{i,j}, X_j^\dagger)
\quad\text{ with }\quad
\sprd(\pp_{i,j})\quad j=1,\dots,50, \quad i=125, 200
\]
taking $\pp_{i,0} = \ppi_{i}$. The plots show that the forecast rmse and forecast ensemble spread are comparable. Further, the difference between corresponding rmse and spread, whether starting with $\ppi_{125}$ or $\ppi_{200}$,
are more or less equal. 
Both features indicate the filter is keeping track of the signal. Otherwise it is unlikely that the difference between forecast rmse and forecast spread is maintained starting from different posterior distributions. Note that there are 75 assimilation steps between $\ppi_{125}$ and $\ppi_{200}$, which amounts to $3$ ett.}

\MODIFY{In figures \ref{figure:spde trajectory and spread initial ensemble 1} and \ref{figure:spde trajectory spread initial ensemble set 2} we show the Eulerian trajectories of the velocity y-component at four spatial locations. Figure \ref{figure:spde trajectory and spread initial ensemble 1} corresponds to the experiment using initial ensemble 1. Figure \ref{figure:spde trajectory spread initial ensemble set 2} corresponds to the experiment using initial ensemble 2. In each figure, the two subfigures correspond to the same experiment at the same grid locations.}

\MODIFY{In the subfigures \ref{figure:spde_trajectory initial ensemble 1} and \ref{figure:spde traj initial ensemble 2}, we plot the true state (in red), the true state plus observation noise (in dashed pink), the posterior ensemble mean (in blue) and the prior ensemble mean (in orange). 
Since initial ensemble 1 start very close to the initial truth, we see in subfigure \ref{figure:spde_trajectory initial ensemble 1} that the prior ensemble mean's initial deviation from the truth is small. But the deviation become much more pronounced after assimilation step 100 at grid location $[0.5, 0.5]$. The posterior ensemble mean (in blue) stays close to the observed truth (in pink) at all four grid locations. In subfigure \ref{figure:spde spread initial ensemble 1}, trajectory of individual ensemble members are plotted. It shows how the prior ensemble spread increases as time goes on, but the posterior ensemble spread seems stable. This supports the features shown in figures \ref{figure:spde_statistics}, \ref{figure: rmse and spread initial ensemble set 1} and \ref{figure: rmse and spread initial ensemble set 2}. In this scenario though, because there is no model error, we see in subfigure \ref{figure:spde spread initial ensemble 1} that the truth does not deviate from the spread of the prior ensemble. The assimilated data allowed the posterior ensemble to offer a reasonably accurate approximation of the truth, whilst reducing the approximation uncertainty.}

\MODIFY{Initial ensemble 2 start alot farther from the initial truth compared to initial ensemble 1. In subfigures \ref{figure:spde traj initial ensemble 2} and \ref{figure: spde spread initial ensemble 2}, we see that the prior ensemble mean deviates from the truth more greatly than shown in subfigures \ref{figure:spde_trajectory initial ensemble 1} and \ref{figure:spde spread initial ensemble 1}. This is further evidence to support the observation that the filtering algorithm was able to eliminate the unlikely particle positions whilst maintaining ensemble diversity, to reasonably approximate the truth.}

\MODIFY{Lastly, in figures \ref{figure:spde_rankhistorgram_ux} and \ref{figure:spde_rankhistorgram_ux_msh8} we show rank histogram plots of the Eulerian velocity x-component, at nine grid locations. Figure \ref{figure:spde_rankhistorgram_ux} corresponds to the experiment using initial ensemble 1, $289$ observations, assimilation period $0.04$ ett and observation noise scaling $\lambda=10$.
The plots do not show features of strong bias, under-dispersion or over-dispersion over the experiment period of $10$ ett. 
Figure \ref{figure:spde_rankhistorgram_ux_msh8} corresponds to the same repeated experiment but using a fewer number of weather stations ($81$ observations). We observe that, assimilating less data leads to more pronounced features of skew, e.g. at grid locations $[0.75, 0.75]$, $[0.25, 0.75]$ and $[0.5, 0.5]$.
}

\subsection{Imperfect model scenario}\label{sec:imperfectmodelresults}

\begin{figure}[!htbp]
	\centering
	\begin{minipage}{0.9\textwidth}
		\includegraphics[width=\textwidth]{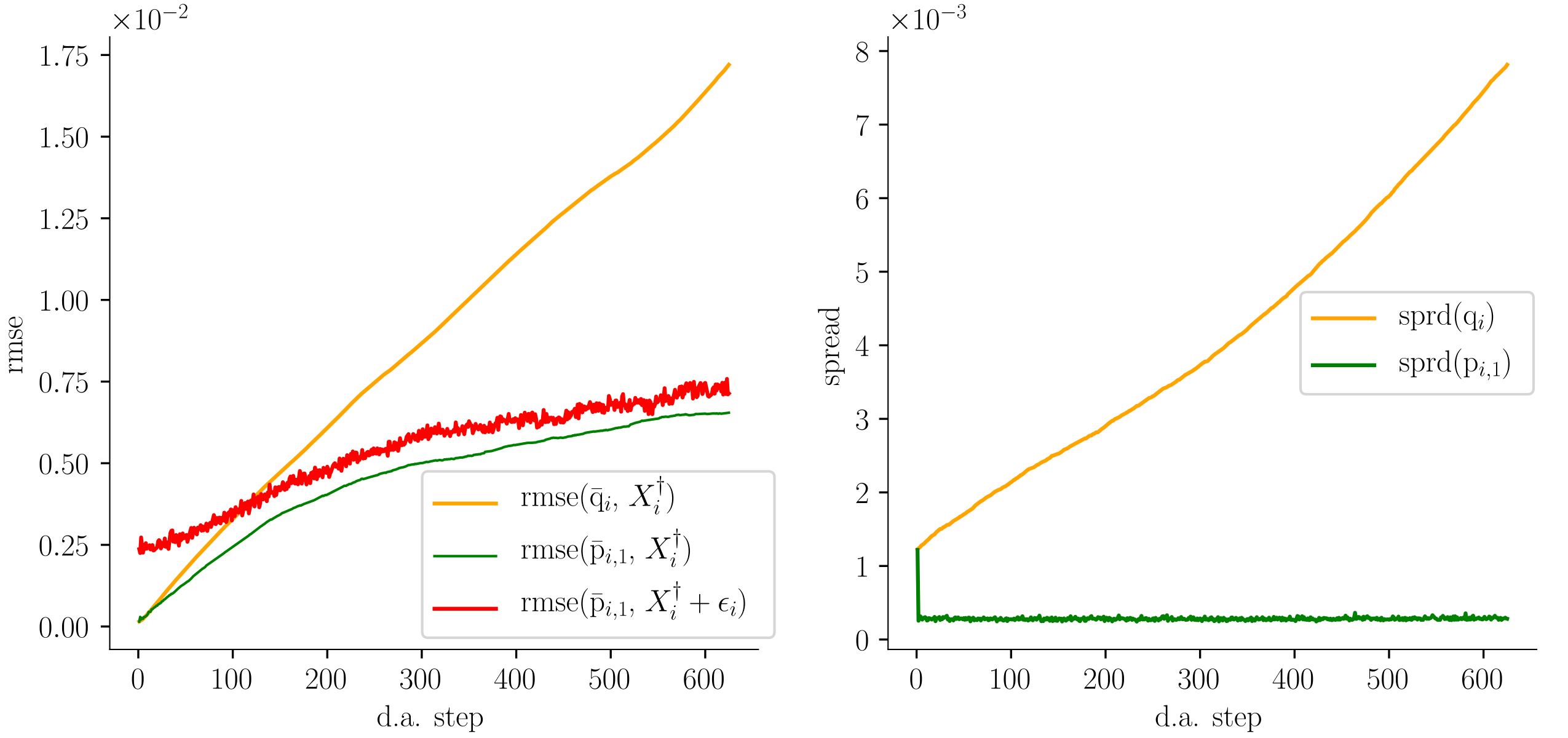}
		\caption{\MODIFY{Imperfect model scenario (experiment using initial ensemble set 1, and $289$ observations). The left figure compares rmse values and the right figure compares the ensemble spreads. In the rmse figure, the green plot shows the rmse between the one step forecast ensemble mean and the true state; the red plot shows the rmse between the one step forecast ensemble mean and the true state plus observation noise; the orange plot shows the rmse between the prior ensemble mean and the true state. In the spread figure, the green plot shows the spread of one step forecast ensemble; the orange plot shows the spread of the prior distribution ensemble.} }
		\label{figure: pde rmse and spread initial ensemble set 1}
	\end{minipage}
\end{figure}

\begin{figure}[!htbp]
	\centering
	\begin{minipage}{0.9\textwidth}
		\includegraphics[width=\textwidth]{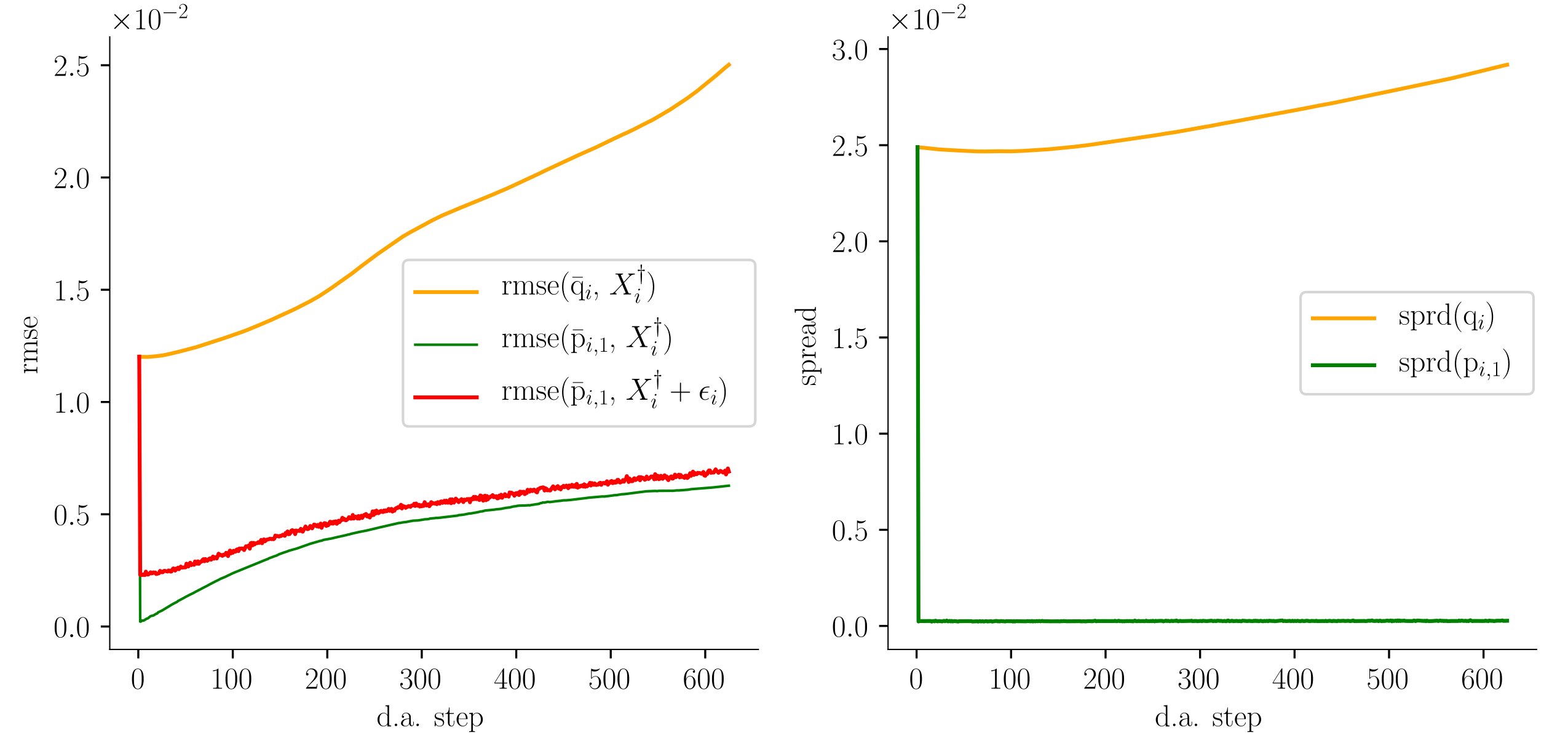}
		\caption{\MODIFY{Imperfect model scenario (experiment using initial ensemble set 2, and $1089$ observations). The left figure compares rmse values and the right figure compares the ensemble spreads. In the rmse figure, the green plot shows the rmse between the one step forecast ensemble mean and the true state; the red plot shows the rmse between the one step forecast ensemble mean and the true state plus observation noise; the orange plot shows the rmse between the prior ensemble mean and the true state. In the spread figure, the green plot shows the spread of one step forecast ensemble; the orange plot shows the spread of the prior distribution ensemble.} }
		\label{figure: pde rmse and spread initial ensemble set 2}
	\end{minipage}
\end{figure}

\begin{figure}[!htbp]
	\centering
	\begin{minipage}{0.9\textwidth}
		\includegraphics[width=\textwidth]{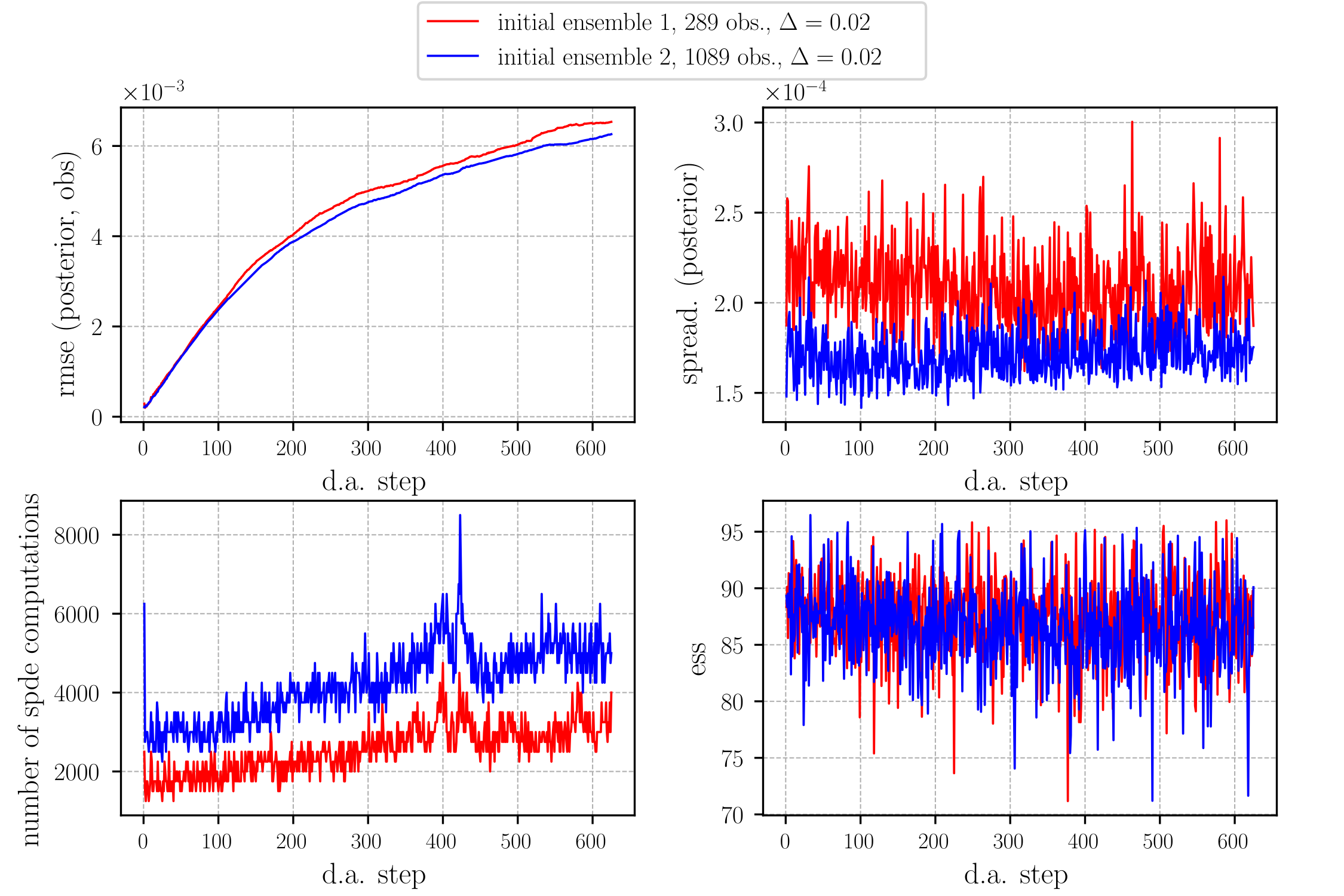}
		\caption{\MODIFY{Imperfect model scenario. The four figures show the root mean square error between the posterior ensemble mean and true state, ensemble spread of the posterior, number of SPDE evaluations at each assimilation step and ess. The red plots corresponds to the experiment using initial ensemble set 1 and $289$ observations. The blue plots corresponds to the experiment using initial ensemble set 2 and $1028$ observations. For both experiments the assimilation interval $\Delta = 0.08$ ett ($0.02$ time units) was used. Both experiments were run for a total of $5$ ett, which amounted to $625$ data assimilation steps.} }
		\label{figure:pde_statistics}
	\end{minipage}
\end{figure}

\begin{figure}[!htbp]
	\centering
	\begin{minipage}{0.9\textwidth}
		\includegraphics[width=\textwidth]{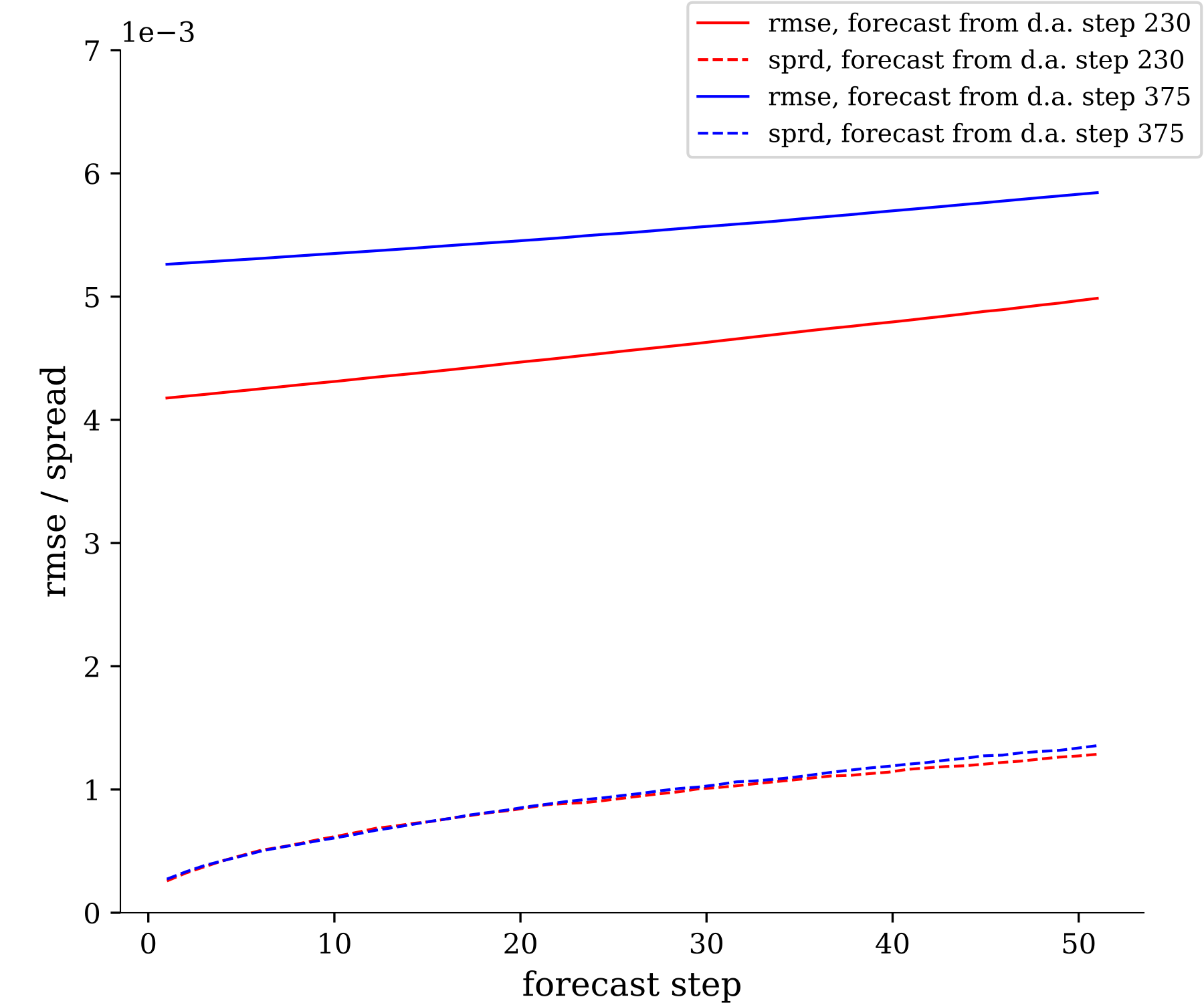}
		\caption{\MODIFY{Imperfect model scenario (experiment using initial ensemble set 1, and $289$ observations). In this figure, we compare a 50 step forecast rmse with ensemble spread starting from two different posterior distributions. For the red plots, we start from the posterior distribution at d.a. step $125$ and compute the forecast distribution for $50$ steps, i.e. $\pp_{125,j}$, for $j=1,\dots,50$, with $\pp_{125,0} = \ppi_{125}$. For the blue plots, we start from $\pp_{200,0} = \ppi_{200}$ and compute $\pp_{200,j}$ for $j=1,\dots,50$.} }
		\label{figure: pde rmse vs spread initial ensemble set 1}
	\end{minipage}
\end{figure}

\begin{figure}[!htbp]
	\centering
	\begin{minipage}{0.95\textwidth}
		\begin{subfigure}{\textwidth}
			\includegraphics[width=\textwidth]{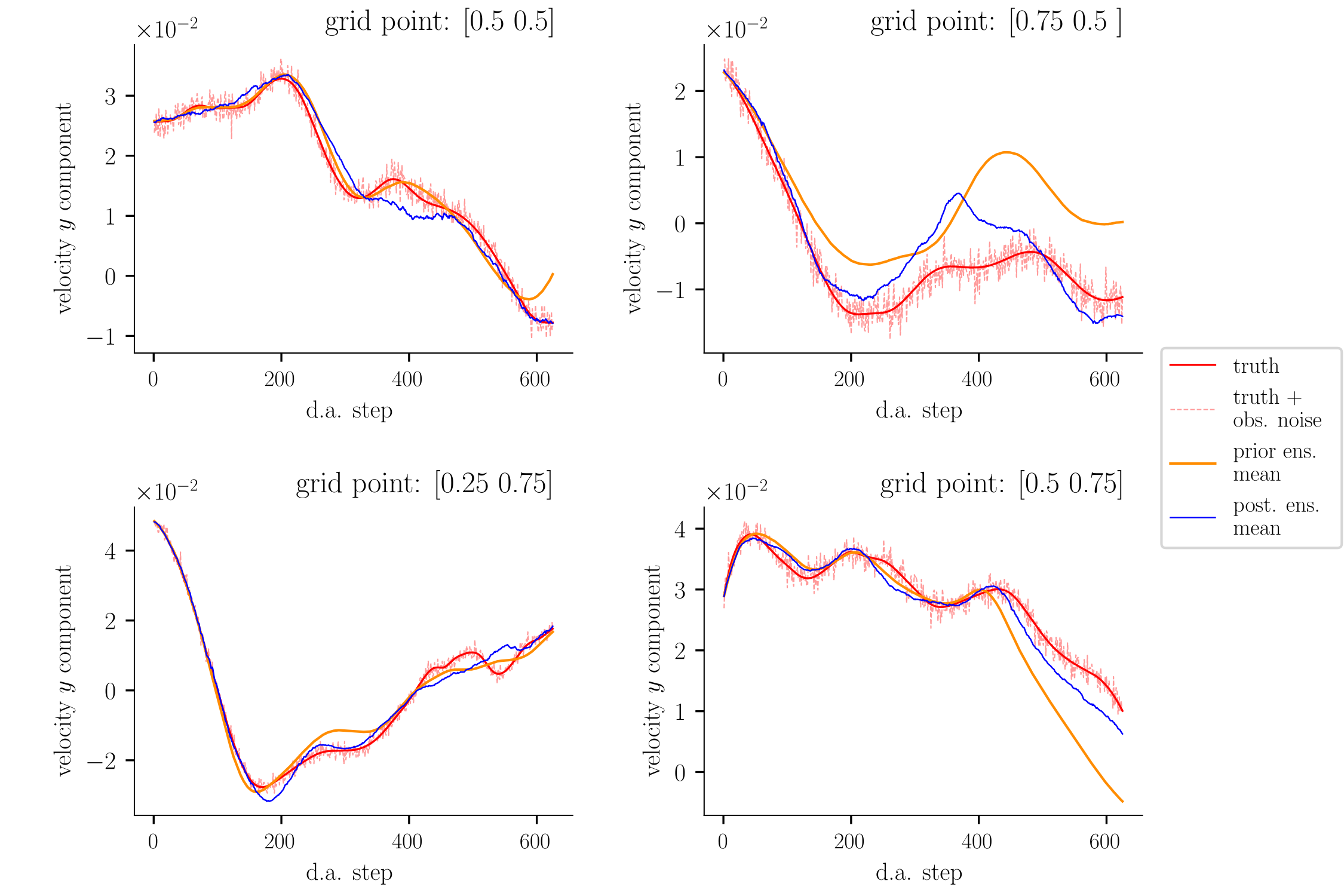}
			\caption{Ensemble mean trajectories} 
			\label{figure:pde traj initial ensemble 1}
		\end{subfigure}
		
		\vspace{0.3cm}
		
		\begin{subfigure}{\textwidth}
			\includegraphics[width=\textwidth]{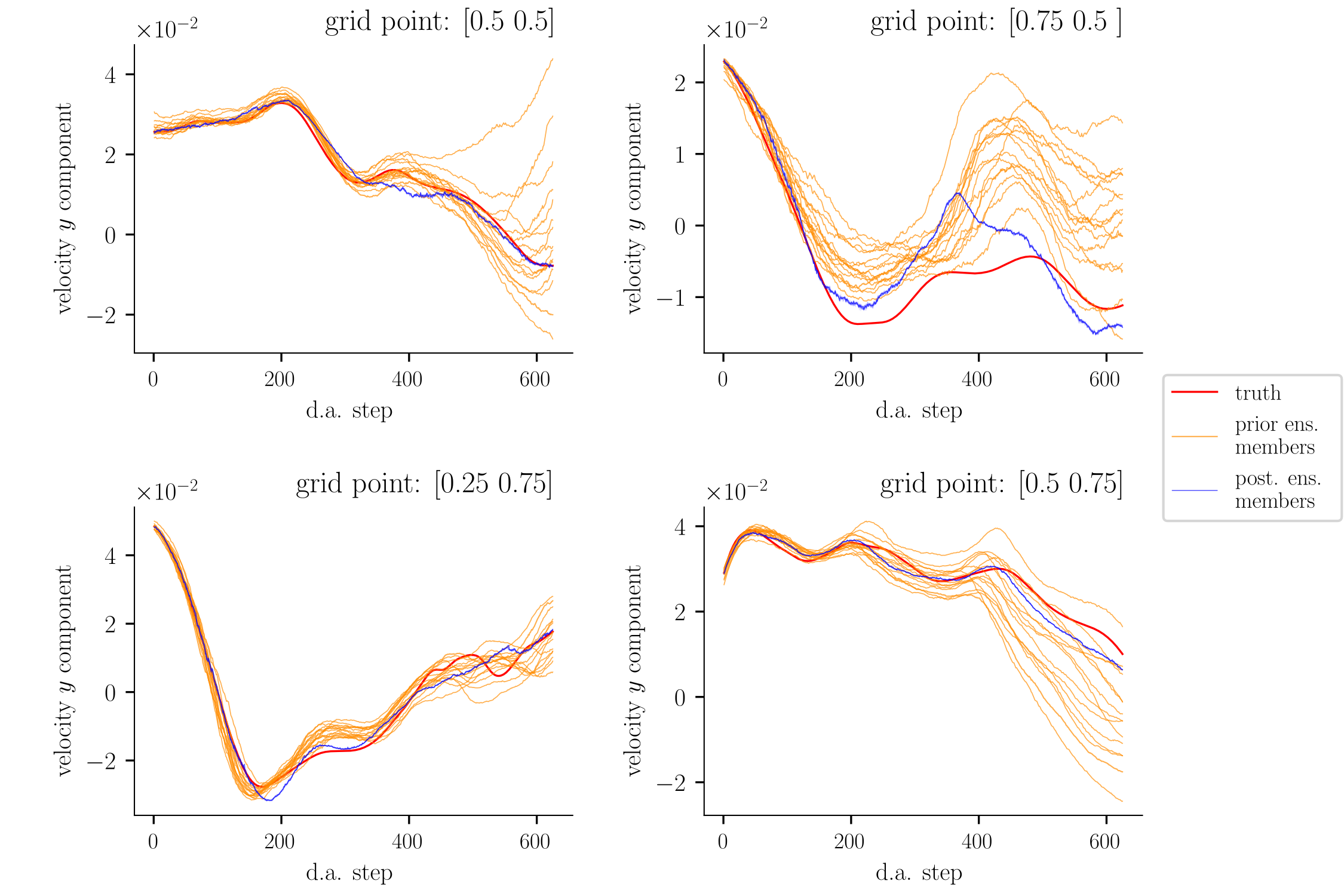}
			\caption{Ensemble member trajectories}
			\label{figure:pde spread initial ensemble 1}
		\end{subfigure}
		
		\caption{\MODIFY{Imperfect model scenario (experiment using initial ensemble set 1, and $289$ observations). In the first subplot, we show the Eulerian trajectories of the truth (red), truth plus observation noise (dashed pink), prior ensemble mean (orange) and posterior ensemble mean (blue). In the second subplot, Eulerian trajectories of 15 individual ensemble members are plotted.} }
		\label{figure:pde trajectory and spread initial ensemble 1}
	\end{minipage}
\end{figure}

\begin{figure}[!htbp]
	\centering
	\begin{minipage}{0.95\textwidth}
		\begin{subfigure}{\textwidth}
			\includegraphics[width=\textwidth]{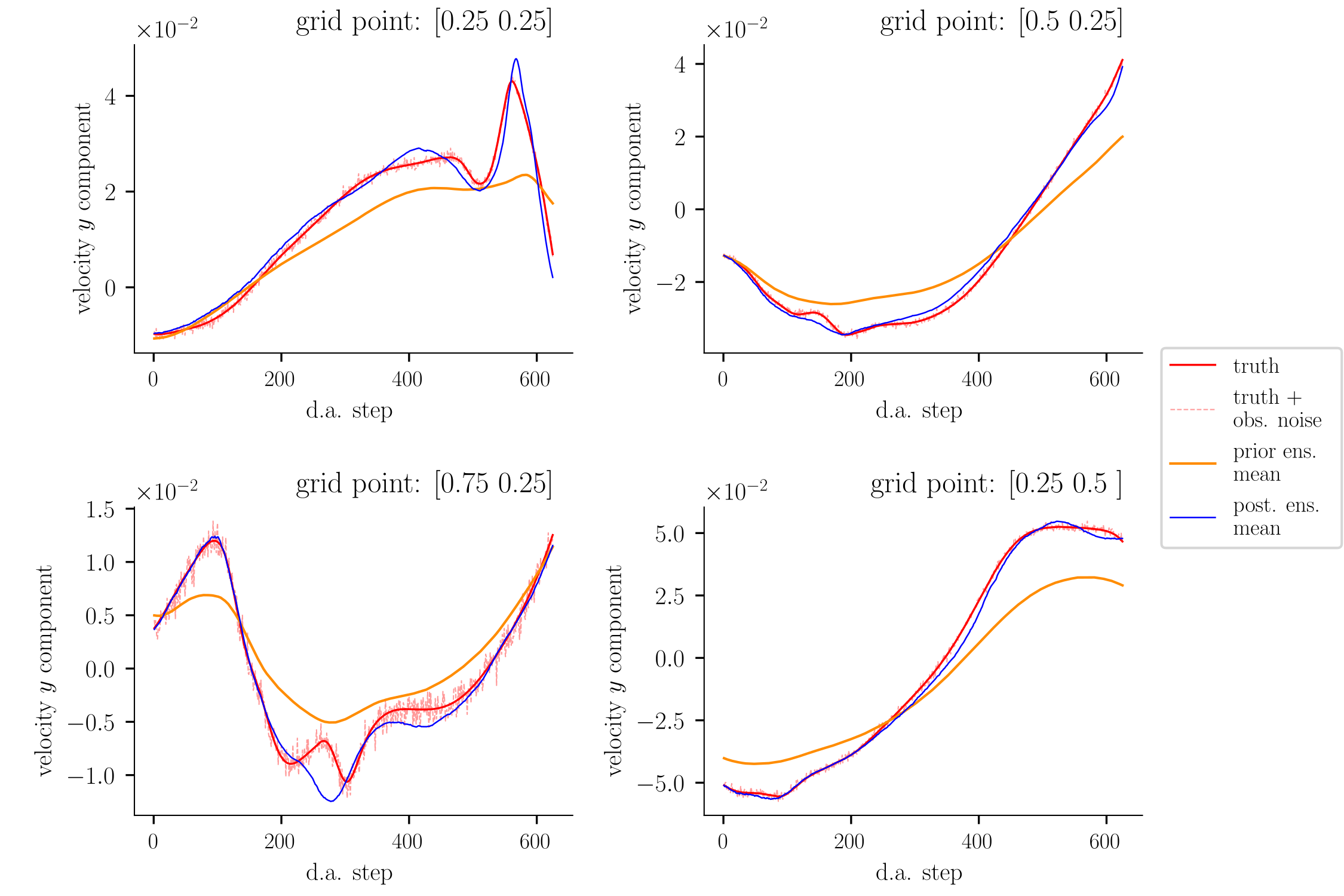}
			\caption{Ensemble mean trajectories} 
			\label{figure:pde traj initial ensemble 2}
		\end{subfigure}
		
		\vspace{0.3cm}
		
		\begin{subfigure}{\textwidth}
			\includegraphics[width=\textwidth]{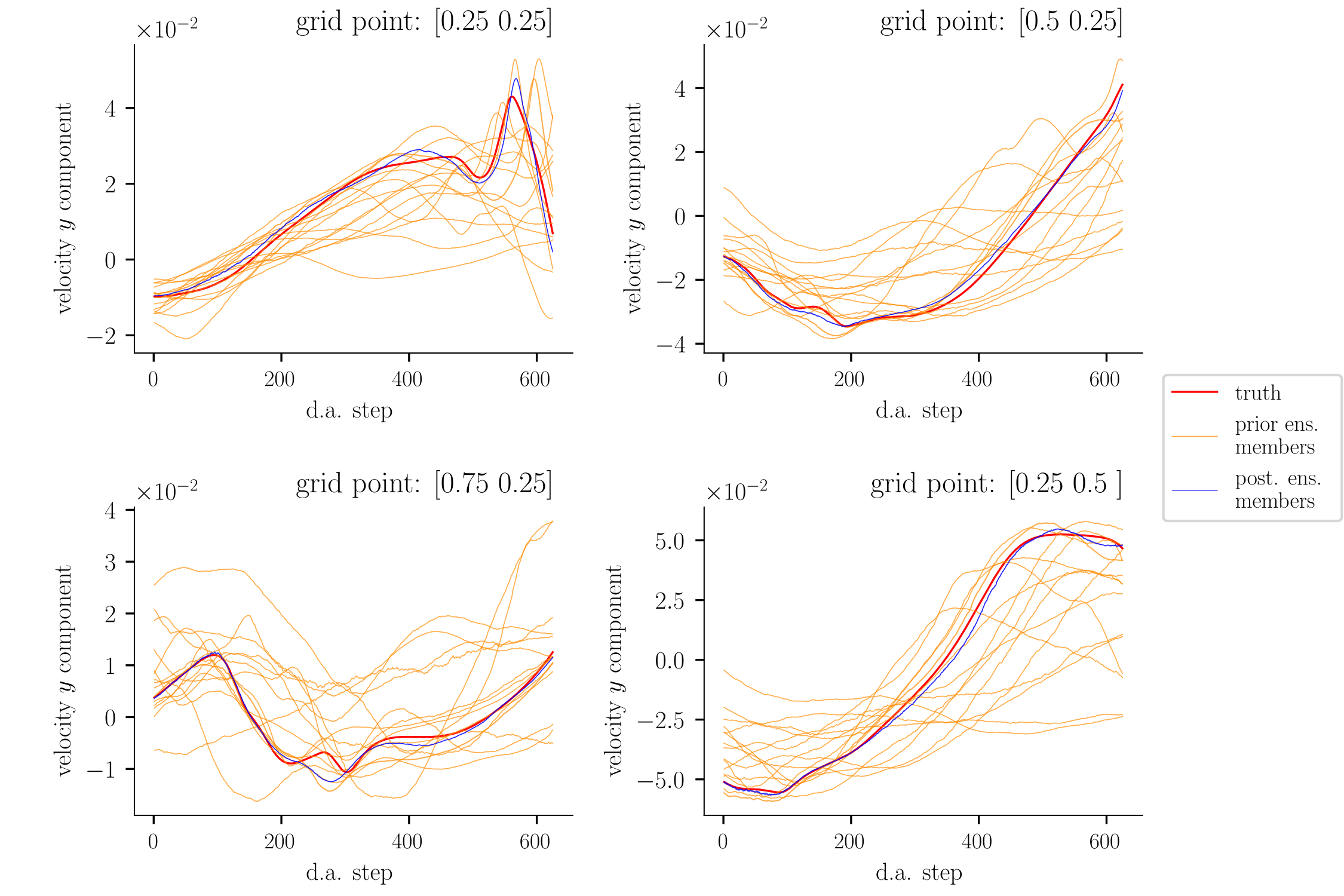}
			\caption{Ensemble member trajectories}
			\label{figure: pde spread initial ensemble 2}
		\end{subfigure}
		\caption{\MODIFY{Imperfect model scenario (experiment using initial ensemble set 2, and $1089$ observations). In the first subplot, we show the Eulerian trajectories of the truth (red), truth plus observation noise (dashed pink), prior ensemble mean (orange) and posterior ensemble mean (blue). In the second subplot, Eulerian trajectories of 15 individual ensemble members are plotted.} }
		\label{figure:pde trajectory spread initial ensemble set 2}
	\end{minipage}
\end{figure}

\begin{figure}[!htbp]
	\centering
	\begin{minipage}{\textwidth}
		\includegraphics[width=\textwidth]{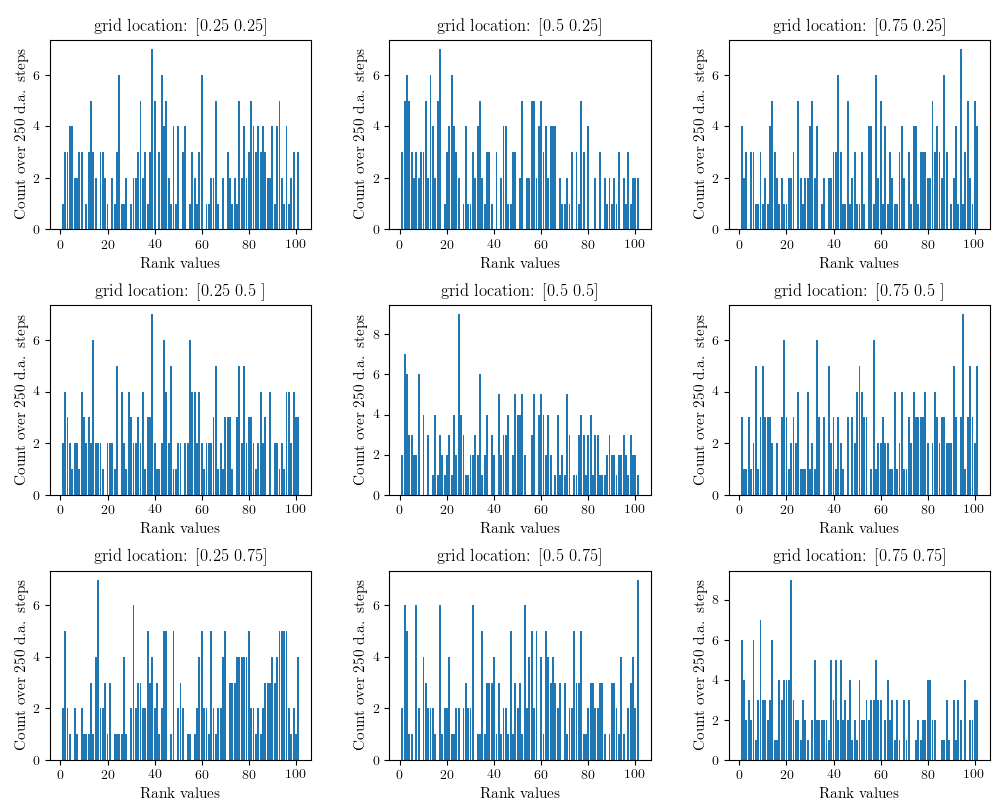}
		\caption{\MODIFY{Imperfect model scenario (using initial ensemble set 1). Forecast reliability rank histogram plots at nine grid locations, for a single run using the parameters: $100$ particles, assimilation period $\Delta = 0.04$ ett, observation noise scaling $\lambda=10$ and $289$ weather stations. Experiment period: $10$ ett. Grid locations are shown above the plots.}}
		\label{figure:pde_rankhistorgram_ux}
	\end{minipage}
\end{figure}

\begin{figure}[!htbp]
	\centering
	\begin{minipage}{0.9\textwidth}
		\includegraphics[width=\textwidth]{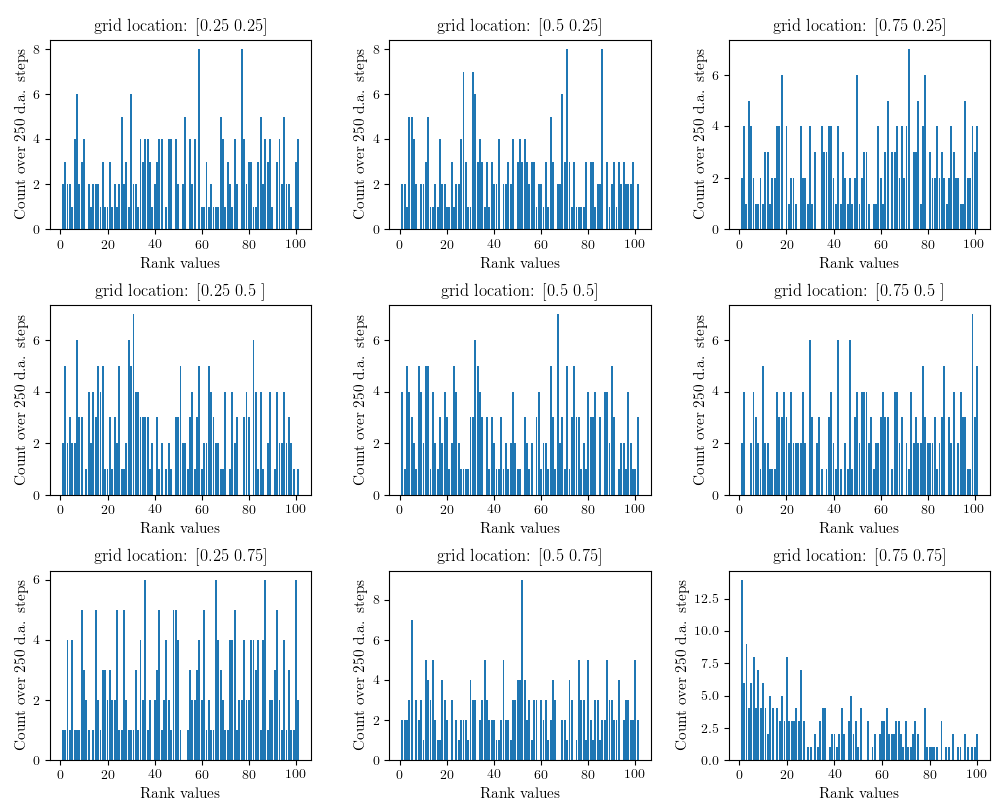}
		\caption{\MODIFY{Imperfect model scenario (using initial ensemble set 1). Velocity x-component rank histogram plots at nine grid locations, for a single run using the parameters: $100$ particles, assimilation period $\Delta = 0.04$ ett, observation noise scaling $\lambda=10$ and $81$ weather stations. Experiment period: $10$ ett. Grid locations are indicated above the plots.}}
		\label{figure:pde_rankhistorgram_ux_msh8}
	\end{minipage}		
\end{figure}

\MODIFY{The coarse grained fine resolution PDE solution is used as the true state in this experiment scenario, see remark \ref{remark: coarse graining}. As explained in earlier sections, the calibrated SPDE is a result of model reduction applied to the fine resolution PDE. Given the adequate results shown in section \ref{sec:perfectmodelscenario}, this scenario tests the feasibility of combining model reduction with the filtering algorithm.}

\MODIFY{We ran two different experiments independently of each other, for a total experiment time of $5$ ett:}
\begin{enumerate}
	\item \MODIFY{time interval between assimilations $\Delta=0.008$ ett (every coarse time step), observation error scaling $\lambda=0.6$, initial ensemble set 1, number of weather stations $d_y = 289$;}
	
	\item \MODIFY{time interval between assimilations $\Delta=0.008$ ett, observation error scaling $\lambda=0.6$, initial ensemble set 2, number of weather stations $d_y = 1089$.}
\end{enumerate}
\MODIFY{For both, the assimilation interval choice amounts to a total of $625$ data assimilation steps.}

\MODIFY{It is to be expected that the results would not be comparable to those from the previous subsection. In this scenario the truth is from a different dynamical system (PDE) to the signal process (SPDE), see the discussion around \eqref{eq:filtering_density}. Additional sources of error are thus  introduced into the particle filter algorithm, see \citet{ClarkCrisan}. 
Further, using $289$ weather stations amounts to observing $0.01837\%$ of the truth state space in this scenario. This percentage is improved to $0.06924\%$ when using $1089$ weather stations in the experiment using initial ensemble 2. In either case, we have two orders of magnitude less information about the truth state than the experiments in section \ref{sec:perfectmodelscenario}.}

\MODIFY{Figure \ref{figure: pde rmse and spread initial ensemble set 1} corresponds to the experiment using initial ensemble 1 and $289$ observations. Figure \ref{figure: pde rmse and spread initial ensemble set 2} corresponds to the experiment using initial ensemble 2 and $1089$ observations. The left subplots in figures \ref{figure: pde rmse and spread initial ensemble set 1} and \ref{figure: pde rmse and spread initial ensemble set 2} show comparisons of the rmse between the one step forecast ensemble mean and the true state (in green) \[\rmse(\bar{\pp}_{i,1},X^\dagger_i),\] the rmse between the prior ensemble mean and the true state (in orange) \[\rmse(\bar{\qq}_i, X_i^\dagger),\] and lastly the rmse between the one step forecast ensemble mean and the true state plus observation noise (in red) \[\rmse(\bar{\pp}_{i,1},X_i^\dagger + \epsilon_i).\] }

\MODIFY{In both rmse subplots in figures \ref{figure: pde rmse and spread initial ensemble set 1} and \ref{figure: pde rmse and spread initial ensemble set 2}, $\rmse(\bar{\pp}_{i,1},X_i^\dagger + \epsilon_i)$ (in red) is \emph{not stable} but show increasing trends. Thus the size of the observation noise is not the dominating factor (unlike in the perfect model scenario). The reason is most likely due to the difference between the SPDE and the PDE solution submanifolds. Better calibration, and additional data assimilation techniques such as nudging maybe help to improve this, see \citet{Cotter2020}. In figure \ref{figure: pde rmse and spread initial ensemble set 1}, the increase in $\rmse(\bar{\pp}_{i,1},X_i^\dagger + \epsilon_i)$ at assimilation step $625$ is about $1.5$ times its initial value. In figure \ref{figure: pde rmse and spread initial ensemble set 2}, the increase in $\rmse(\bar{\pp}_{i,1},X_i^\dagger + \epsilon_i)$ at assimilation step $625$ is no more than $1.5$ times its value after the first assimilation 
(after the jump). The rmse values are of order $10^{-3}$, therefore are comparable to those shown in the perfect model scenario.}

\MODIFY{In both figures, increasing trends are also observed for the rmse between the one step forecast ensemble mean and the truth without observation noise $\rmse(\bar{\pp}_{i,1},X_i^\dagger)$ (in green).  
We note that the green rmse values are all of order $10^{-3}$ and therefore
remain comparable to those shown in the perfect model scenario.}

\MODIFY{Comparing the evolution of the rmse of the forecast ensemble (in green) with that of the rmse of the prior ensemble (in orange), we see the former has considerably smaller increase. 
Taking this into account as well as the order of the rmse values,
the data we assimilated, albeit very low dimensional relative to the PDE degree of freedom, still provides sufficient information to control the posterior ensemble. 
Therefore we judge that the posterior ensemble mean offers a reasonably accurate approximation of the signal.}

\MODIFY{In the ensemble spread subplots in figures \ref{figure: pde rmse and spread initial ensemble set 1} and \ref{figure: pde rmse and spread initial ensemble set 2}, we see that due to the resampling step at each assimilation time, there is a small difference between the ensemble spreads of the posterior (in blue) and the one step forecast (in green). 
One can also see that the posterior ensemble spreads are reasonably stable. However, in the absence of data assimilation corrections, for the prior distribution (in orange) we see a continuous increase in both the rmse and spread.}

\MODIFY{Figure \ref{figure:pde_statistics} compares the experiment which used initial ensemble 1 and $289$ observations, with the experiment which used initial ensemble 2 and $1089$ observations.
In the figure, the four subplots compare $\rmse(\bar{\ppi}_i, X^\dagger_i)$, $\sprd(\ppi_i)$, number of SPDE computations and ess.}

\MODIFY{As shown in table \ref{tab:initial_ensemble}, the two initial ensembles produce very different initial rmse and ensemble spread. The rmse values of initial ensemble 1's ensemble mean are two orders of magnitude smaller than initial ensemble 2's rmse values. The ensemble spread of initial ensemble 1 is also one order of magnitude smaller than the ensemble spread of initial ensemble 2. Despite these relatively large initial differences, after one data assimilation step is completed, the rmse of the corresponding ensembles become comparable. The reason is that all unlikely particles are immediately eliminated whilst the diversity of the ensemble is kept high through the tempering procedure. This is also visualised in figure \ref{figure: pde rmse and spread initial ensemble set 2}, indicated by the initial sharp decrease in the one step forecast rmse, $\rmse(\bar{\pp}_{i,1},X^\dagger_i)$, and one step forecast ensemble spread, $\sprd(\pp_{i,1})$ (both plotted in green).}

\MODIFY{From the rmse subplot in figure \ref{figure:pde_statistics}
one can also see that, overall, the blue rmse plot show slight improvement on the red rmse plot. 
This is most likely because more data were assimilated. The dimension of the assimilated data also impacts the ensemble spread and computational cost. 
We see from the ensemble spread subplot in figure \ref{figure:pde_statistics} that, increasing the observed data dimension led to smaller posterior ensemble spread. }

\MODIFY{The computational cost shown in the bottom left subfigure of figure \ref{figure:pde_statistics} was computed in the same way as in section \ref{sec:perfectmodelscenario}. We see that, increasing the observed data dimension from $289$ to $1089$, led to around $1.5$ -- $2$ times computational increase. It is also interesting to note that the computation cost show an increasing trend. This perhaps reflects the increase in the rmse of the posterior ensemble.}

\MODIFY{Finally, the ess value in figure \ref{figure:pde_statistics} shows that the tempering procedure is successful in keeping the ess values near the chosen threshold of $80\%$.}

\MODIFY{In figure \ref{figure: pde rmse vs spread initial ensemble set 1}, we compare forecast rmse with forecast ensemble spread, i.e.
\[
\rmse(\bar{\pp}_{i,j}, X_j^\dagger)
\quad\text{ with }\quad
\sprd(\pp_{i,j})\quad j=1,\dots,50, \quad i=230, 375
\]
taking $\pp_{i,0} = \ppi_{i}$. Note that there are 145 assimilation steps between $\ppi_{230}$ and $\ppi_{375}$, which amounts to $1.16$ ett.
In this scenario, the plots show that the forecast rmse and forecast ensemble spread are less comparable (rmse seem to be about 5 times larger than the spread), than in the perfect model scenario, but the quantities remain the same order $O(10^{-3})$. Further, the difference between corresponding rmse and spread is not quite maintained -- 
the distance between the blue solid line and blue dashed line, is a little larger than the distance between the red solid line and the red dashed line.
These features reflect the evolution of the forecast ensemble rmse and spread, shown in figures \ref{figure: pde rmse and spread initial ensemble set 1} and \ref{figure: pde rmse and spread initial ensemble set 2}, as well as the evolution of the posterior ensemble rmse and spread show in figure \ref{figure:pde_statistics}. }

\MODIFY{In figures \ref{figure:pde trajectory and spread initial ensemble 1} and \ref{figure:pde trajectory spread initial ensemble set 2} we show the Eulerian trajectories of the velocity y-component at four spatial locations. Figure \ref{figure:pde trajectory and spread initial ensemble 1} corresponds to the experiment using initial ensemble 1 and $289$ observations. Figure \ref{figure:pde trajectory spread initial ensemble set 2} corresponds to the experiment using initial ensemble 2 and $1089$ observations. In each figure, the two subfigures correspond to the same experiment at the same grid locations.}

\MODIFY{In the subfigures \ref{figure:pde traj initial ensemble 1} and \ref{figure:pde traj initial ensemble 2}, we plot the true state (in red), the true state plus observation noise (in dashed pink), the posterior ensemble mean (in blue) and the prior ensemble mean (in orange). 
Since initial ensemble 1 start very close to the initial truth, we see in subfigure \ref{figure:pde traj initial ensemble 1} that the prior ensemble mean's initial deviation from the truth is small. But the deviation become much more pronounced after assimilation step 100 at grid location $[0.75, 0.5]$. The posterior ensemble mean (in blue) largely stays close to the observed truth (in pink) at all four grid locations. However at grid point [0.75, 0.5], we see the blue line lost track of the truth before recovering. Also at grid point [0.5, 0.75], we see the blue line gradually deviating from the truth after step 400.}

\MODIFY{In subfigure \ref{figure:pde spread initial ensemble 1}, trajectory of individual ensemble members are plotted. It shows how the prior ensemble spread increases as time goes on, but the posterior ensemble spread seems stable. This supports the features shown in figures \ref{figure:pde_statistics}, \ref{figure: pde rmse and spread initial ensemble set 1} and \ref{figure: pde rmse and spread initial ensemble set 2}. In this scenario though, because there is model error, we see in subfigure \ref{figure:pde spread initial ensemble 1}, at grid point [0.75, 0.5], that the truth deviated from the spread of the prior ensemble.}

\MODIFY{Initial ensemble 2 start alot farther from the initial truth compared to initial ensemble 1. In subfigures \ref{figure:pde traj initial ensemble 2} and \ref{figure: pde spread initial ensemble 2}, we see that the prior ensemble mean deviates from the truth more greatly than shown in subfigures \ref{figure:pde traj initial ensemble 1} and \ref{figure:pde spread initial ensemble 1}. For example, at grid point [0.25, 0.5]. But the posterior ensemble mean (in blue) stay reasonably close to the observed truth at all grid points. At grid point [0.75, 0.25], the posterior lost track of the truth at between step 200 and 300, before recovering.
Overall, despite the filter performing less well in this scenario compared to the perfect model scenario, the results provide good evidence to support the observation that the filtering algorithm was able to eliminate the unlikely particle positions whilst maintaining ensemble diversity, to reasonably approximate the truth.}

\MODIFY{Lastly, in figures \ref{figure:pde_rankhistorgram_ux} and \ref{figure:pde_rankhistorgram_ux_msh8} we show rank histogram plots of the Eulerian velocity x-component, at nine grid locations. Figure \ref{figure:pde_rankhistorgram_ux} corresponds to the experiment using initial ensemble 1, $289$ observations, assimilation period $0.04$ ett and observation noise scaling $\lambda=10$.
The plots show some features of skew, e.g. at grid locations [0.5, 0.5], [0.75, 0.75] and [0.5, 0.25].
Figure \ref{figure:pde_rankhistorgram_ux_msh8} corresponds to the same repeated experiment but using a fewer number of weather stations ($81$ observations). We observe that, assimilating less data led to more pronounced features of skew, e.g. at grid locations $[0.75, 0.75]$.
}

\section{Conclusion}\label{sec:conclusion}
In this work we used a particle filter which included three additional procedures (model reduction, tempering and jittering) in a high dimensional data assimilation (DA) case study. We interpreted the task as solving a filtering problem with a continuous time signal via  discrete observations. 

\MODIFY{In the main numerical experiments, which we called ``imperfect model scenario",} the ``truth" was modelled by a highly resolved numerical solution of a damped and forced incompressible 2D Euler equation which had roughly $3.1\times10^6$ degrees of freedom. The data consisted of a time series of \MODIFY{$625$} discrete observations of the fluid velocity measured on a sparse spatial grid which varied in size from \MODIFY{$289$ to $1089$}. The model reduction involved the addition of a stochastic parametrisation of the above equation solved on a coarser grid of about $4.9\times10^4$ degrees of freedom.  

\MODIFY{For our chosen parameters, the numerical results show, the combination of the stochastic model reduction with the tempering based particle filter algorithm produced a posterior ensemble of 100 particles which, despite some model bias, approximated the signal reasonably accurately for a total experiment period of $5$ eddy turnover times ($5000$ fine resolution time steps). All computations were done using modest computational hardware, employing all of the additional procedures (model reduction, tempering and jittering).} 
\MODIFY{We also tested the reliability of the assimilated ensemble system. There, the numerical results show some features of bias and under-dispersion, due to model discrepancies. Nevertheless, our results show the combined algorithm is sufficiently robust.}

In a sequel to this work we aim to incorporate additional procedures (nudging, space-time data assimilation) to the ones discussed here, to further refine the performance of the \MODIFY{combined} particle filter algorithm discussed here. 

\subsubsection*{Code and data availability}
All Python implementation code and experiment data files are available from the corresponding author upon request.

\section*{Acknowledgements}
The authors thank The Engineering and Physical Sciences Research Council (EPSRC) for their support of this work through the grant EP/N023781/1. The authors also thank Nikolas Kantas, Peter Korn, Sebastian Reich, Paul-Marie Grollemund for the many useful, constructive discussions held with them throughout the preparation of this work. 

\bibliographystyle{plainnat}
\bibliography{2DEulerDataAssimilation}

\end{document}